\documentclass[12pt]{article}

\usepackage[T1]{fontenc}
\usepackage[utf8]{inputenc}

\usepackage[left=1in,top=1in,right=1in,bottom=1in,head=.1in]{geometry}

\usepackage[bookmarksopen=true,bookmarks=true,pdfencoding=auto,psdextra,colorlinks,
linkcolor={blue!80!black},
citecolor={green!40!black},
urlcolor={red!50!black}
]{hyperref}
\usepackage{bookmark}
\usepackage[usenames,svgnames,dvipsnames]{xcolor}
\usepackage[font=small,labelfont=bf]{caption}
\usepackage{booktabs}
\usepackage{multicol}
\usepackage{multirow}
\usepackage{bm}

\usepackage[normalem]{ulem}

\usepackage{algorithm}
\usepackage{algorithmicx}
\RequirePackage{algpseudocode}

\usepackage{tikz}
\usetikzlibrary{arrows.meta,arrows}

\usepackage[many]{tcolorbox}  
\definecolor{main}{HTML}{5989cf}    % setting main color to be used
\definecolor{sub}{HTML}{cde4ff}     % setting sub color to be used

\usepackage{amsmath,amssymb,amsthm}
\usepackage{dsfont}
\usepackage{enumerate}
\usepackage{natbib}
\usepackage{graphicx}
\usepackage{subfigure}
\usepackage{soul}
\setstcolor{red}

\newcommand{\beginsupplement}{
    \appendix
    \setcounter{section}{0}
    \renewcommand{\thesection}{S\arabic{section}}
    \setcounter{table}{0}
    \renewcommand{\thetable}{S\arabic{table}}
    \setcounter{figure}{0}
    \renewcommand{\thefigure}{S\arabic{figure}}
    \newcounter{SIfig}
    \renewcommand{\theSIfig}{S\arabic{SIfig}}}

\usepackage{mystyle}
\usepackage{bibunits}
\allowdisplaybreaks

\newcommand{\titletext}{Extrapolated cross-validation for randomized ensembles}

\newcommand{\removelinebreaks}[1]{%
      \def\\{\relax}#1}
\def\titleRLB{\removelinebreaks{\titletext}}

\title{\titletext\bigskip}

\newcommand{\footremember}[2]{%
    \footnote{#2}
    \newcounter{#1}
    \setcounter{#1}{\value{footnote}}%
}
\newcommand{\footrecall}[1]{%
    \footnotemark[\value{#1}]%
}

\author{\normalsize
Jin-Hong Du\footremember{cmustats}{Department of Statistics and Data Science, Carnegie Mellon University, Pittsburgh, PA 15213, USA.}\footremember{cmuml}{Machine Learning Department, Carnegie Mellon University, Pittsburgh, PA 15213, USA.} 
\and \normalsize  Pratik Patil\footremember{berkeleystats}{Department of Statistics, University of California, Berkeley, CA 94720, USA.}
\and {\normalsize Kathryn Roeder\footrecall{cmustats}} 
\and {\normalsize Arun Kumar Kuchibhotla\footrecall{cmustats}}
}

\date{\bigskip\normalsize\today}

\begin{document}

\maketitle

\begin{abstract}    
    Ensemble methods such as bagging and random forests are ubiquitous in various fields, from finance to genomics. 
    Despite their prevalence, the question of the efficient tuning of ensemble parameters has received relatively little attention.
    This paper introduces a cross-validation method, ECV (Extrapolated Cross-Validation), for tuning the ensemble and subsample sizes in randomized ensembles.
    Our method builds on two primary ingredients: initial estimators for small ensemble sizes using out-of-bag errors and a novel risk extrapolation technique that leverages the structure of prediction risk decomposition.
    By establishing uniform consistency of our risk extrapolation technique over ensemble and subsample sizes, we show that ECV yields $\delta$-optimal (with respect to the oracle-tuned risk) ensembles for squared prediction risk.
    Our theory accommodates general predictors, only requires mild moment assumptions, and allows for high-dimensional regimes where the feature dimension grows with the sample size.
    As a practical case study, we employ ECV to predict surface protein abundances from gene expressions in single-cell multiomics using random forests under a computational constraint on the maximum ensemble size.
    Compared to sample-split and $K$-fold cross-validation, ECV achieves higher accuracy by avoiding sample splitting. 
    Meanwhile, its computational cost is considerably lower owing to the use of the risk extrapolation technique.
\end{abstract}

\bigskip
\textbf{Keywords:}
Ensemble learning; Bagging; Random forest; Risk extrapolation; Tuning and model selection; Distributed learning.

\setcounter{tocdepth}{2}
% \setcounter{section}{-1}
% \clearpage
% \tableofcontents

\setcounter{tocdepth}{3}
%%%%%%%%%%%%%%%%%%%%%%%%%%%%%%%%%%%%%%%%%%%%%%%%%%%%%%%%%%%%%%%%%%%%%%%%%% 
\clearpage

\begin{bibunit}[apalike]
\section{Introduction}

% Importance of tuning k and M
Bagging and its variants are popular randomized ensemble methods in statistics and machine learning. These methods combine multiple models, each fitted on different bootstrapped or subsampled datasets, to improve prediction accuracy and stability \citep{breiman_1996,pugh_2002}, which is well-suited for large-scale distributed computation.
The success of these methods lies in the careful tuning of key parameters: the \emph{ensemble size} $M$ and the \emph{subsample size} $k$ \citep{hastie2009elements}.
As $M$ grows, the predictive accuracy improves while prediction variance decreases and stabilizes, a concept known as algorithmic convergence \citep{lopes2019estimating,lopes2020measuring}.
However, in the era of big data \citep{politis2023}, achieving a precise approximation in the infinity ensemble is challenging due to computational costs.
This necessitates the selection of a suitable value of $M$ to strike a balance between data-dependent considerations and budget constraints, without the requirement for it to scale proportionally with the sample size.
Further, the number of subsampled/bootstrapped observations, $k$, used for each predictor plays a crucial role in ensemble learning \citep{martinez2010out}.
In low-dimensional scenarios, a smaller $k$ can yield consistent results \citep{politis1994large,bickel1997resampling}; however, in high-dimensional scenarios, the prediction risk may not have a straightforward monotonic relationship with subsample size, exhibiting instead multiple descent behaviors \citep{hastie2022surprises,chen_min_belkin_karbasi_2020,patil2022mitigating}.
By restricting the subsample size, there has been some theoretical work on random forests showing that consistency and/or asymptotic normality can be achieved under certain regularity conditions \citep{scornet2015consistency,mentch2016quantifying,peng2022rates,wager2018estimation}.
Selecting the right subsample size is thus also of paramount importance for optimal predictive performance.

% variance stabilization for M and CV for k 
Several strategies have been proposed to tune either the ensemble size $M$ or the subsample size $k$.
For instance, to choose an ensemble size $M$, a variance stabilization strategy is proposed by \citet{lopes2019estimating} and \citet{lopes2020measuring}.
This approach relies on the convergence rate of variance or quantile estimators, using these metrics to gauge the point at which the ensemble's performance stabilizes as $M$ approaches infinity.
Such an approach effectively helps reduce computational expenses.
On the other hand, determining the optimal subsample size $k$ is a more difficult task and generally tuned by standard cross-validation (CV) methods.
As one of the most basic of the CV methods, sample-split CV estimates the predictive risk of every predictor associated with a configuration of parameters using independent hold-out observations \citep{patil2022bagging}.
Another commonly used CV method, $K$-fold CV, repeatedly fits each candidate predictor on $K$ different subsets of the data and uses their average to estimate the prediction risk.

% our work and difference to related work.
While these aforementioned methods offer ways to tune $M$ and $k$, they have several drawbacks.
In terms of turning over the ensemble size $M$, the specialized method proposed by \citet{lopes2019estimating} and \citet{lopes2020measuring} involves monitoring the variability of the test errors as a function of $M$.
However, as this approach focuses solely on \emph{the scale of variance of the risk} rather than the \emph{prediction risk} itself, it does not provide any suboptimality guarantee compared to the optimal risk of an infinite ensemble \citep{lejeune2020implicit}.
Furthermore, the method does not provide any estimators for the prediction risks for any finite ensemble size $M$.
In terms of tuning the subsample size $k$, the traditional CV methods such as \splitcv and $K$-fold CV can be significantly impacted by finite-sample effects due to sample splitting, particularly in high-dimensional scenarios \citep{wang2018approximate,rad_maleki_2020,patil2022bagging}.
Furthermore, these generic CV methods must evaluate every possible ensemble and subsample size within an arbitrarily chosen search space, which often requires exploring larger ensemble sizes, thus demanding more computational resources. 
Yet, even with these, certifying any optimality outside this predefined search is not generally possible. These drawbacks highlight the need for more efficient tuning methods for ensemble learning.

\begin{figure}[!t]
    \centering
    \includegraphics[width=\textwidth]{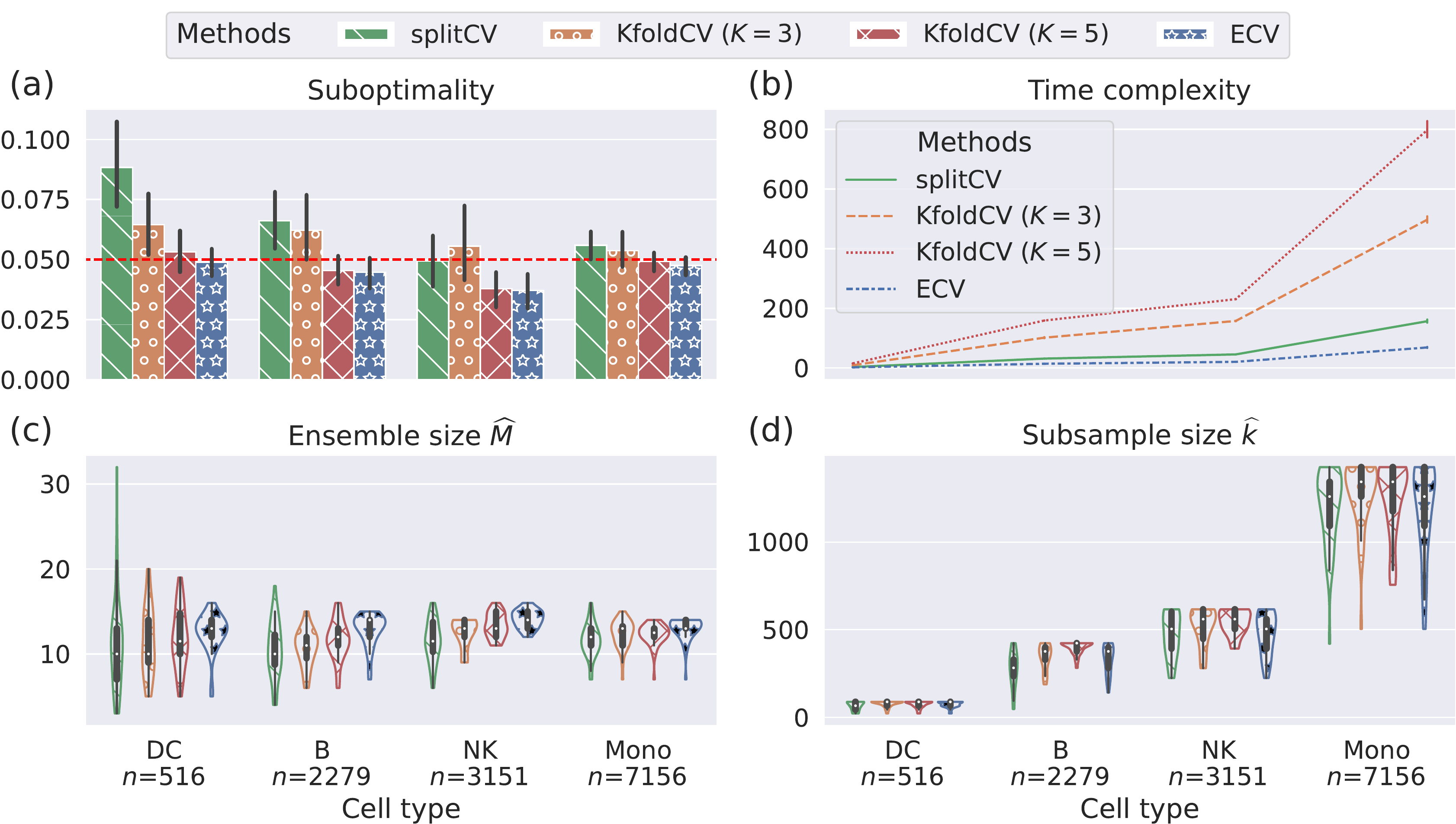}
    \caption{Overview of different cross-validation methods for predicting 50 surface proteins based on 5,000 genes in the single-cell sequencing multiomic datasets from \citet{hao2021integrated} using random forests.
    The proposed \ECV method estimates the prediction risks based on 20 trees. 
    Each panel shows cell types DC, B, NK, and Mono.
    {(a)} The out-of-sample suboptimality compared to the optimal random forest with $50$ trees.
    The CV predictors are tuned to be $\delta$-optimal in terms of normalized mean squared errors.
    The error bars show the standard deviations over 50 proteins, and the red dashed line indicates the optimality threshold $\delta=0.05$.
    {(b)} The time consumption in seconds.
    {(c)} The cross-validated ensemble size $\hat{M}$.
    {(d)} The cross-validated subsample size $\hat{k}$.}
    \label{fig:oobcv_overview}
\end{figure}

This paper seeks to address both the theoretical and practical challenges associated with ensemble parameter tuning.
To this end, we develop a CV method that can efficiently and consistently tune both the ensemble and subsample sizes.
We focus primarily on randomized ensembles such as bagging and subagging \citep{breiman_1996,buhlmann2002analyzing} and rely on out-of-bag observations to estimate the conditional prediction risk.
Our proposed method, termed ECV ({E}xtrapolated {C}ross-{V}alidation), enjoys several advantages over the previously mentioned approaches.
We highlight two of them below:
(1) \emph{Statistical consistency}: ECV is versatile and model-agnostic and is applicable to general ensemble predictors.
It provides uniform consistency in estimating the actual prediction risk of ensembles over all ensemble and subsample sizes under mild conditions.
It also notably outperforms standard CV methods in finite samples, especially in high-dimensional settings.
(2) \emph{Computational efficiency}: 
\ECV operates by estimating the risk of ensembles with small ensemble sizes ($M=1,2$) using out-of-bag observations and then extrapolates the risk estimates to arbitrary ensemble sizes.
Unlike sample-split and $K$-fold CV, our method does not require fitting an ensemble for every ensemble size or explicitly estimating prediction risks for every ensemble size, and can serve as a valuable tool for assessing the fitness of ensemble predictors.
Though our focus in the current paper is on ensemble and subsample size tuning, \ECV is flexible for tuning other hyperparameters efficiently by risk extrapolation.
As a result, ECV significantly reduces the computational burden, making it an efficient method for ensemble parameter tuning.

Before delving into the details of our method, we take the opportunity to demonstrate these key points through a real-world example. We apply \ECV on four single-cell datasets and aim to select a $\delta$-optimal random forest in \Cref{sec:app-sc}, so that its prediction risk is no more than $\delta=0.05$ away from the best random forest consisting of 50 trees.
In \Cref{fig:oobcv_overview}, we compare the performance of the sample-split CV \citep{patil2022bagging}, $K$-fold CV with $K=3,5$, and our method \ECV applied to four datasets, each corresponding to a different cell type obtained from \citep{hao2021integrated}. 
Further details regarding this application can be found in \Cref{sec:app-sc}.
Both of the commonly used CV methods require estimating the risks with all possible choices of ensemble size $M$ and subsample size $k$. 
To ensure a fair comparison, we use the same search space of $(M,k)$ for all methods.
Because sample splitting introduces additional randomness and the reduced sample size has significant finite sample effects, we observe from \Cref{fig:oobcv_overview}(a) that sample-split CV does not control the out-of-sample error within the specified tolerance of $\delta=0.05$ away from the best possible error.
On the other hand, even though $K$-fold CV gives the valid error control as \ECV, it costs extra computational time, which significantly increases as the sample size increases; see \Cref{fig:oobcv_overview}(b).
Overall, the distributions of tuned ensemble parameters $(\hat{M},\hat{k})$ are similar between $K$-fold CV and \ECV; see \Cref{fig:oobcv_overview}(c)-(d).
These results demonstrate the practical effectiveness of \ECV in addressing the above drawbacks for ensemble parameter tuning.

We next provide an overview of the main results and delineate the structure for the remainder of the paper.
In \Cref{sec:bagging}, we outline the setup of the tuning problem in the context of randomized ensemble learning.
We also provide a comprehensive review of prior work on cross-validation and tuning and contrast our method to earlier work.
In \Cref{sec:oobcv}, we lay the foundation for our proposed method by constructing the theoretical ingredients essential to our main algorithm, in particular, the extrapolation risk estimation strategy. 
Through a non-asymptotic analysis, we further demonstrate the convergence rate and uniform consistency of the proposed risk estimator for general predictors and data distributions under a mild moment condition (see \Cref{lem:consistency-oobcv}).
        
In \Cref{subsec:risk-extrapolation}, we introduce our main algorithm and discuss its theoretical properties and various practical considerations.
We prove that the resulting tuned ensemble obtains the best possible ensemble risk over all ensemble and subsample sizes up to a specified tolerance of $\delta$ (see \Cref{thm:limiting-oobcv-for-arbitrary-M-cond}).
In \Cref{sec:simulation}, we examine \ECV's generality and effectiveness with various types of predictors. 
In \Cref{sec:app-sc}, comparisons with sample-split and $K$-fold CV in the protein prediction problem highlight the statistical and computational benefits of \ECV on low- and high-dimensional datasets, under a computational budget for the maximum ensemble size.
In \Cref{sec:conclusion}, we conclude the paper with a brief discussion that acknowledges some limitations of the method and provides avenues for future work.
The code to replicate all our experiments can be obtained from \url{https://jaydu1.github.io/overparameterized-ensembling/ecv} and the Python package implementing the ECV method can be found on the GitHub repository \url{https://github.com/jaydu1/ensemble-cross-validation}.

\section{Randomized ensembles}\label{sec:bagging}

    We consider a supervised regression setup.
    Suppose $\cD_n = \{(\bx_1, y_1), \ldots, (\bx_n, y_n)\}$ represents a dataset with independent and identically distributed random vectors from $\RR^p\times\mathbb{R}$.
    We will not assume any specific data model, only that the second moment of the response is finite, i.e., $\EE(y_1^2)<\infty$.
    A prediction procedure $\hf(\cdot; \cdot)$ is defined as a map from $\RR^p\times \mathscr{P}(\cD_n) \to \RR$, where $\sP(A)$, for any set $A$, represents the power set of $A$. 

    % \subsection{Preliminaries}
    Bagging (\smash{as in {b}ootstrap-{agg}regat{ing}}) traditionally refers to computing predictors multiple times based on bootstrapped data~\citep{breiman_1996}, which can involve repeated observations. 
    There is another version of bagging called subagging (\smash{as in {sub}sample-{agg}regat{ing}}) in \citet[Section 3.2]{buhlmann2002analyzing} where we sample observations without replacement.
    Our method and analysis apply to both of these sampling strategies. Formally, these can be understood as simple random samples from a finite set, commonly used in survey sampling.
    Fix any $k\in[n]$, we define the indices $\{I_{\ell}\}_{\ell=1}^M$ to be $M$ independent samples with replacement from $\cI_k$ (denoted by $\{I_{\ell}\}_{\ell=1}^M\overset{\SRS}{\sim}\cI_k$). Here $\cI_k$ is defined for bagging and subagging as
    \begin{align}
        \cI_k= 
        \begin{cases}
            \{\{i_1, i_2, \ldots, i_k\}:\, 1\le i_1 \le i_2 \le \cdots \le i_k \le n
            \} ,&\text{(bagging)}\\
            \{\{i_1, i_2, \ldots, i_k\}:\, 1\le i_1 < i_2 < \cdots < i_k \le n\}.&\text{(subagging)}\\
        \end{cases}\label{eq:I_M}
    \end{align}
    For bagging, $\cI_k$ represents the set of all possible independent draws from $[n]$ with replacement, and there are $n^k$ many of them.
    For subagging, $\cI_k$ represents the set of all $k$ subset choices from $[n]$, and there are $n!/(k!(n-k)!)$ many of them.
    Throughout the paper, we mainly focus on subagging but the results apply equally well to bagging. 
    For this reason, we do not distinguish different choices of $\cI_k$.
    For any $I\in\cI_k$, let $\cD_I$ and the corresponding subsampled predictor be defined as $\mathcal{D}_{I} = \{(\bx_j, y_j):\, j \in I\}$ and $\hf(\bx; \mathcal{D}_I) = \hf(\bx; \{(\bx_j, y_j):\, j\in I\})$.
    Then the randomized ensemble using either bootstrap or subsampling is defined as follows:
    \begin{align}
        \widetilde{f}_{M,k}(\bx; \{\mathcal{D}_{I_{\ell}}\}_{\ell = 1}^M) &=  \frac{1}{M}\sum_{\ell = 1}^M \hf(\bx; \mathcal{D}_{I_{\ell}})\quad\mbox{with}\quad I_{1}, \ldots, I_M\overset{\SRS}{\sim} \cI_k.\label{eq:bagged-predictor-bagging}
    \end{align}
    In the context where we want to highlight the size of bootstrap/subsample, $k$, we write $I_{k,1}, \ldots, I_{k,M}$ instead of $I_1, \ldots, I_M$.

    % \subsection{Prediction risk.}
    We are interested in the performance of our predictors (computed on $\cD_n$) on future data from the same distribution $P$.
    We consider the behavior of the predictors conditional on $\cD_n$ and $\{I_{\ell}\}_{\ell=1}^M$.
    More specifically, for a predictor $\widehat{f}$ fitted on $\cD_n$ and its bagged predictor $\tf_{M,k}$ fitted on $\{\mathcal{D}_{I_{\ell}}\}_{\ell=1}^M$, with $\{I_{\ell}\}_{\ell=1}^M\overset{\SRS}{\sim}\cI_k$, the data and subsample conditioned risks are defined as:    
    \begin{align}
        R(\widehat{f}; \mathcal{D}_n) 
        &= 
        \int
        (y-\widehat{f}(\bx; \cD_n))^2
        \, \mathrm{d}P(\bx, y), \notag\\        
        R(\tf_{M,k} ; \, \cD_n,
        \{I_{\ell}\}_{\ell = 1}^{M}) 
        &= 
        \int
        \left(
        y-\tf_{M, k}(\bx; \{ \cD_{I_\ell} \}_{\ell = 1}^{M})
        \right)^2 
        \, \mathrm{d}P(\bx, y). \label{eq:conditional-risk}
    \end{align}
    % tuning
    The two critical quantities for ensemble learning are the ensemble size $M$ and the subsample size $k$.
    Different values of $M$ trade-off model stability and computational burden.
    As $M$ increases, the bagged predictors get more stabilized while requiring more time to fit them.
    On the other hand, the subsample size trades off the bias and variance of the bagged predictors.
    A smaller subsample size has a considerable bias but may reduce the variance.
    For example, in the context of subagging minimum norm least square predictors with $M=\infty$, a properly chosen subsample size $k$ strictly less than $n$ can have a higher variance reduction compared to the inflation of bias.
    This raises the question: how to efficiently choose both the ensemble size ($M$) and the subsample size ($k$) to minimize prediction risk \eqref{eq:conditional-risk} for general predictors. 
    We address this question in the next section.
    Before that, we review some related work on cross-validation and situate our work in the context of other related work.

    There is extensive literature on cross-validation (CV) approaches; see \Cref{sec:related-work-general-cv} for a detailed survey.
    In the context of bagging and subagging, 
    \citet{liu2019reducing} study parameters selection for bagging in sparse regression based on the derived error bound.
    For subsample size tuning, a sample-split CV method is proposed in \citep{patil2022mitigating,patil2022bagging}.
    To estimate the prediction risk without sample splitting, the other line of research uses the out-of-bag (OOB) observations \citep{breiman2001random}.
    For example, the algorithmic variance of ensemble regression functions at a fixed test point is studied in \citep{oshiro2012many,wager2014confidence};
    in \cite{lopes2019estimating,lopes2020measuring}, the authors extrapolate the algorithmic variance and quantile of random forests for classification and regression problems, respectively.
    Their extrapolated estimators based on the heuristic scaling improve computation empirically, but theoretically, the statistical property is still unclear.
    \citet{politis2023} discuss scalable subbagging estimator when the subsample size $k$ and the ensemble size $M$ scale with the sample size $n$.
    
    The current paper differs from the previously mentioned works in two significant aspects.
    First, the consistency of \splitcv is shown in \citet{patil2022bagging} for a fixed ensemble size $M$ for subagging, which suffers from the finite-sample effects because of sample splitting. 
    Additionally, their results rely on stringent assumptions that require the asymptotic risk to satisfy certain analytic properties and do not provide any convergence rates.
    In contrast, our work establishes uniform consistency over both the ensemble size $M$ and the subsample size $k$ for bagging as well as subagging.
    More importantly, we also characterize the proposed estimators' convergence rate and require much milder assumptions on the risk, in the form of certain moment conditions.
    Second, \citet{lopes2019estimating} and \citet{lopes2020measuring} rely on the convergence rate of variance to extrapolate the fluctuations of the estimates and require the ensemble size $M$ to approach infinity. In contrast, \ECV directly estimates the extrapolated risk (not just the scale of variances or quantiles) for an arbitrary range of ensemble sizes in a consistent manner.
    Additionally, while their papers only focus on tuning the ensemble size $M$, we also tune the subsample size $k$, which is crucial to minimizing the predictive risks, especially in high-dimensional scenarios, as alluded to in the introduction.

\section{Method motivation}\label{sec:oobcv}

In this section, we derive preliminary results that serve as the foundation for our extrapolated cross-validation method in \Cref{subsec:risk-extrapolation}.
Our approach relies on utilizing out-of-bag (OOB) observations to tune both the ensemble size $M$ and the subsample size $k$.
Let us now describe the main components behind our methodology.

\subsection{Decomposition and risk estimation}

To begin with, we will fix $k$ and focus on tuning over $M\in\NN$.
Recall the conditional risk $R(\tf_{M,k} ; \, \cD_n,\{I_{\ell}\}_{\ell = 1}^{M})$ associated with an $M$-bagged predictor $\tf_{M,k}$, as defined in \eqref{eq:conditional-risk}.
The subsequent proposition demonstrates that $R(\tf_{M,k}; \, \cD_n,\{I_{\ell}\}_{\ell = 1}^{M})$ can be expressed as a linear combination of the conditional risks for $M=1$ and $M=2$.

    \begin{proposition}[Squared conditional risk decomposition]\label{prop:squared_risk_decom}
        The conditional prediction
        risk defined in \eqref{eq:conditional-risk}
        for a bagged predictor $\tf_{M,k}$ decomposes into
        \begin{align}
            R(\tf_{M,k} ; \, \cD_n,\{I_{\ell}\}_{\ell = 1}^{M}) &= -\left(1-\frac{2}{M}\right)a_{1,M} + 2\left(1 - \frac{1}{M}\right) a_{2,M},\label{eq:risk-decomp-M}\\
            \text{where }\qquad a_{1,M} = \frac{1}{M} \sum_{\ell=1}^{M} R(\tf_{1,k}; &\cD_n, \{ I_\ell \}),~ 
            a_{2,M} = \frac{1}{M(M-1)} \sum\limits_{\substack{\ell,m\in[M]\\\ell\neq m}}  R(\tf_{2,k} ; \, \cD_n, \{I_{\ell},I_m\}). \qquad \notag
        \end{align}
    \end{proposition}

    The proof of \Cref{prop:squared_risk_decom} can be found in \Cref{app:prop:squared_risk_decom}.
    The statement follows due to a special decomposition that governs the squared risk of the $M$-bagged predictor.
    The components $a_{1,M}$ and $a_{2,M}$ in the decomposition \eqref{eq:risk-decomp-M} are the averages of the 1-bagged and 2-bagged conditional risks, respectively.    
    Conditional on $\cD_n$, note that $a_{1,M}$ and $a_{2,M}$ are $U$-statistic based on i.i.d.\ elements $\{I_{\ell}\}_{\ell=1}^M$ sampled from $\cI_k$.
    
    Towards motivating our \ECV method, let us assume that there exist constants $\mathfrak{R}_{1,k}$ and $\mathfrak{R}_{2,k}$ such that as $n\rightarrow\infty$, both $a_{1,M} - \mathfrak{R}_{1,k}$ and $a_{2,M} - \mathfrak{R}_{2,k}$ converge almost surely to $0$.
    Then \Cref{prop:squared_risk_decom} implies that the conditional prediction risk of $\tf_{M,k}$ can be approximated asymptotically 
    by $-(1-2/M)\mathfrak{R}_{1,k} + 2(1-1/M)\mathfrak{R}_{2,k}$.
    This serves as the basis for the concept of \ECV, where consistent estimation of $\mathfrak{R}_{1,k}$ and $\mathfrak{R}_{2,k}$ allows for a consistent estimator of the risk of $\tf_{M,k}$ for every $M\in\NN$, hence justifying the name ``extrapolated'' cross-validation.
       
    To consistently estimate the basic components $\mathfrak{R}_{1,k}$ and $\mathfrak{R}_{2,k}$ in \eqref{eq:risk-decomp-M}, we first leverage the OOB risk estimator for an arbitrary predictor $\hf$ fitted on $\cD_I$ based on the OOB test dataset $\cD_{I^c}=\cD_n\setminus\cD_I$.
    One can simply consider the average of the squared loss (\AVG) on OOB observations in $\cD_{I^c}$ as a risk estimator.
    This choice, however, is not suitable for heavy-tailed data and hence we also consider a median-of-means (\MOM) estimator:
    \begin{align}\label{eq:h}
        \hR(\hf , \cD_{I^c}) &= \begin{cases}
            \frac{1}{|I^c|}\sum_{i\in I^c}(y_i-\hf(\bx_i))^2 ,& \text{ if \CEN=\xspace\AVG},\\
            \mathop{\mathrm{median}}\left\{\frac{1}{|I^{(b)}|}\sum_{i\in I^{(b)}}(y_i-\hf(\bx_i))^2,\ b\in[B]\right\},& \text{ if \CEN=\xspace\MOM},
        \end{cases}
    \end{align}
    with $B = \lceil 8 \log(1/\eta) \rceil$ and $I^{(1)},\ldots,I^{(B)}$ being $B= \lceil 8 \log(1/\eta) \rceil$ random splits of $I^c$ for some $\eta>0$.
    The median-of-means estimator was developed for heavy-tailed mean estimation and is commonly used in robust statistics \citep{lugosi2019mean}.
    
    We will provide a condition to certify the pointwise consistency of risk estimates $\hR$ under certain assumptions on the data distribution.
    Towards that end, for any non-negative loss function $\cL$ and a given test observation $(\bx_0,y_0)$ from $P$, define the conditional $\psi_1$-Orlicz norm of $\cL(y_0,\hf(\bx_0))$ given $\cD_I$ as
    $\| \cL(y_0,\hf(\bx_0)) \|_{\psi_1 \mid \cD_I} = \inf \big\{ C > 0: \, \mathbb{E} \big[
                \exp\big( \cL(y_0,\hf(\bx_0)) C^{-1} \big) \mid \mathcal{D}_I
            \big] 
            \le 2
        \big\}.$
    Similarly, for $r \ge 1$, define the conditional $L_r$-norm as
        $\|\cL(y_0,\hf(\bx_0))\|_{L_r \mid \cD_I}
        =
        \big(
             \mathbb{E}[\cL(y_0,\hf(\bx_0))^{r} \mid \cD_I]\big)^{1/r}.$
    See \citet[Chapter 2]{vershynin_2018} for more details.
    The following proposition provides the condition for consistency.
    \begin{proposition}[Consistent risk estimators]
        \label{prop:bounded-variance-error-control-mul-form}
        Let $\hf(\cdot;\cD_I)$ be any predictor trained on $\cD_I\subset\cD_n$,
        and $\CEN=\AVG$ or $\CEN = \MOM$ with $\eta = n^{-A}$, where $A\in(0,\infty)$ is a fixed constant.
        Define $\hsigma_I=\|(y_0- \hf(\bx_0;\cD_I))^2 \|_{\psi_1 \mid \cD_I}$.
        If $\hsigma_I/\sqrt{|I^c|/\log n}\rightarrow 0$ in probability, then $|\hR(\hf , \cD_{I^c}) - R(\hf;\cD_I) | \rightarrow 0$ in probability.  
        The conclusion remains true for $\CEN = \MOM$ even when the conditional $\psi_1$-Orlicz norm in the definition of $\hsigma_I$ is replaced with $\|\cdot\|_{L_2 \mid \cD_I}$.
    \end{proposition}

    The proof of \Cref{prop:bounded-variance-error-control-mul-form} can be found in \Cref{app:prop:bounded-variance-error-control-mul-form}. It follows from the results in \citet[Lemma 2.9, Lemma 2.10]{patil2022bagging}, which are used to demonstrate the strong consistency of \splitcv in \citep{patil2022mitigating,patil2022bagging} for subsample tuning.
    However, our analysis will utilize the finite-sample tail bounds that underlie \Cref{prop:bounded-variance-error-control-mul-form} to obtain convergence~rates.
    
    Recall from~\eqref{eq:bagged-predictor-bagging} that $\tf_{1,k}$ is computed from one dataset $\cD_{I_1}$ and $\tf_{2,k}$ is computed from two datasets $\cD_{I_1}$ and $\cD_{I_2}$ or equivalently that $\tf_{2,k}$ is computed on $\cD_{I_1\cup I_2}$. Hence, \Cref{prop:bounded-variance-error-control-mul-form} can be applied to $\tf_{1,k}$ with $I = I_1$ and to $\tf_{2,k}$ with $I = I_1\cup I_2$ to consistently estimate the conditional prediction risks of $\tf_{1,k}$ and $\tf_{2,k}$. To obtain consistency, we need to ensure that the assumption on $\hsigma_I$ becomes reasonable if $|I^c|/\log n \to \infty$ as $n\to\infty.$
    For subagging predictors $\tf_{1,k}$ and $\tf_{2,k}$, we have $|I_1^c| = n(1-k/n)$ and by~\Cref{lem:i0_mean}, $|(I_1\cup I_2)^c|\approx n(1 - k/n)^2$, asymptotically. On the other hand, for bagging predictors $\tf_{1,k}$ and $\tf_{2,k}$, we have $|I_1^c| \approx n\exp(-k/n)$ and $|(I_1\cup I_2)^c| \approx n\exp(-2k/n)$, asymptotically. Hence, collectively, assuming $k \le n(1 - 1/\log n)$ implies that $|I^c| \gtrsim n/\log^2n$, asymptotically.
    
    Under the assumption $k \le n(1 - 1/\log n)$, the risks of $\tf_{1,k}$ and $\tf_{2,k}$ computed on $\cD_{I_1}$ and $\cD_{I_1\cup I_2}$, respectively, can be estimated consistently. 
    Further, if we assume that the limiting conditional risks $(\mathfrak{R}_{1,k}$ and $\mathfrak{R}_{2,k}$) of $\tf_{1,k}$ and $\tf_{2,k}$ do not depend on the particular subsets $I_1, I_2$, then $\widehat{R}(\tf_{1,k}, \cD_{I_1^c})$ and $\widehat{R}(\tf_{2,k}, \cD_{(I_1\cup I_2)^c})$ consistently estimate $\mathfrak{R}_{1,k}$ and $\mathfrak{R}_{2,k}$, as $n\to\infty$.
    Note, however, that because the limiting conditional risks do not depend on specific subsets $I_1, I_2$, we can reduce the variance in our estimates of $\mathfrak{R}_{1,k}$ and $\mathfrak{R}_{2,k}$ by averaging the estimated risks over several subsets $I_{\ell}$'s. 
    This observation suggests the out-of-bag risk estimates for $M=1,2$ as
    \begin{align}\label{eq:Roob-M-12}
        \hRoob_{M,k} &= 
        \begin{cases}
            \frac{1}{M_0}\sum\limits_{\ell=1}^{M_0} \hR(\tf_{1,k}(\cdot;\cD_n,\{I_{\ell}\}), \cD_{I_{\ell}^c}),&M=1,\\
            \frac{1}{M_0(M_0-1)}\sum\limits_{\substack{\ell,m\in[M_0]\\\ell\neq m}} \hR(\tf_{2,k}(\cdot;\cD_n,\{I_{\ell},I_{m}\}), \cD_{(I_{\ell}\cup I_{m})^c}),&M=2,
        \end{cases}
    \end{align}
    where $I_1, \ldots, I_{M_0}$ are i.i.d.\ samples from $\cI_k$ and $M_0\geq 2$ is a pre-specified natural number. Increasing $M_0$ improves estimates but also increases computation.

    % \subsection{(3) Risk extrapolation.}
    As hinted above, the risk decomposition \eqref{eq:risk-decomp-M}, along with the component risk estimation \eqref{eq:Roob-M-12}, suggests a natural estimator for the $M$-bagged risk:
    \begin{align} \label{eq:Roob-M}
        \hRoob_{M,k}=-\left(1-\frac{2}{M}\right) \hRoob_{1,k} + 2\left(1 - \frac{1}{M}\right)\hRoob_{2,k},\quad M>2.
    \end{align}
    We call this ``extrapolated'' risk estimation according to \eqref{eq:risk-decomp-M} because the $M$-bagged risk is \emph{extrapolated} from the $1$- and $2$-bagged risks.
    If the prediction risk of $M=1,2$ can be consistently estimated, the extrapolated estimates \eqref{eq:Roob-M} are also pointwise consistent over~$M\in\NN$ for $k$ fixed.

\subsection{Uniform risk consistency}
    The next step is then to tune both $M$ and $k$.
    To tune $k$, we define $\cK_n\subset[n]$ to be a grid of subsample sizes. 
    In practice, we would like $\cK_n$ to cover the full range of $n$ asymptotically (in the sense that $\cK_n/n \approx [0, 1]$), and one simple choice is to set $\cK_n=\{0, k_0,2k_0,\ldots. \lfloor n (1 - (\log n)^{-1}) / k_0\rfloor k_0\}$ where the minimum subsample size is $k_0=\lfloor n^{\nu}\rfloor$ for some $\nu\in(0,1)$.
    Here we adopt the convention that when $k=0$, the ensemble predictor reduces to the \emph{null predictor} that always outputs zero.
    To facilitate our theoretical results for general predictors, we make the following two assumptions. 
    The results are stated asymptotically as $n$ tends to infinity, where we view both $k$ and $p$ as sequences $\{k_n\}$ and $\{p_n\}$ indexed by $n$, and assume $k$ diverges with the sample size $n$ (except when $k\equiv 0$), but the feature dimension $p$ may or may not diverge.
    
    \begin{assumption}[Variance proxy]\label{ass:kappa}
         For $k\in\cK_n$, assume for all $\{I_1,I_2\}\overset{\SRS}{\sim}\cI_k$, as $n\to\infty$,
         \[
         \frac{\log n}{\sqrt{n}}\widehat{\sigma}_{I_1} \to 0,\quad\mbox{and}\quad \frac{(\log n)^{3/2}}{\sqrt{n}}\widehat{\sigma}_{I_1\cup I_2} \to 0,
         \]
         in probability, where the variance proxies $\widehat{\sigma}_{I_1}$ and $\widehat{\sigma}_{I_1\cup I_2}$ are defined in \Cref{prop:bounded-variance-error-control-mul-form}.
    \end{assumption}
    \begin{assumption}[Convergence of asymptotic risks]\label{cond:conv-risk-M12}
        For $k\in\cK_n$, assume for all $\{I_1,I_2\}\overset{\SRS}{\sim}\cI_k$ and for $M=1,2$, there exist constants $\epsilon\in(0,1)$, $C_0>0$, $\eta_0\geq 1$, $\gamma_{M,n}=o(n^{-\epsilon})$, and $\sR_{M,k}\geq 0$, such that the following holds:
        \begin{align}
            \limsup_{n\rightarrow\infty} \sup_{\eta\geq\eta_0}\eta^{1/\epsilon}\PP(\gamma_{M,n}^{-1}| R_{M,k}(\tf_{M,k}; \, \cD_n, \{I_{\ell}\}_{\ell = 1}^{M}) - \sR_{M,k} |
            \leq \eta) \leq C_0.\label{eq:subsample_cond_risk_M12}
        \end{align}
    \end{assumption}
    \Cref{ass:kappa} is used to show consistent risk estimation with $M = 1, 2$ in~\Cref{app:prop:bounded-variance-error-control-mul-form}. \Cref{cond:conv-risk-M12} formalizes the assumption that the limiting values of the conditional risks $R_{M,k}(\tf_{M,k}; \cD_n, \{I_{\ell}\}_{\ell = 1}^M)$ do not depend on $\{I_{\ell}\}_{\ell = 1}^M$. This assumption also requires certain rate and tail assumptions, which are used to provide the rate of consistency of our ECV procedure.
    In \Cref{cond:conv-risk-M12}, $\gamma_{M,n}$ for $M=1,2$ represent the lower bounds of the rates of convergence over $k\in\cK_n$, which are typically on a scale of $n^{-\alpha}$ for some constant $\alpha>0$.
    Condition \eqref{eq:subsample_cond_risk_M12} is also known as the weak moment norm condition \citep{rio2017constants,guo2019berry}, which ensures that the expected differences between the risks and the limits also converge to zero in certain rates.
    Under classical linear models with fixed-$X$ design and Gaussian noises, the risk of linear predictor concentrates around its mean and satisfies \eqref{eq:subsample_cond_risk_M12}; see, e.g., \citet[Lemma 3.1]{bellec2018optimal}.
    Another sufficient condition for \eqref{eq:subsample_cond_risk_M12} is the strong moment condition that $\EE(\gamma_{M,n}^{-1/\epsilon}| R_{M,k}(\tf_{M,k}; \, \cD_n, \{I_{\ell}\}_{\ell = 1}^{M}) - \sR_{M,k} |^{1/\epsilon})$ is bounded.

    From now on, we shall write $R_{M,k} = R(\tf_{M,k}; \, \cD_n, \{I_{\ell}\}_{\ell = 1}^{M}) $ to indicate the dependency only on $M$ and $k$ and to simplify the notations.
    In \Cref{app:ridge}, we present an example of ridge regression where the convergence is under the proportional asymptotics (i.e., both the \emph{data aspect ratio} $p/n$ and the \emph{subsample aspect ratio} $p/k$ converge to fixed constants), and Assumptions~\ref{ass:kappa}-\ref{cond:conv-risk-M12} are satisfied.
    The following theorem guarantees uniform consistency over both $M$ and $k$.
    \begin{theorem}[Uniform consistency of risk extrapolation]\label{lem:consistency-oobcv}
        Suppose \Cref{ass:kappa,cond:conv-risk-M12} hold for all $k\in\cK_n$, then \ECV estimates defined in \eqref{eq:Roob-M} satisfy that 
        \begin{align*}
            \sup_{M\in\NN,k\in\cK_n}\left|\hRoob_{M,k} - R_{M,k}\right| =\cO_p(\zeta_n),
        \end{align*}
        where $\zeta_n=\hsigma_n \log n /n^{1/2}+ n^{\epsilon}(\gamma_{1,n}+ \gamma_{2,n})$, and $\hsigma_n= \max_{m,\ell\in[M_0],k\in\cK_n}\hsigma_{I_{k,\ell}\cup I_{k,m}}$.
    \end{theorem}

    The result in \Cref{lem:consistency-oobcv} is of paramount importance to establish the convergence rate of the CV-tuned estimator returned by our algorithm.
    In words, the theorem says that the extrapolated error depends on three factors: the cross-validation error and the rates for $M=1,2$.
    While the convergence rates of asymptotic risks of $M=1,2$ usually depend on the chosen predictor, the non-asymptotic analysis in \Cref{lem:consistency-oobcv} allows us to derive convergence rates even for general predictors. This is particularly advantageous when compared to the consistency results established in \citep{patil2022bagging} for subagged ridge predictors.
    
    The proof \Cref{lem:consistency-oobcv} is rather involved and can be found in \Cref{app:risk}.
    For the convenience of the readers, we provide a schematic of the whole proof in \Cref{fig:proof}. We will now explain the key ideas involved in the proof. The proof strategy relies on deriving concentration results of varied random quantities to their limits in a specific order.
    First, we establish the uniform consistency of the risk estimates $\hRoob_{M,k}$ over $k\in\cK_n$ to the risks $R_{M,k}$ for $M=1,2$ (see \Cref{prop:consistency-M12}).
    Building upon this result and the risk decomposition presented in \Cref{prop:squared_risk_decom}, we then derive the uniform consistency of the risk estimates $\hRoob_{M,k}$ over $(M,k)\in\NN\times \cK_n$ to the deterministic limits $\sR_{M,k}$ (\Cref{prop:consistency-M}).
    On the other hand, to establish the concentration for subsample conditional risks over $M\in\NN$ and $k\in\cK_n$, we first establish the concentration for the expected conditional risk $\EE[R_{M,k}\mid\cD_n]$ over $k\in\cK_n$ (\Cref{lem:conv:data-risk-M12}). We then apply the reverse martingale concentration bound (\Cref{lem:conv:risk-M}).

\section{Main proposal: Extrapolated cross-validation}\label{subsec:risk-extrapolation}

    Based on the previous discussion in \Cref{sec:oobcv}, we present the proposed cross-validation algorithm for tuning the ensemble parameters without sample splitting in \Cref{alg:cross-validation}.
    The procedure requires a dataset $\cD_n$ of $n$ observations, a base prediction procedure $\hf$, a natural number $M_0$ for risk estimation, and some other parameters.
    It first constructs the grid of subsample sizes $\cK_n$ and fits only $M_0$ base predictors accordingly.
    Then, the prediction risk for $M$-bagged predictors can be estimated based on the OOB observations through \eqref{eq:Roob-M-12} and \eqref{eq:Roob-M}.
    Observe that the optimal risk of \eqref{eq:Roob-M} for any $k$ is obtained at $\hRoob_{\infty,k}=2\hRoob_{2,k} -\hRoob_{1,k} $ when $M=\infty$.
    Thus, to tune $k$, it suffices to perform a grid search to minimize $\hRoob_{\infty,k}$ over $k\in\cK_n$, because the optimal ensemble size happens to be infinity from the previous results \citep{lopes2019estimating,patil2022bagging}.
    However, calculating it is prohibitive in practice.
    Thus, we pick the smallest $\widehat{M}$ such that $\hRoob_{\widehat{M},\hat{k}}$ is close to $\hRoob_{\infty,\hat{k}}$ within $\delta$ error, where $\delta$ is the suboptimality parameter.
    Finally, \Cref{alg:cross-validation} returns a $\hat{M}$-bagged predictor using subsample size $\hat{k}$.
    Note that \Cref{alg:cross-validation} naturally applies to other randomized ensemble methods, such as random forests, when fixing other hyperparameters.

    \begin{algorithm}[!ht]
    \caption{Tuning of ensemble and subsample sizes without sample splitting}\label{alg:cross-validation}
        \begin{algorithmic}[1]
        \Require a dataset $\cD_n = \{ (\bx_i, y_i) \in \RR^{p} \times \RR : 1 \le i \le n \}$, a base prediction procedure $\hf$, 
        a real number $\nu \in (0, 1)$ (subsample size unit parameter),
        a ensemble size $M_0\geq 2$ for risk estimation,
        a centering procedure $\CEN \in \{ \AVG, \MOM \}$,
        a real number $A$ used to compute $\eta$ when $\CEN = \MOM$, 
        and optimality tolerance parameter $\delta$.

        \State\label{algo:2}
        Construct a grid $\cK_n = \{0,k_{0}, 2k_{0},\ldots, \lfloor n(1-1/\log n)/k_0 \rfloor k_0\}$ where $k_{0}=\lfloor n^\nu \rfloor$.
        
        \State Build ensembles $\tf_{M_0, k}(\cdot) = \tf_{M_0}(\cdot; \{ \cD_{I_{k,\ell}} \}_{\ell=1}^{M_0})$ with $M_0$ base predictors, where $I_{k,1},\ldots,I_{k,M_0}\overset{\SRS}{\sim}\cI_k$ for each $k \in \cK_n$.
        
        \State Estimate the conditional prediction risk on OOB observations of $\tf_{M_0,k}$ with $\hRoob_{M,k}$ defined in \eqref{eq:Roob-M-12} for $k \in \cK_n$ and $M=1,2$.

        \State\label{alg:line-risk-extrapolation}
        Extrapolate the risk estimations $\hRoob_{M,k}$ using~\eqref{eq:Roob-M}.

        \State Select a subsample size $\hat{k} \in \cK_n$ that minimizes the extrapolated estimates using
        \begin{align*}
            \hat{k} \in \argmin_{k\in\cK_n}  2\hRoob_{2,k} - \hRoob_{1,k}.
        \end{align*}

        \State Select an ensemble size $\hat{M}\in\NN$ for the $\delta$-optimal risk with a plug-in estimator:
        \begin{align*}
            \hat{M} = \left\lceil\frac{2}{\max\{\delta, n^{-1/2}\}} ( \hRoob_{1,\hat{k}} - \hRoob_{2,\hat{k}} )  \right\rceil.
        \end{align*}
        
        \State If $\hat{M}>M_0$, fit a $\hat{M}$-bagged predictor $\tf_{\hat{M},\hat{k}} = \tf_{\hat{M},\hat{k}}(\cdot;  \{\cD_{I_{\hat{k},\ell}}\}_{\ell=1}^{\hat{M}})$.
        
        \Ensure Return the \ECV-tuned predictor $\tf_{\hat{M},\hat{k}}$, and the risk estimators $\hRoob_{M,k}$ for all~$M,k$.
        \end{algorithmic}
    \end{algorithm}

    As a byproduct, \Cref{alg:cross-validation} also gives the ECV risk estimates $\hRoob_{M,k}$ for all $M\in\NN$ and $k\in\cK_n$.
    This risk profile in $(M,k)$ is helpful for users to investigate whether the given base predictor $\hf$ well fits the dataset $\cD_n$ or not.
    For instance, one can tune the ensemble predictors under a computational budget on the maximum ensemble size $M_{\max}$. 
    This gives rise to practical considerations presented later in \Cref{subsec:practice}.

\subsection{Theoretical guarantees}
\label{subsec:theory}
    By combining the ingredients in \Cref{sec:oobcv}, our main theorem states that~\Cref{alg:cross-validation} yields an ensemble predictor whose risk is at most $\delta$ away from the best ensemble predictor. Further, it provides an estimator of the risk of the selected ensemble predictor. 
    
    \begin{theorem}[Optimality of OOB estimate and \ECV-tuned risk]\label{thm:limiting-oobcv-for-arbitrary-M-cond}
        Under the assumed conditions in \Cref{lem:consistency-oobcv}, for any $\delta> 0$, the OOB estimate and the \ECV-tuned risk output by \Cref{alg:cross-validation} satisfy the following properties respectively:
        \begin{align}\label{eq:risk-minimization-estimate}
            \left|\hRoob_{\hat{M},\hat{k}} - R_{\hat{M},\hat{k}} \right| &= \cO_p(\zeta_n),\qquad 
            \left| R_{\hat{M},\hat{k}} - \inf_{M\in\NN,k\in\cK_n}R_{M,k} \right| = \delta + \cO_p(\zeta_n),
        \end{align}
        where $\zeta_n$ is the quantity defined in~\Cref{lem:consistency-oobcv}.
    \end{theorem}

    \Cref{thm:limiting-oobcv-for-arbitrary-M-cond} implies that the OOB estimate is close to the true risk because of the uniform consistency from \Cref{lem:consistency-oobcv}.
    Furthermore, the \ECV-tuned ensemble parameters $\hat{M}$ and $\hat{k}$ produce a bagged predictor with a risk $\delta$-close to the optimal predictor.
    The optimality is model-agnostic because it does not directly depend on the feature and response models.
    On the other hand, in real-world applications, the infinite ensemble may not be of interest because of the computational cost.
    Because of the uniform consistency established in \Cref{lem:consistency-oobcv}, one can naturally extend \Cref{thm:limiting-oobcv-for-arbitrary-M-cond} to tuning with restriction on maximum ensemble sizes.
    In \Cref{app:ridge}, we also provide a concrete example of the application of \Cref{thm:limiting-oobcv-for-arbitrary-M-cond} to ridge predictors, which verifies \Cref{ass:kappa,cond:conv-risk-M12} under mild moment assumptions.

    \begin{remark}[From additive to multiplicative optimality]\label{rm:mul-opt}
        \Cref{thm:limiting-oobcv-for-arbitrary-M-cond} provides a guarantee of additive optimality for tuned predictor returned by \Cref{alg:cross-validation}, while it is also useful to consider the multiplicative optimality.
        Towards that end, with the choice of $\widehat{k}$ in Step 5 of~\Cref{alg:cross-validation} and change the choice of $\widehat{M}$ in Step 6 to
        \begin{align*}
            \hat{M} = \left\lceil\frac{2}{\max\{\delta, n^{-1/2}\}} 
            \frac{ \hRoob_{1,\hat{k}} - \hRoob_{2,\hat{k}} }{2\hRoob_{2,\hat{k}} - \hRoob_{1,\hat{k}} }\right\rceil,
        \end{align*}
        then, under the assumption that the irreducible risk $\int \left(y - \EE(y \mid \bx)\right)^2 \, \mathrm{d}P(\bx, y)$ is strictly positive, \Cref{prop:mul-opt} guarantees
        \begin{align}\label{eq:risk-minimization-estimate-multiplicative}
            R_{\hat{M},\hat{k}}  = (1+\delta )\inf_{M\in\NN,k\in\cK_n}R_{M,k}(1 + \cO_p(\zeta_n)).
        \end{align}
        Compared to \eqref{eq:risk-minimization-estimate}, the optimality upper bound on the right-hand side of
        \eqref{eq:risk-minimization-estimate-multiplicative} depends on the scale of the optimal prediction risk.
    \end{remark}

    \subsection{Computational considerations}\label{subsec:practice}
    \Cref{alg:cross-validation} estimates the ensemble parameters $\hat{k}$ and $\hat{M}$ to derive a $\delta$-optimal bagged predictor.
    Here, we compare the computational complexity of \Cref{alg:cross-validation} with other common CV methods.
    For each $k\in\cK_n$, suppose the computational complexity of fitting one base predictor on $k$ subsampled observations and obtain their predicted values on $n-k$ OOB observations is $\cO(C_n)$ (ignoring $k$).
    Then, the computational complexity of estimating $\hat{M}$ and $\hat{k}$ for all three validation methods are given below.
    \begin{itemize}[labelsep=1mm,leftmargin=7mm]
        \item \ECV:
        For each $k\in\cK_n$, we need to fit $M_0$ base predictors.
        Then we estimate $R_{1,k}$ and $R_{2,k}$ in $\cO(M_0(n-k))$ and $\cO(M_0^2(n-2k+i_0))$, respectively, where $i_0=k^2/n$ is the intersect observations between two indices of a simple random sample.
        The computational complexity of risk extrapolation is negligible compared to the above time consumption.
        All in all, it takes $\cO(C_n(|\cK_n|M_0+M_{\max}))$ to obtain tuned bagging parameter by \ECV.

        \item \splitcv:
        Suppose the ratio of training data is $\alpha\in(0,1)$. 
        Similar to \ECV, each base predictor is fitted and evaluated on $\lceil n\alpha\rceil$ observations and we need to fit $M_{\max}$ base predictors.
        We then compute the moving average of the predicted values for $M$ varying from 1 to $M_{\max}$, which gives the predicted values of the $M$-bagged predictors, which takes $\cO(M_{\max})$ operations.
        We note that one can alternatively fit one bagged predictor for each $k$ and each $M$; however, this will cause much more time consumption compared to the simple matrix computation operations we used above.
        All in all, it takes $\cO(|\cK_n|M_{\max}(C_{n\alpha}+n) )$ to obtain the tuned parameter.
        
        \item \kfoldcv: We follow the same strategy for fitting base predictors so that \kfoldcv has roughly $K$ times of complexity as \splitcv. Specifically, it takes $\cO(K|\cK_n|M_{\max}(C_{n/K}+n) )$ to obtain the tuned parameter.
    \end{itemize}    
    In general, we expect $C_n$ to grow much faster than $n$, because fitting one base predictor may involve matrix multiplication operation, which takes $\cO(n^2)$.
    Therefore, the computational complexity of the three methods has the ordering: \ECV $ \leq$ \splitcv $ \leq$ \kfoldcv, provided that $M_0$ is much smaller than~$M_{\max}$.

    Besides the computational efficiency gained by ECV, we also discuss some considerations when the proposed method is used in practice.
    
    \begin{enumerate}[labelsep=1mm,noitemsep,leftmargin=7mm,]
        \item{Maximum ensemble size:} 
        \Cref{alg:cross-validation} determines $\hat{k}$ and $\hat{M}$ by minimizing the estimated risk \eqref{eq:Roob-M} with the infinite ensemble.
        However, it may still be computationally infeasible if $\hat{M}$ is too large.
        In such cases, based on the extrapolated risk estimation in \Cref{alg:cross-validation}, we can also derive the $\delta$-optimal bagged predictor whose ensemble size is no more than a pre-specified number $M_{\max}$, which we call the restricted oracle.
        That is, we choose subsample and ensemble size to restrict the computational cost: $\hat{k} \in \argmin_{k\in\cK_n} \hRoob_{M_{\max},k}$ and $\hat{M} = \left\lceil 2 (\delta + \hRoob_{M_{\max},\hat{k}} - \hRoob_{\infty,\hat{k}})^{-1}  ( \hRoob_{1,\hat{k}} - \hRoob_{2,\hat{k}} )  \right\rceil.$
        On the other hand, it also controls the suboptimality to the oracle:
        \begin{align*}
            \underbrace{R_{\hat{M},\hat{k}} - \min_{k\in\cK_n} R_{\infty,k}}_{\text{suboptimality to the oracle}} &= \underbrace{R_{\hat{M},\hat{k}} - \min_{k\in\cK_n} R_{M_{\max},k}}_{\text{suboptimality to the restricted oracle}} 
            ~+~ 
            \quad
            \underbrace{\min_{k\in\cK_n} R_{M_{\max},k} - \min_{k\in\cK_n} R_{\infty,k}}_{\text{unavoidable budget error}}.
        \end{align*}
        When $\delta=0$, the suboptimality to the restricted oracle vanishes, and the tuned ensemble simply tracks the optimal $M_{\max}$-ensemble.

        \item {To bag or not to bag:}
        The benefit of ensemble learning may be slight in some cases.
        For instance, when the number of samples is much larger than the feature dimensions and the signal-noise ratio is large, ensemble learning can only provide minor improvements over the non-ensemble predictor.
        Suppose that $\hRoob_0$ is the estimated risk of the null predictor, $\min_{k\in\cK_n}\hRoob_{1,k}$ is the optimal estimated risk due to subsampling when $M=1$, and $\min_{k\in\cK_n}\hRoob_{M_{\max},k}$ is the optimal ECV estimate with maximum ensemble size $M_{\max}$.
        Let $\zeta >0$ be a user-specified improvement factor that encodes the desired excess risk improvement in a multiplicative sense.
        Then we can decide to bag if either the null risk is smaller in the sense that $\hRoob_0<\min_{k\in\cK_n}\hRoob_{1,k}$, or the improvement due to ensemble exceeds $\zeta$ times the improvement due to subsampling:
        \begin{align}
            \underbrace{\min_{k\in\cK_n}\hRoob_{1,k} - \min_{k\in\cK_n}\hRoob_{M_{\max},k}}_{\text{improvement due to ensemble}} 
            \quad
            &~>~ 
            \underbrace{\vphantom{\min_{k \in \cK_n} \hRoob_{1,k}} \zeta}_{\text{improvement factor}}
            \times \underbrace{\hRoob_0-\min_{k\in\cK_n}\hRoob_{1,k}}_{\text{improvement due to subsampling}}.
              \label{eq:tobag}
        \end{align}

        \item {Absolute versus normalized tolerances:} The choice of the tolerance threshold $\delta$ is for controlling the absolute suboptimality, but the scale of the prediction risk may be different for different predictors and datasets.
        One can normalize the estimated risk by the null predictor's estimated risk to make the tolerance threshold comparable across different predictors and datasets, or tune based on the multiplicative guarantee \eqref{eq:risk-minimization-estimate-multiplicative}. 
    \end{enumerate}

\section{Numerical illustrations}\label{sec:simulation}

    In this section, we evaluate \Cref{alg:cross-validation} on synthetic data.
    In \Cref{subsec:experiment-risk-extrapolation}, we inspect whether the extrapolated risk estimates $\hRoob_{M,k}$ serve as reasonable proxies for the actual out-of-sample prediction errors for various base predictors on uncorrelated features.
    In \Cref{subsec:ex-k-M}, we further evaluate the risk minimization performance with tuned ensemble parameters $(\hat{M},\hat{k})$.
    Finally, in \Cref{subsec:tuning-random-forests}, we consider tuning $M$ for random forests on correlated features.

\subsection{Validating extrapolated risk estimates}\label{subsec:experiment-risk-extrapolation}
In this simulation, we examine whether our risk extrapolation strategy provides reasonable risk estimates for specific values of ensemble size $M$ and subsample size $k$ based on only the risk estimates for $M = 1$ and $M = 2$.
We evaluate six base predictors (in \cref{fig:predictors}) on data models:
    \begin{enumerate}[(M1)]
        \item\label{model:linear} Linear: A linear model: $y= \bx^{\top}\bbeta + \epsilon$,

        \item\label{model:quad} Quad: A polynomial regression model: $y= \bx^{\top}\bbeta + ((\bx^{\top}\bbeta)^2 - \tr(\bSigma_{\rhoar})/p) + \epsilon$,

        \item\label{model:tanh} Tanh: A single-index regression model: $y= \tanh(\bx^{\top}\bbeta) + \epsilon$,
    \end{enumerate}
    where the features and the coefficients are generated from $\bx\sim \cN_p(\zero,\bSigma_{\rhoar})$ and $\bbeta$ is the average of $\bSigma_{\rhoar}$'s eigenvectors associated with the top-5 eigenvalues.
    Here $\bSigma_{\rhoar} = (\rhoar^{|i-j|})_{1\leq i,j\leq p}$ is the covariance matrix of an auto-regressive process of order 1 (AR(1)) with $\rhoar = 0.5$.
    The additive noise is sampled from $\epsilon\sim\cN(0,\sigma^2)$ with $\sigma=0.5$.
    In this setup, \ref{model:linear} has a signal-noise ratio of 2.4.
    Here, the data aspect ratio $\phi=p/n$ varies from 0.1 (low-dimensional regime) to 10 (high-dimensional regime), and the subsample aspect ratio $\phi_s=p/k$ varies from $0.1$ to $10$ and from $10$ to $100$, respectively.
    The \emph{null risk}, the risk of the null predictor that always outputs zero, can also be estimated at each $\phi$.
    For ridgeless and lassoless predictors, we use rule \eqref{eq:tobag} in \Cref{subsec:practice} to exclude $k$ with exploding risks more than 5 times the estimated null risk.
    
    \begin{wrapfigure}[15]{r}{0.45\textwidth}
        \centering
        \includegraphics[width=0.45\textwidth]{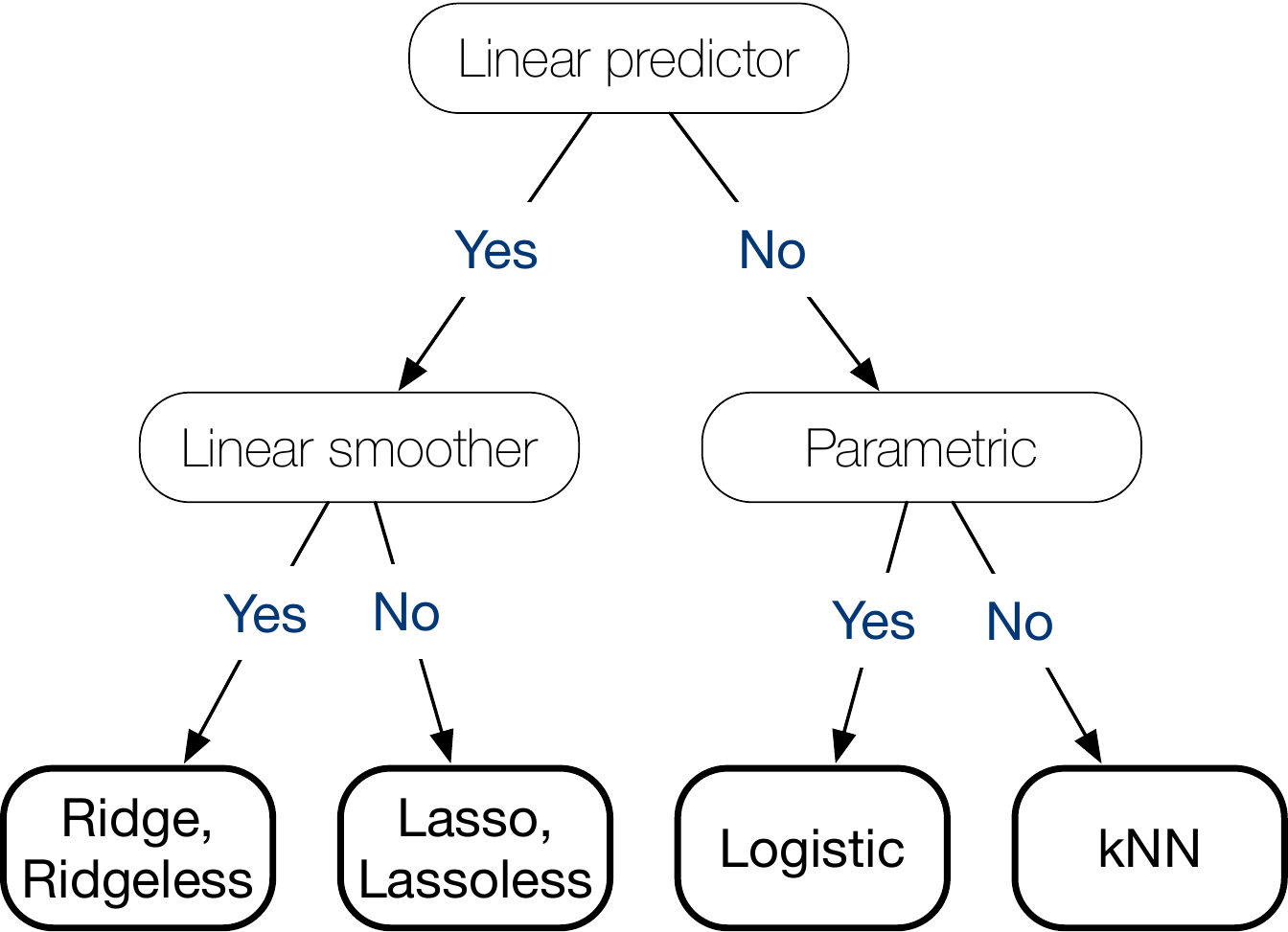}
        \caption{Predictors for regression tasks evaluated in \Cref{subsec:experiment-risk-extrapolation,subsec:ex-k-M}.}\label{fig:predictors}\end{wrapfigure}
The \ECV estimates and the corresponding prediction errors for different ensemble sizes $M$ and subsample aspect ratios $\phi_s$ are then summarized in \Cref{app:ex-result} for bagged and subagged predictors.
As $k$ decreases, the ensemble with subsample size $k$ behaves more like the null predictor.
As a result, the risk curves approach a particular value as the subsample aspect ratio $p/k$ increases.
From \crefrange{fig:est-ridge}{fig:est-kNN}, we observe a good match between the ECV estimates and the out-of-sample prediction errors.
This suggests that \Cref{thm:limiting-oobcv-for-arbitrary-M-cond} potentially applies to various types of predictors.
Comparing the results of bagging and subagging, the risk estimates are very similar, especially in the overparameterized regime when $p>n$.
Thus, we will only present the results using bagging for illustration purposes.

\subsection[Tuning ensemble and subsample sizes]{Tuning ensemble and subsample sizes}\label{subsec:ex-k-M}

Next, we examine the performance of \ECV on predictive risk minimization. 
More specifically, we apply \Cref{alg:cross-validation} using the rule \eqref{eq:tobag} with $\zeta=5$ to tune an ensemble that is close to the optimal $M_{\max}$-ensemble up to additive error $\delta$, where the maximum ensemble size $M_{\max}$ is $50$ and optimality threshold $\delta$ ranges from 0.01 to 1.
With the same predictors and data used in \Cref{subsec:experiment-risk-extrapolation}, their out-of-sample mean squared errors are evaluated on the same test set.

\begin{figure}[!t]
    \centering
    \includegraphics[width=\textwidth]{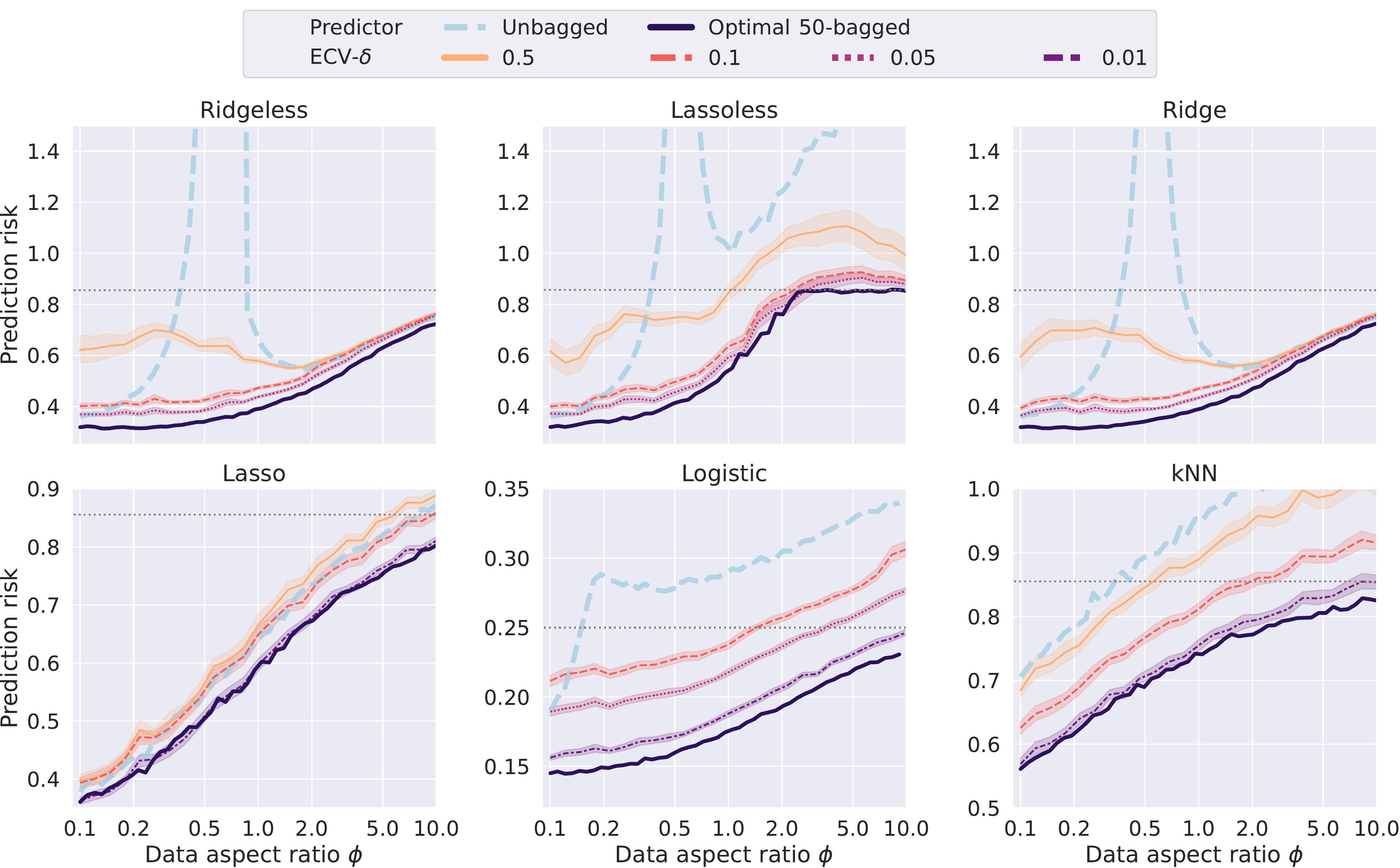}
    \caption{Prediction risk for different bagged predictors by \ECV, under model \ref{model:quad} with $\sigma=\rhoar=0.5$, $M_0=10$, and $M_{\max}=50$, for varying $\phi$ and tolerance threshold $\delta$.
    An ensemble is fitted when \eqref{eq:tobag} is satisfied with $\zeta=5$.
    The null risks and the risks for the non-ensemble predictors are marked as gray dotted lines and blue thick dashed lines, respectively.
    The points denote finite-sample risks averaged over 100 dataset repetitions, and the shaded regions denote the values within one standard deviation, with $n=1,000$ and $p=\lfloor n\phi\rfloor$.}
    \label{fig:oobcv-quad-bagging}
\end{figure}

As summarized in \cref{fig:oobcv-quad-bagging}, the thick dashed lines represent the non-ensemble prediction risk, and the thick solid lines represent the prediction risk of optimal $50$-bagged predictors using a finer grid.
Note that the former may be non-monotonic in the data aspect ratio $\phi$, but the latter is increasing in $\phi$.
The finite-sample prediction errors of the \ECV-tuned predictor are shown as thin lines.
As we can see, when $\delta$ decreases, the prediction errors of \ECV get closer to those of the optimal 50-bagged predictor.
The slight discrepancy between the \ECV-tuned risks with the least $\delta$ and the oracle risks comes from the fact that a coarser grid $\cK_n$ is used for \ECV tuning.
Overall, the results suggest that the ECV-tuned ensemble parameters $(\hat{M},\hat{k})$ give risks close to the oracle choices for various predictors within the desired optimality threshold $\delta$ in finite samples.
Lastly, though \ECV is proposed for regression tasks, the numerical results in \Cref{app:subsec:imbalance} support its superiority over \kfoldcv in imbalanced binary classification scenarios.

\subsection[Tuning ensemble sizes of random forests]{Tuning ensemble sizes of random forests}\label{subsec:tuning-random-forests}

When the data aspect ratio $p/n$ is too small, tuning both the subsample size $k$ and the ensemble size $M$ may be unnecessary.
In such cases, tuning the ensemble size $M$ (in the sense that how large $M$ is sufficient to have good performance) is a more substantial and practical consideration.
In this experiment, we apply \Cref{alg:cross-validation} to tune only the ensemble size of random forests.

Since the most crucial advantage of the random forest model is its flexibility to incorporate highly correlated variables while avoiding multi-collinearity issues, we consider the nonlinear model \ref{model:quad} with a non-isotropic AR(1) covariance matrix.
For a given dataset, we examine two strategies to estimate the conditional prediction risks.
The first utilizes the OOB observations according to \Cref{alg:cross-validation}, while the other uses a hold-out subset to estimate the risks.
Similar to \citet{lopes2019estimating}, $\lfloor n/6\rfloor$ observations are randomly selected as the evaluation set for the hold-out estimates.
As suggested by \citet{hastie2009elements}, each decision tree uses $\lfloor p/3\rfloor$ randomly selected features with a minimum node size of 5 as the default without pruning.
To build each tree, we fix the subsample size $k=n(1-1/\log n)$ observations for subagging.
The results are shown in \cref{fig:ar1-rf}, where the standard deviation of the estimates are also visualized as error bars.
In the underparameterized regime when $n>p$, we observe that ECV and hold-out estimates have similar performance.
Both of them are close to the out-of-sample errors in this case.
However, in the overparameterized regime when $n<p$, the hold-out estimates suffer from biases due to sample splitting.
On the contrary, the ECV estimates are still accurate and have smaller variability compared to the hold-out estimates.
In the right panel of \cref{fig:ar1-rf}, we see that ECV estimates provide a valid extrapolation path from $M=20$ to $M=500$ in the high-dimensional scenarios.

Finally, we conduct a sensitivity analysis of the hyperparameter $M_0$, the correlative strength $\rhoar$ of the feature covariance matrix, and the covariance structures in \Crefrange{app:subsubsec:M0}{app:subsubsec:cov}.
The results suggest that \ECV is relatively robust under various scenarios and various choices of hyperparameters.
In \Cref{app:subsubsec:mtry}, we illustrate the utility of \ECV for tuning the number of features drawn when splitting each node of random forests.

\begin{figure}[!t]
    \centering
    \includegraphics[width=0.8\textwidth]{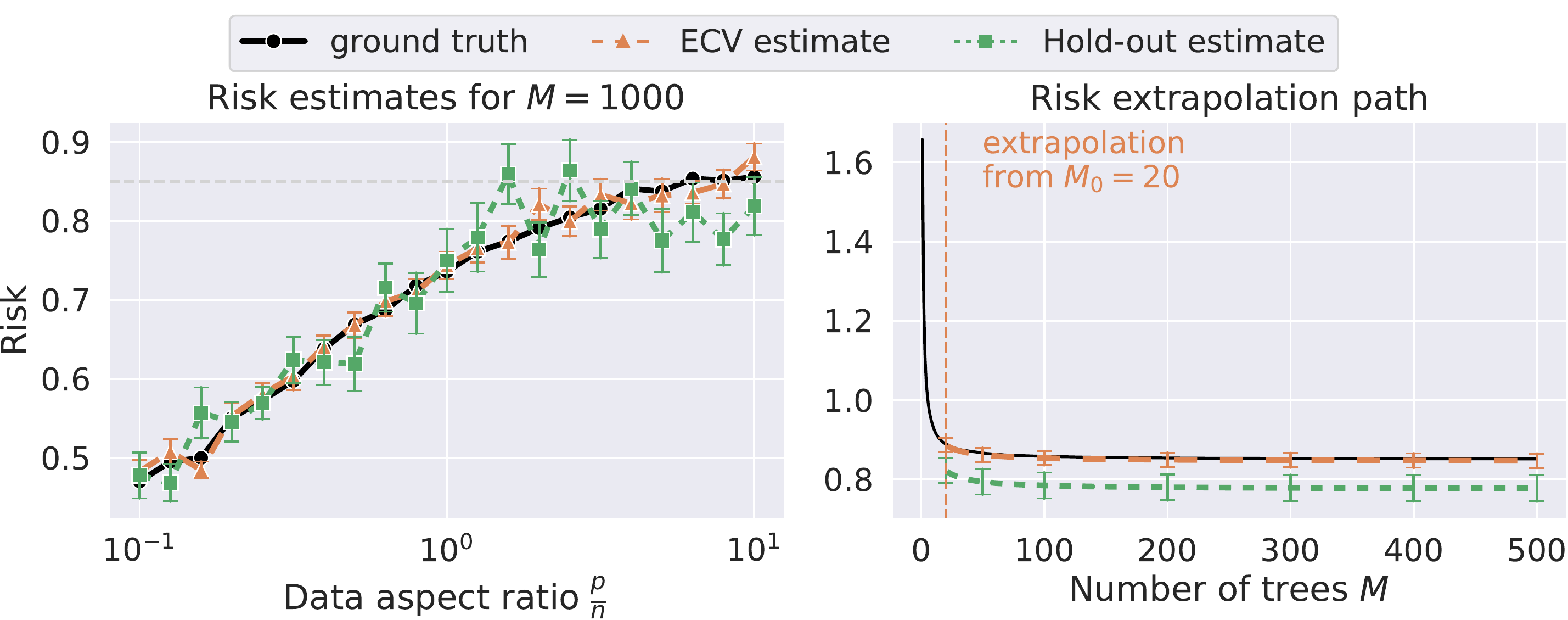}
    \caption{Risk extrapolation for random forest based on the first $M_0=20$ trees, under model \ref{model:quad} with $\sigma=\rhoar=0.5$.
    The left panel shows the risk estimates for random forests with $M=500$ trees and varying data aspect ratios, and the gray dash line $y=0.85$ denotes the risk of the null predictor; the right panel shows the full risk extrapolation path in $M$ when $p/n=8$.
    The error bar denotes one standard deviation across 100 simulations, and the ground truth is estimated from $2,000$ test observations, with $n=500$.}
    \label{fig:ar1-rf}
\end{figure}

\section{Applications to single-cell multiomics}\label{sec:app-sc}
    In genomics, cell surface proteins act as primary targets for therapeutic intervention and universal indicators of particular cellular processes. 
    More importantly, immunophenotyping of cell surface proteins has become an essential tool in hematopoiesis, immunology, and cancer research over the past 30 years \citep{hao2021integrated}. 
    However, most single-cell investigations only quantify the transcriptome without cell-matched measures of related surface proteins due to technical limitations and financial constraints \citep{zhou2020surface,du2022robust}. 
    The specific cell types and differentially abundant surface proteins are determined after thoroughly analyzing the transcriptome. 
    This has led researchers to investigate how to reliably predict protein abundances in individual cells using their gene expressions.
    Specifically, the effectiveness of ensemble methods has been illustrated by \citep{heckmann2018machine,li2019joint,xu2021ensemble} on the protein prediction problem. 
    Yet, in practice, because of the lack of theoretical results and pragmatic guidelines, the ensemble and subsample sizes are generally determined by ad hoc criteria.

    In this section, we apply the proposed method to real datasets in single-cell multi-omics  \citep{hao2021integrated}. See \Cref{app:sc} for details on the datasets and the preprocessing steps.
    Based on these real-world datasets, we compare three different cross-validation methods for tuning both the ensemble size and the subsample size of random forests: (1) \kfoldcv: the $K$-fold CV ($K=5$); 
    (2) \splitcv: sample-split or holdout CV (the ratio of training to validation observations is 5:1); and (3) \ECV: the proposed extrapolated CV.
    The grid of subsample sizes $\cK_n$ is generated according to \Cref{alg:cross-validation}.
    To ensure all the CV methods are comparable and fairly evaluated, we evaluate the three CV methods on the same grid $[M_{\max}]\times \cK_n$ for $M_{\max}=50$.
    Decision trees are used as the base predictors to predict the abundance of each protein based on the gene expressions of subsampled cells.
    After the tuning parameters $(\hat{M},\hat{k})$ are obtained, we refit the ensemble on the entire training set and evaluate it on the test set.
    For each method, the $M_{\max}$ base predictors are fitted once so that the training costs are almost the same for all three.
    The computational complexity of the three methods is discussed in \Cref{subsec:practice}.
    Because different proteins may have different variances, we measure the overall protein prediction accuracy by the normalized mean squared error (NMSE), which is the ratio of the mean squared error to the empirical variance on the test set.
    Our target is to select a $\delta$-optimal random forest so that its NMSE is no more than $\delta=0.05$ away from the best random forest with 50 trees.

     To illustrate our proposed method, we visualize the prediction risk estimate and out-of-sample error for surface protein CD103 in \cref{fig:heatmap_5,fig:heatmap}.
    Because the response is centered, the empirical variance of the response serves as an estimate of the null risk, i.e., the risk of the null predictor that always outputs zeros.
    A value of NMSE less than one indicates that the predictor performs better than the null risk.
    From \cref{fig:heatmap_5}, we see that the ECV extrapolated estimates in the left panel are largely consistent with the actual prediction errors in the right panel.
    The out-of-sample error can still be considerable when $M=10$ as shown in \cref{fig:heatmap_5}.
    As the ensemble size $M$ increases, both become more stable for various subsample sizes.
    Further, we observe that the tuned ensemble and subsample sizes are close to the optimal ones on the test dataset in finite samples.
    The tune subsample size $k$ is much smaller than the total sample size $n$.
    This indicates that our proposed method tracks the out-of-sample optimal parameter well.

    \begin{figure}[!t]
        \centering
            \includegraphics[width=\textwidth]{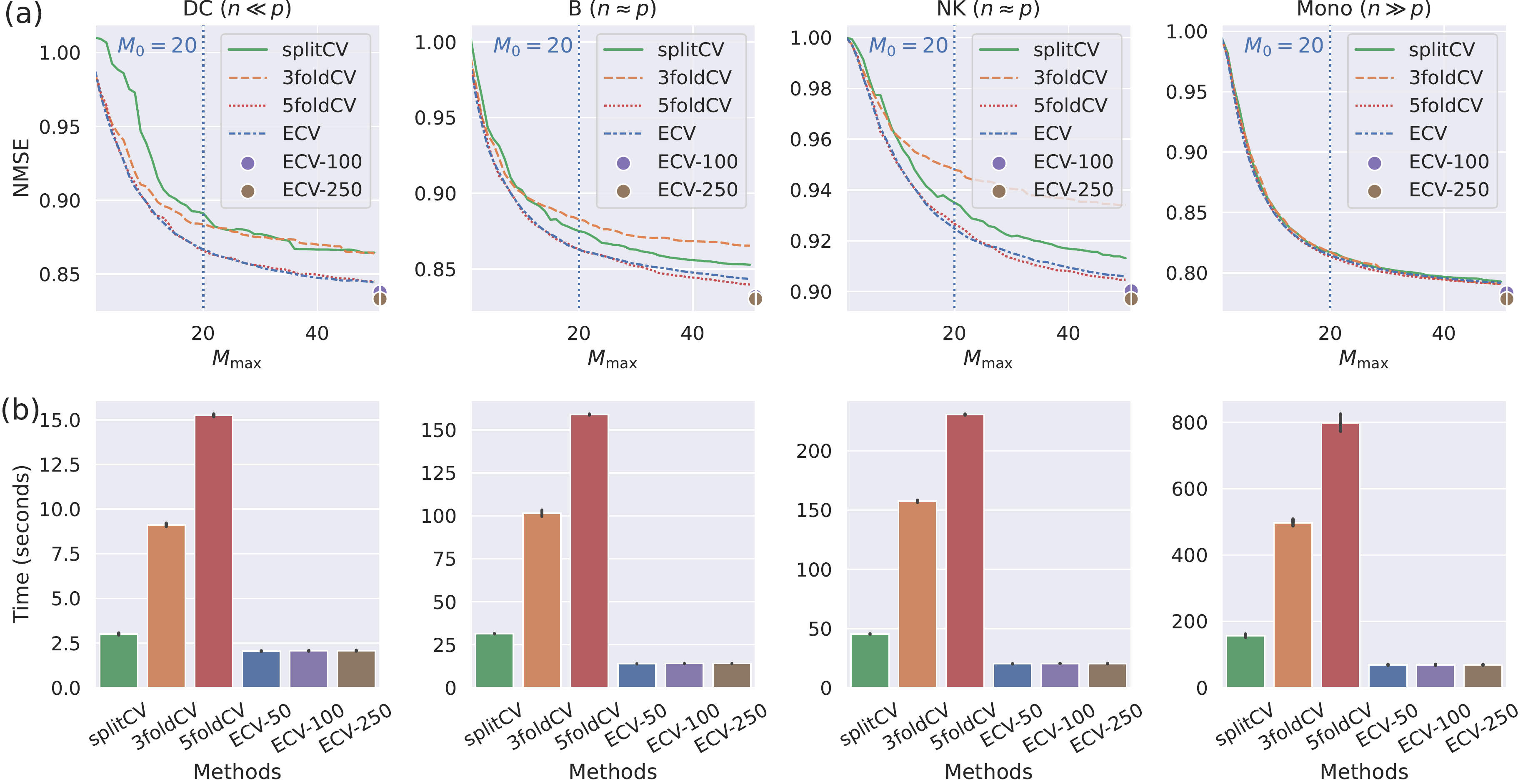}
        \caption{Performance of CV methods on predicting the protein abundances in different cell types. 
        {(a)} The average NMSEs of the cross-validated predictors for different methods. \ECV uses $M_0=20$ trees to extrapolate risk estimates, and the points correspond to $M_{\max}\in\{100,250\}$.
        {(b)} The average CV time consumption in seconds.
        }
        \label{fig:sc}
    \end{figure}
    
    We compare the performance of different CV methods on various cell types.
    As shown in \cref{fig:sc}(a), \kfoldcv is very sensitive to the choice of the number of fold $K$.
    \splitcv and \kfoldcv with $K=3$ have the worst out-of-sample performance among all methods.
    In the low-dimensional dataset of the Mono cell type, all methods have similar performance.
    However, \ECV significantly improves upon \splitcv with insufficient sample sizes, as in the DC, B, and NK cell types.
    In all cases, the out-of-sample NMSEs of \ECV are comparable to \kfoldcv with $K=5$.
    Yet, \ECV is more flexible to tuning large ensemble sizes, such as $M_{\max}=100$ and $250$.
    This suggests that the extrapolated risk estimates of \ECV based on only $M_0=20$ trees are accurate for tuning ensemble parameters on these datasets.

    Regarding time consumption, recall that the base predictors are only fitted once (\Cref{subsec:practice}) so that the comparison is fair for all methods. 
    From \cref{fig:sc}(b), we see that tuning by \ECV is significantly faster than \kfoldcv because we are extrapolating the risks from $20$ trees rather than estimating them for every $M > 20$.
    Though \ECV has similar time complexity as \splitcv when the sample sizes are small, it uses less than 50\% of the time used by \splitcv as the sample size increases.
    Thus, \ECV achieves better computational efficiency than the alternative CV methods in a variety of settings.

\section{Discussion}\label{sec:conclusion}

    This paper addresses the challenge of tuning the ensemble and subsample sizes for randomized ensemble learning.    
    While previous work has extensively focused on the statistical properties of bagging and random forests, ensemble tuning remains an area that has received comparatively less attention.
    To bridge this gap, this paper introduces a method that efficiently provides $\delta$-optimal ensemble parameters (for the tuned risk) without requiring an exhaustive search over all feasible parameter combinations, a common limitation of conventional cross-validation methods.
    
    Our method hinges on the extrapolation of risk estimates, which is afforded due to a special decomposition of the squared risk and the consistent risk estimators of each component.
    Our method is sample-efficient and can be naturally extended to tuning hyperparameters other than subsample sizes, as shown in \Cref{app:subsubsec:mtry}.
    Furthermore, our algorithm guarantees the $\delta$-optimality of the tuned ensemble risk in relation to the oracle ensemble risk.

    To demonstrate the practical utility of \ECV, we apply it to the task of predicting proteins in single-cell data.
    In scenarios with small sample sizes, \ECV achieves smaller out-of-sample errors without sample splitting compared to traditional \splitcv.
    On the other hand, it drastically reduces the time complexity compared to $K$-fold CV and maintains comparable accuracy without repeatedly fitting numerous predictors.
    In summary, \ECV exhibits both statistical and computational efficiency in protein prediction~tasks.

    We next point out some limitations of this work and propose possible future directions.
    Future directions for this work include extending the theoretical framework to other types of loss functions, such as smooth loss functions, by applying the "sandwich" sub-optimality gap approach \citep[Proposition 3.6]{patil2022bagging}. Additionally, exploring variance estimation and extrapolation for the proposed risk estimators could provide a more comprehensive understanding of their uncertainty. Integrating the concept of performing CV with confidence \citep{lei2020cross} could address the important issue of quantifying confidence in a tuned model for model selection.

    \paragraph{Acknowledgments}
    We thank the editor, the associate editor, and two anonymous reviewers for their valuable and constructive comments, which led to various improvements in this paper.
    This work used the Bridges-2 system at the Pittsburgh Supercomputing Center (PSC) through allocations BIO220140 and MTH230020 from the Advanced Cyberinfrastructure Coordination Ecosystem: Services \& Support (ACCESS) program.
    This project was partially funded by the National Institute of Mental Health (NIMH) grant R01MH123184.

\bigskip

{

}
\end{bibunit}

\clearpage

\begin{bibunit}[apalike]

\beginsupplement

\begin{center}
{\Large
{\bf
Supplementary material for \\``\titleRLB''
}}
\end{center}

This document acts as a supplement to the paper ``\titleRLB.'' 
The section numbers in this supplement begin with the letter ``S'' and the equation numbers begin with the letter ``S'' to differentiate them from those appearing in the main paper.

\addcontentsline{toc}{section}{Notation and organization}
\section*{Notation and organization}

\subsection*{Notation}
\label{sec:notation}

Below, we provide an overview of the notation used in the main paper and the supplement.
\begin{enumerate}[labelsep=1mm,leftmargin=7mm]
    \item \underline{General notation}: We denote scalars in non-bold lower or upper case (e.g., $n$, $\lambda$, $C$), vectors in lower case (e.g., $\bx$, $\bbeta$), and matrices in upper case (e.g., $\bX$).
    For a real number $x$, $(x)_{+}$ denotes its positive part, $\lfloor x \rfloor$ its floor, and $\lceil x \rceil$ its ceiling.
    For a natural number $n$, $n!=\prod_{i=1}^n i$ denotes the $n$ factorial.
    For a vector $\bbeta$, $\| \bbeta \|_{2}$ denotes its $\ell_2$ norm.
    For a pair of vectors $\bv$ and $\bw$, $\langle \bv, \bw \rangle$ denotes their inner product.
    For an event $A$, $\ind_A$ denotes the associated indicator random variable.
    We use $\cO_p$ and $o_p$ to denote probabilistic big-O and little-o notation, respectively.

    \item \underline{Set notation}: We denote sets using calligraphic letters (e.g., $\cD$), and use blackboard letters to denote some special sets: $\NN$ denotes the set of positive integers, $\RR$ denotes the set of real numbers, $\RR_{\ge 0}$ denotes the set of non-negative real numbers, and $\RR_{> 0}$ denotes the set of positive real numbers.
    For a natural number $n$, we use $[n]$ to denote the set $\{ 1, \dots, n \}$.

    \item \underline{Matrix notation}: For a matrix $\bX \in \RR^{n \times p}$, $\bX^\top \in \RR^{p \times n}$ denotes its transpose.
    For a square matrix $\bA \in \RR^{p \times p}$, $\bA^{-1} \in \RR^{p \times p}$ denotes its inverse, provided it is invertible.
    For a positive semi-definite matrix $\bSigma$, $\bSigma^{1/2}$ denotes its principal square root.
    A $p \times p$ identity matrix is denoted $\bI_p$, or simply by $\bI$, when it is clear from the context.
\end{enumerate}

\subsection*{Organization}

Below, we outline the structure of the rest of the supplement.

\begin{itemize}[labelsep=1mm,leftmargin=7mm]
    \item In \Cref{app:oobcv}, we present proofs of results appearing in \Cref{sec:oobcv}.
    \begin{itemize}
        \item \Cref{app:prop:squared_risk_decom} proves \Cref{prop:squared_risk_decom}.
        \item \Cref{app:prop:bounded-variance-error-control-mul-form} proves \Cref{prop:bounded-variance-error-control-mul-form}.

    \item \Cref{app:subsec:M12-cv}
    proves \Cref{lem:consistency-oobcv},
    conditional on certain helper lemmas presented in the next section.

    \end{itemize}
    
    \item 
    \Cref{app:risk} provides intermediate concentration results 
    used in the proof of \Cref{lem:consistency-oobcv}.
    
    \item In \Cref{app:extrapolation}, we present proof of results in \Cref{subsec:risk-extrapolation}.
    \begin{itemize}
        \item \Cref{app:thm:limiting-oobcv-for-arbitrary-M-cond} proves \Cref{thm:limiting-oobcv-for-arbitrary-M-cond}.
        \item \Cref{app:subsec:mul-opt} extends additive optimality in \Cref{thm:limiting-oobcv-for-arbitrary-M-cond} to multiplicative optimality
        stated in \Cref{rm:mul-opt}.
        \item \Cref{app:ridge} specialize \Cref{thm:limiting-oobcv-for-arbitrary-M-cond} to ridge predictor under weaker assumptions.
    \end{itemize}
    \item In \Cref{sec:appendix-concerntration}, we collect various technical helper lemmas
    related to concentrations and convergences along with their proofs that are used in various proofs in 
    \Crefrange{app:oobcv}{app:extrapolation}.

    \item In \Cref{app:ex-result}, we present additional numerical results
    for \Cref{sec:simulation} and \Cref{sec:app-sc}.
    \begin{itemize}
        \item \Cref{app:ex-risk-extrapolation} presents additional illustrations for subagging and bagging in \Cref{subsec:experiment-risk-extrapolation}.
        
        \item \Cref{app:ex-k-M} presents additional illustrations for subagging and bagging in \Cref{subsec:ex-k-M}.

        \item \Cref{app:subsec:imbalance} presents results of ECV on imbalanced classification.
        
        \item \Cref{app:tuning-random-forests} presents additional illustrations for bagging in \Cref{subsec:tuning-random-forests}.
        
        \item \Cref{app:sc} presents additional illustrations for \Cref{sec:app-sc}.
    \end{itemize}
\end{itemize}

\section{Related work on cross-validation}\label{sec:related-work-general-cv}
    Different cross-validation (CV) approaches have been proposed for parameter tuning and model selection \citep{allen_1974, stone_1974,stone_1977,geisser_1975}.
    We refer readers to \cite{arlot_celisse_2010,zhang2015cross} for a review of different CV variants used in practice.
    The simplest version of CV is the sample-split CV \citep{hastie2009elements}, which holds out a specific portion of the data to evaluate models with different parameters.
    By repeated fitting of each candidate model on multiple subsets of the data, $K$-fold CV extends the idea of the sample splitting and reduces the estimation uncertainty.
    When $K$ is small, the risk estimate may inherit more uncertainty; however, it can be computationally prohibitive when $K$ is~large. Asymptotic distributions of suitably normalized $K$-fold CV are obtained in~\citet{austern_zhou_2020}, under some stability conditions on the predictors.

    In a high-dimensional regime where the number of variables is comparable to the number
    of observations, the commonly-used small values of $K$ such as $5$ or $10$
    suffer from bias issues in risk estimation \citep{rad_maleki_2020}.
    Leave-one-out cross-validation (LOOCV), i.e., the case when $K = n$, alleviates the bias issues in risk estimation, whose theoretical properties have been analyzed in recent years by \citet{kale_kumar_vassilvitskii_2011,kumar_lokshtanov_vassilviskii_vattani_2013,rad_zhou_maleki_2020}.
    However, LOOCV, in general, is computationally expensive to evaluate, and there has been some work on approximate LOOCV to address the computational issues \citep{wang2018approximate,stephenson_broderick_2020,wilson_kasy_mackey_2020,rad_maleki_2020}.
    Another line of research about CV is on statistical inference; see, for example, \citet{wager2014confidence,lei2020cross,bates2021cross}.
    Central limit theorems for CV error and a consistent estimator of its variance are derived in \citet{bayle_bayle_janson_mackey_2020}, which assumes certain stability assumptions, similar to \citet{kumar_lokshtanov_vassilviskii_vattani_2013, celisse_guedj_2016}. 
    Their results yield asymptotic confidence intervals for the prediction error and apply to $K$-fold CV and LOOCV.
    A naive application of these traditional CV methods for ensemble learning to tune $M$ and $k$ requires fitting the ensembles of arbitrary sizes $M$, leading to a higher computational cost.

\section{Proofs of results in \Cref{sec:oobcv}}\label{app:oobcv}

\subsection[Proof of \Cref{prop:squared_risk_decom}]{Proof of \Cref{prop:squared_risk_decom} (Squared risk decomposition)}\label{app:prop:squared_risk_decom}
\begin{proof}[Proof of \Cref{prop:squared_risk_decom}.]
    We start by expanding the squared risk as:
    \begin{align*}
        &R(\tf_{M,k} ; \, \cD_n,\{I_{\ell}\}_{\ell = 1}^{M}) \\
        &= \int\left(y
        -\frac{1}{M}\sum\limits_{\ell=1}^M\hf(\bx;\cD_{I_\ell})\right)^2 \,\rd P(\bx,y) \\
        &=
        \int
        \left(
        \frac{1}{M}
        \sum_{\ell = 1}^{M}
        \big(y - \hf(\bx; \cD_{I_\ell})\big)
        \right)^2
        \, \mathrm{d}P(\bx, y) \\
        &=
        \frac{1}{M^2}
        \sum_{\ell=1}^{M}
        \int \big(y - \hf(\bx; \cD_{I_\ell})\big)^2
        \, \mathrm{d}P(\bx, y)  + 
        \frac{1}{M^2}
        \sum_{i = 1}^{M}
       \sum\limits_{\substack{j=1 \\ j \neq i}}^{M}
        \int
        \big(y - \hf(\bx; \cD_{I_i})\big) \big(y - \hf(\bx; \cD_{I_j})\big)
        \, \mathrm{d}P(\bx, y) \\
        &=
        \frac{1}{M^2}
        \sum_{\ell=1}^{M}
        R(\tf_{1, k}; \cD_n, I_\ell) + 
        \frac{1}{M^2}
        \sum_{i = 1}^{M}
       \sum\limits_{\substack{j=1 \\ j \neq i}}^{M}
        \int
        (y - \hf(\bx; \cD_{I_i})) (y - \hf(\bx; \cD_{I_j}))
        \, \mathrm{d}P(\bx, y) \\
        &\stackrel{(i)}{=}
        \frac{1}{M^2}
        \sum_{\ell=1}^{M}
        R(\tf_{1, k}; \cD_n, I_\ell) \\
        &~
        \qquad  + 
        \frac{1}{M^2}
        \sum_{i = 1}^{M}
       \sum\limits_{\substack{j=1 \\ j \neq i}}^{M}
        \int
        \frac{1}{2}
        \left\{
        4
        \Big(y - \frac{1}{2} \big(\hf(\bx; \cD_{I_i}) + \hf(\bx; \cD_{I_j})\big)  \Big)^2 \right.\\
        &\qquad \qquad  \left.
        - \big(y - \hf(\bx; \cD_{I_i})\big)^2
        - \big(y - \hf(\bx; \cD_{I_j})\big)^2
        \right\}
        \, \mathrm{d}P(\bx, y) \\
        &=
        \frac{1}{M^2}
        \sum_{\ell=1}^{M}
        R(\tf_{1, k}; \cD_n, I_\ell) \\
        &~
        + 
        \frac{1}{M^2}
        \sum_{i = 1}^{M}
       \sum\limits_{\substack{j=1 \\ j \neq i}}^{M}
       \frac{1}{2}
       \left\{
       4 R(\hf_{2, k}; \cD_n; I_{i}, I_{j})
       - R(\tf_{1, k}; \cD_{n}; I_i)
       - R(\tf_{1, k}; \cD_{n}; I_j)
       \right\} \\
        &=
        \frac{1}{M^2}
        \sum_{\ell=1}^{M}
        R(\tf_{1, k}; \cD_n, I_\ell)
        \\
        &~ - \frac{1}{2M^2}
        \sum_{i=1}^{M} 
        \sum\limits_{\substack{j=1 \\ j \neq i}}^{M}
        R(\tf_{1,k}; I_i)
        - \frac{1}{2M^2}
        \sum_{i=1}^{M} 
        \sum\limits_{\substack{j=1 \\ j \neq i}}^{M}
        R(\tf_{1,k}; I_j)
        + 
        \frac{1}{M^2}
        \sum_{i=1}^{M}
        \sum\limits_{\substack{j = 1 \\j \neq i}}^{M}
       2 R(\hf_{2, k}; \cD_n; I_{i}, I_{j}) \\
       &=
       \frac{1}{M^2} \sum_{\ell = 1}^{M} R(\tf_{1, k}; \cD_n; I_{\ell})
       - \frac{1}{2M^2} 
       \cdot 2  \cdot (M - 1) \sum_{\ell=1}^{M}
       R(\tf_{1, k}; I_{\ell})
       + \frac{2}{M^2}
       \sum\limits_{\substack{i, j \in [M] \\ i \neq j}}
       R(\hf_{2, k}; \cD_n; I_i, I_j) \\
       &=
       \left(
       \frac{1}{M^2} 
       - \frac{(M - 1)}{M^2} 
       \right)
       \sum_{\ell = 1}^{M} R(\tf_{1, k}; \cD_n; I_{\ell})
       \\
       &\qquad + \frac{2}{M^2}
       \sum\limits_{\substack{i, j \in [M] \\ i \neq j}}
       R(\hf_{2, k}; \cD_n; I_i, I_j) \\
        &= -\left(\frac{1}{M}-\frac{2}{M^2}\right)\sum_{\ell=1}^M R(\tf_{1,k}; \, \cD_n,\{I_{\ell}\}) + \frac{2}{M^2}\sum\limits_{\substack{i,j\in[M]\\i\neq j}}  R(\tf_{2,k} ; \, \cD_n, \{I_{i},I_j\}).
    \end{align*}
    In the expansion above,
    for equality $(i)$,
    we used the fact that $ab = \{ 4(a/2 + b/2)^2 - a^2 - b^2 \} /2$.
    This finishes the proof.
\end{proof}

\subsection[Proof of \Cref{prop:bounded-variance-error-control-mul-form}]{Proof of \Cref{prop:bounded-variance-error-control-mul-form} (Consistent component risk estimation)}\label{app:prop:bounded-variance-error-control-mul-form}
\begin{proof}[Proof of \Cref{prop:bounded-variance-error-control-mul-form}.]
    Let $\Delta_n = |\hR(f , \cD_I)-R(\hf;\cD_I)|$.
    We will view $p$, $|I|$, and $|I^c|$ as sequences indexed by $n$.
    Define $\hsigma_I= \|(y_0- \hf(\bx_0;\cD_I))^2 \|_{\psi_1 \mid \cD_I}$.
    From \citet[Lemma 2.9, Lemma 2.10]{patil2022bagging}, we have    
    \begin{align}
        \PP\left(\Delta_n \geq C \hsigma_I \max\left\{\sqrt{\frac{A\log n}{|I^c|}},  \frac{A\log n}{|I^c|}\right\}\right) \leq n^{-A}, \label{eq:bound-Delta}
    \end{align}
    for some positive constant $C$.
    Let $\kappa_n = C \hsigma_I \max\left\{\sqrt{\frac{A\log n}{|I^c|}},  \frac{A\log n}{|I^c|}\right\}$.
    
    Since $\hsigma_I =o_p(\sqrt{|I^c|/\log n})$  as $n\rightarrow\infty$, we have that $\kappa_n = o_p(1)$.
    Let $A>0$ be fixed,
    For all $\epsilon>0$, we have that
    \begin{align*}
        \PP(\Delta_n >\epsilon) &= \PP(\Delta_n >\epsilon \geq \kappa_n)  + \PP(\Delta_n>\epsilon,\kappa_n>\epsilon) \\
        &\leq n^{-A}+ \PP(\kappa_n>\epsilon) \rightarrow0,
    \end{align*}
    which implies that $\Delta_n=o_p(1)$.    
    The proof for $\|\cdot\|_{L_2\mid\cD_I}$ follows analogously.
\end{proof}

\subsection[Uniform consistency for (M,k)]{Proof of \Cref{lem:consistency-oobcv} (Uniform risk estimation over $(M, k)$)}\label{app:subsec:M12-cv}
    \begin{figure}[!ht]
        \centering
        \begin{tikzpicture}[scale=0.9, transform shape]

            \node[below of = Rh12cv, draw,
            text width=5cm, node distance=2.5cm, 
            align=center](R12) at (0,0) {$\sup\limits_{k\in\cK_n}|R_{M,k} - \sR_{M,k}|$ \\ $M=1,2$};

            \node[above of=R12,text width=5cm,
            align=center]{\Cref{cond:conv-risk-M12}};
            
            \node[below of = R12,
            text width=5cm,
            align=center,
            node distance=4.0cm, draw](ER12s) {$\EE\left(\sup\limits_{k\in\cK_{n}}|R_{M,k} - \sR_{M,k}| \mid \cD_n\right)$ $M=1,2$};

            \node[below of = ER12s, 
                text width=5cm, 
                align=center,
                node distance=4cm, draw](RM) {$\sup\limits_{M\in\NN,k\in\cK_n}|R_{M,k}-\sR_{M,k}|$};

            \node[right of =R12, draw,
            text width=5cm, node distance=10.0cm,
            align=center](Rh12cv) {$\sup\limits_{k\in\cK_n}|\hRoob_{M,k} - R_{M,k}|$ \\ $M=1,2$};

            \node[above of=Rh12cv,text width=5cm,
            align=center,above=-0.2cm]{\Cref{ass:kappa} and \Cref{prop:consistency-M12}};

            \node[right of = ER12s, 
                text width=5cm, 
                align=center,
                node distance=10.0cm, draw](RhM) {$\sup\limits_{M\in\NN,k\in\cK_n}| \hRoob_{M,k}-\sR_{M,k}|$};

            \node[right of = RM, 
                text width=5cm, 
                align=center,
                node distance=10.0cm, draw](RhR) {$\sup\limits_{M\in\NN,k\in\cK_n}|\hRoob_{M,k}-R_{M,k}|$};

            \draw[-triangle 45] (R12) to node [text width=2.5cm,midway,right,align=center] {
            \Cref{lem:conv:data-risk-M12}} (ER12s);

            \draw[-triangle 45] (ER12s) to node [text width=2.5cm,midway,right,align=center] {\Cref{lem:conv:risk-M}} (RM);
            
            \draw[-triangle 45] (R12.south) to node [text width=2.5cm,midway,right=3.2cm,above=-0.4cm,align=center] {\Cref{prop:consistency-M}} (RhM.north);

            \draw[-triangle 45] (Rh12cv) to node [text width=2.5cm,midway,right,align=center] {} (RhM);

            \draw[-triangle 45] (RhM.south) to node [text width=2.5cm,midway,right,align=center] {} (RM.north);

            \draw[dashed,-triangle 45] (RM) to node [text width=2.5cm,midway,right,align=center] {} (RhR);

            \draw[dashed,-triangle 45] (RhM) to node [text width=2.5cm,midway,right,align=center] {} (RhR);
            
        \end{tikzpicture}
        \caption{Reduction strategy for obtaining concentration results in \Cref{lem:consistency-oobcv}.
        The solid lines indicate basic components presented later in \Cref{app:risk} and the dashed lines indicate the proof strategy in \Cref{app:subsec:M12-cv}.
        Here, $\sR_{M,k}=2\sR_{2,k}-\sR_{1,k}+2(\sR_{1,k}-\sR_{2,k})/M$.}
        \label{fig:proof}
    \end{figure}
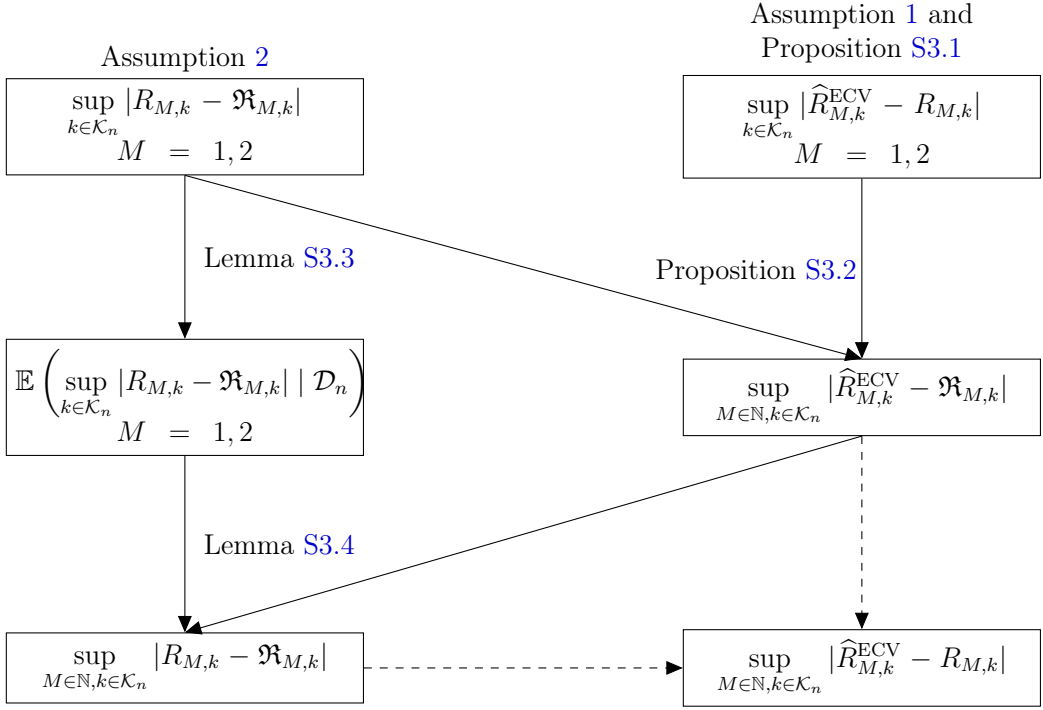
    
    \begin{proof}[Proof of \Cref{lem:consistency-oobcv}.]
        Before we prove \Cref{lem:consistency-oobcv}, we present the proof strategy in \cref{fig:proof}.
        In \cref{fig:proof}, four important auxiliary lemmas and propositions are deferred to the next section.
        For the asymptotic results, we will let $k,p$ be sequences of integers $\{k_n\}_{n=1}^{\infty}$, $\{p_n\}_{n=1}^{\infty}$ indexed by $n$ but drop the subscripts $n$.

        From \Cref{lem:conv:risk-M} we have
        \begin{align*}
             \sup_{M\in\NN,k\in\cK_n} | R_{M,k} - \sR_{M,k} |=\cO_p(n^{\epsilon}(\gamma_{1,n}+ \gamma_{2,n})).
        \end{align*}
        On the other hand, from \Cref{prop:consistency-M} we have
        \begin{align*}
            \sup_{M\in\NN,k\in\cK_n} | \hRoob_{M,k} - \sR_{M,k} |=\cO_p(\hsigma_n \log n /n^{1/2}+n^{\epsilon}(\gamma_{1,n}+ \gamma_{2,n})),
        \end{align*}
        where $\hsigma_n= \max_{m,\ell\in[M_0],k\in\cK_n}\hsigma_{I_{k,\ell}\cup I_{k,m}}$.
        By the triangle inequality, we have
        \begin{align*}
             \sup_{M\in\NN,k\in\cK_n} | \hRoob_{M,k} - R_{M,k} |=\cO_p(\hsigma_n \log n /n^{1/2} + n^{\epsilon}(\gamma_{1,n}+\gamma_{2,n}) ),
        \end{align*}
        which completes the proof.
    \end{proof}

\section{Intermediate concentration results for \Cref{lem:consistency-oobcv}}
\label{app:risk}

    In this section, we show the concentration of conditional prediction risks to their limits.
    \Cref{prop:consistency-M12}
    derives the uniform consistency over $k\in\cK_n$ of the cross-validated risk estimates to the risk $R_{M,k}$ for $M=1,2$.
    \Cref{prop:consistency-M}
    derives the uniform consistency over $(M,k)\in\NN\times \cK_n$ of the cross-validated risk estimates to the deterministic limits $\sR_{M,k}$.
    \Cref{lem:conv:data-risk-M12} establishes the concentration for the expected (with respect to sampling) conditional risk over $k\in\cK_n$.
    \Cref{lem:conv:risk-M} establishes the concentration for subsample conditional risks over $M\in\NN$ and $k\in\cK_n$.

    \subsection[Uniform consistency of cross-validated risk over k for M=1,2]{Uniform consistency of cross-validated risk over $k$ for $M=1,2$}\label{app:prop:consistency-M12}

In what follows, we derive the uniform consistency over $k\in\cK_n$ of the cross-validated risk estimates for $M=1,2$ in \Cref{app:prop:consistency-M12}.

\begin{proposition}[Uniform consistency over $k$ for $M=1,2$]\label{prop:consistency-M12}
    Suppose that \Cref{cond:conv-risk-M12} holds, then \ECV estimates defined in \eqref{eq:Roob-M-12} satisfy that for $M=1,2$,
    \begin{align*}
        \sup_{k\in\cK_n}\left| \hRoob_{M,k} - R_{M,k}\right| = \cO_p(\hsigma_n \log (n) n^{-1/2} +n^{\epsilon}(\gamma_{1,n}+\gamma_{2,n})),
    \end{align*}
    where $\hsigma_n= \max_{m,\ell\in[M_0],k\in\cK_n}\hsigma_{I_{k,\ell}\cup I_{k,m}}$ and $R_{M,k}=R_{M,k} (\tf_{M,k};\cD_n,\{I_{k,\ell}\}_{\ell=1}^M)$ for $\{I_{k,\ell}\}_{\ell=1}^M\overset{\SRS}{\sim}\cI_k$.
\end{proposition}
\begin{proof}[Proof of \Cref{prop:consistency-M12}.]
    Note that
    \begin{align*}
        \hRoob_{1,k}&=\frac{1}{M_0}\sum_{\ell=1}^{M_0} \hR(\hf(\cdot;\{\cD_{I_{k,\ell}}\}), \cD_{I_{k,\ell}^c})\\
        \hRoob_{2,k}&=\frac{1}{M_0(M_0-1)}\sum_{\substack{\ell,m\in[M_0]\\\ell\neq m}}\hR(\hf(\cdot;\{\cD_{I_{k,\ell}},\cD_{I_{k,m}}\}), \cD_{(I_{k,\ell}\cup I_{k,m})^c}).
    \end{align*}
    Define $\hsigma_I= \|(y_0- \hf(\bx_0;\cD_I))^2 \|_{\psi_1 \mid \cD_I}$ and recall $\Delta_{M,k}$ for $M=1,2$ are defined as follows:
    \begin{align}
            \begin{split}
             \Delta_{1,k} &= \hRoob_{1,k}-\frac{1}{M_0}\sum_{\ell=1}^{M_0}R_{1,k}(\hf; \cD_n, \{I_{k,\ell}\})     \\
            \Delta_{2,k} &= \hRoob_{2,k}-\frac{1}{M_0(M_0-1)}\sum_{\substack{\ell,m\in[M_0]\\\ell\neq m}}   R_{2,k}(\hf_{2,k};\cD_n,\{I_{k,\ell},I_{k,m}\}) .    
            \end{split}        \label{eq:diff-Rh-R}
        \end{align}   
    By the triangle inequality, for $M=1,2$, it holds that
    \begin{align}
        \sup_{k\in\cK_n}|\hRoob_{1,k} - \sR_{1,k}| &\leq \sup_{k\in\cK_n}|\Delta_{1,k}| ~~+ \sup_{k\in\cK_n,\ell\in[M_0]}|R_{1,k}(\hf; \cD_n, \{I_{k,\ell}\}) - \sR_{1,k}| \label{eq:cv-uniform-eq-0-2}\\
        \sup_{k\in\cK_n}|\hRoob_{2,k} - \sR_{1,k}| &\leq \sup_{k\in\cK_n}|\Delta_{2,k}| ~~+ \sup_{k\in\cK_n,m,\ell\in[M_0]}|R_{2,k}(\hf; \cD_n, \{I_{k,m},I_{k,\ell}\}) - \sR_{1,k}|  \label{eq:cv-uniform-eq-0-1}
    \end{align}
    where $\Delta_{M,k}$ are defined in \eqref{eq:diff-Rh-R}.
    Next we analyze each term separately for $M=1,2$.

    \noindent\underline{Part (1).} For the first term, from \eqref{eq:bound-Delta} in \Cref{prop:bounded-variance-error-control-mul-form} and applying union bound over $k\in\cK_n$, we have that
    \begin{align}
    \begin{split}
        \PP\left(\sup_{k\in\cK_n}|\Delta_{1,k}| \geq C \max_{k\in\cK_n,\ell\in[M_0]}\hsigma_{I_{k,\ell}}  \max\left\{\sqrt{\frac{\log(M_0|\cK_n|\eta)}{n/\log n}},  \frac{\log (M_0|\cK_n|\eta)}{n/\log n}\right\}\right) \leq\frac{1}{\eta}, \\
        \PP\left(\sup_{k\in\cK_n}|\Delta_{2,k}| \geq C \max_{k\in\cK_n,\ell\neq m\in[M_0]}\hsigma_{I_{k,\ell}\cup I_{k,m}} \max\left\{\sqrt{\frac{\log(M_0|\cK_n|\eta)}{n/\log n}},  \frac{\log (M_0|\cK_n|\eta)}{n/\log n}\right\}\right) \leq \frac{1}{\eta},
    \end{split}\label{eq:bound-Delta-k}
    \end{align}
    for some constant $C>0$.
    This implies that    
    \begin{align*}
         \sup_{k\in\cK_n}|\Delta_{1,k}| 
         =& \cO_p\left(\max_{k\in\cK_n,\ell\in[M_0]}\hsigma_{I_{k,\ell}} \sqrt{\frac{\log(|\cK_n|)\log n}{n}}\right)  \\
        \sup_{k\in\cK_n}|\Delta_{2,k} | 
         =&\cO_p\left(\max_{k\in\cK_n,\ell\neq m\in[M_0]}\hsigma_{I_{k,\ell}\cup I_{k,m}} \sqrt{\frac{\log(|\cK_n|)\log n}{n}}\right)  .
    \end{align*}
    To summarize, $\sup_{k\in\cK_n}|\Delta_{M,k}|$ for $M=1,2$ can be bounded as:
    \begin{align}
        \sup_{k\in\cK_n}|\Delta_{M,k}| = \cO_p\left( \hsigma_{n}\sqrt{\frac{\log^2 n}{n}}\right),\label{eq:cv-uniform-eq-1}
    \end{align} 
    where $\hsigma_n= \max_{m,\ell\in[M_0],k\in\cK_n}\hsigma_{I_{k,\ell}\cup I_{k,m}}$.

    \noindent\underline{Part (2).} For the second term, from \eqref{eq:subsample_cond_risk_M12} in \Cref{cond:conv-risk-M12} we have that when $n$ is large enough, for all $\eta\geq 1$,
    \begin{align*}
        \PP\left(n^{-\epsilon}\gamma_{1,n}^{-1}| R(\tf_{1,k};\mathcal{D}_{I_{k,\ell}}) - \sR_{1,k} |
        \geq \eta \right) &\leq \frac{C_0}{n\eta^{1/\epsilon}}.
    \end{align*}
    Let $\gamma_{1,n}'= n^{\epsilon}\gamma_{1,n}$.
    Taking a union bound over $k\in\cK_n$ and $\ell\in[M_0]$, we have        
    \begin{align*}
        \PP\left(\gamma_{1,n}'^{-1}\sup_{k\in\cK_n,\ell\in[M_0]}|R_{1,k}(\hf; \cD_n, \{I_{k,\ell}\}) - \sR_{1,k}|  \geq \eta\right) \leq \frac{C_0M_0}{\eta^{1/\epsilon}},
    \end{align*}
    which implies that
    \begin{align}
        \sup_{k\in\cK_n,\ell\in[M_0]}|R_{1,k}(\hf; \cD_n, \{I_{k,\ell}\}) - \sR_{1,k}| = \cO_p(n^{\epsilon}\gamma_{1,n}). \label{eq:cv-uniform-eq-2}
    \end{align}
    Analogously, for $M=2$, we also have
    \begin{align}
        \sup_{k\in\cK_n,m\neq \ell\in[M_0]}|R_{2,k}(\hf; \cD_n, \{I_{k,m},I_{k,\ell}\}) - \sR_{1,k}| = \cO_p(n^{\epsilon}\gamma_{2,n}). \label{eq:cv-uniform-eq-2-2}
    \end{align}

    \noindent\underline{Part (3).}
    Combining \eqref{eq:cv-uniform-eq-0-1}, \eqref{eq:cv-uniform-eq-1}, \eqref{eq:cv-uniform-eq-2} and \eqref{eq:cv-uniform-eq-2-2} yields that
    \begin{align*}
        \sup_{k\in\cK_n}|\hRoob_{M,k} - \sR_{M,k}| = \cO_p\left(\hsigma_n\sqrt{\frac{\log^2n}{n}} + n^{\epsilon}\gamma_{M,n}\right),\qquad M=1,2,
    \end{align*}
    ignoring constant factors.
\end{proof}

\subsection[Uniform consistency of cross-validated risk over (M,k)]{Uniform consistency of cross-validated risk to $\sR_{M,k}$ over $(M,k)$}\label{app:prop:consistency-M}
\begin{proposition}[Uniform consistency to $\sR_{M,k}$ over $(M,k)$]\label{prop:consistency-M}
     Suppose that \Cref{ass:kappa,cond:conv-risk-M12} hold, then \ECV estimates defined in \eqref{eq:Roob-M-12} satisfy that
    \begin{align*}
        \sup_{M\in\NN,k\in\cK_n}\left| \hRoob_{M,k} - \sR_{M,k}\right| &= \cO_p(\hsigma_n\log(n)n^{-\frac{1}{2}} + n^{\epsilon}(\gamma_{1,n}+ \gamma_{2,n})),
    \end{align*}
    where $\hsigma_n= \max_{m,\ell\in[M_0],k\in\cK_n}\hsigma_{I_{k,\ell}\cup I_{k,m}}$.
\end{proposition}
\begin{proof}[Proof of \Cref{prop:consistency-M}.]
    Note that $ \hRoob_{M,k}$ satisfies the squared risk decomposition and the randomness of $ \hRoob_{M,k}$ also due to both the full data $\cD_n$ and random sampling.

    From \Cref{prop:consistency-M12}, we have
    \begin{align*}
        \sup_{k\in\cK_n}|\hRoob_{M,k} - \sR_{M,k}| = \cO_p\left(\hsigma_n\sqrt{\frac{\log^2n}{n}} + n^{\epsilon}\gamma_{M,n}\right),\qquad M=1,2.
    \end{align*}

    On the other hand, by the definition of $ \hRoob_{M,k}$ for $M\in\NN$ in \eqref{eq:Roob-M}, taking the supremum over both $M$ and $k$ yields that
        \begin{align*}
            \sup_{M\in\NN,k\in\cK_n}\left|\hRoob_{M,k} - \sR_{M,k}\right|             
            &\leq \sup_{k\in\cK_n}\left|\hRoob_{1,k} - \sR_{1,k}\right| + 2\sup_{k\in\cK_n}\left|\hRoob_{2,k} - \sR_{2,k}\right|\\
            &=\cO_p\left(\zeta_n\right),
        \end{align*}
        where $\zeta_n=\hsigma_n\sqrt{\log^2n/n} + n^{\epsilon} (\gamma_{1,n}+ \gamma_{2,n})$.
\end{proof}

    \subsection[Proof of \Cref{lem:conv:data-risk-M12}]{Proof of \Cref{lem:conv:data-risk-M12} (Concentration of expected risk over $k\in\cK_n$)}
    \label{app:expected-deviation-risk}
    
        To obtain tail bounds for the subsample conditional risk defined in \eqref{eq:conditional-risk}, we need to analyze its conditional expectation.
        Here the expectation is taken with respect to only the randomness due to sampling and conditioned on $\cD_n$.
        For example, we define the data conditional (on $\cD_n$) risks as:
        \begin{align}
            \label{eq:unconditional-risk}        
            R(\tf_{M,k}(\cdot; \{\mathcal{D}_{I_\ell}\}_{\ell=1}^M) ; \, \cD_n) 
            &= 
            \int
            \EE\left[\left(
            y-\tf_{M, k}(\bx; \{ \cD_{I_{k,\ell}} \}_{\ell = 1}^{M})\right)^2 \mid 
            \mathcal{D}_n
            \right]
            \, \mathrm{d}P(\bx, y).
        \end{align}
        Observe that the conditional (on $\cD_n$) risk of the bagged predictor $\tf_{M, k}(\cdot; \{ \cD_{I_{k,\ell}}\}_{\ell = 1}^{M})$ integrates over the randomness of the future observation $(\bx, y)$ as well as the randomness due the simple random sampling of $I_{k,\ell}$, $\ell = 1, \dots, M$.
        Nevertheless, the subsample conditional risk ignores the expectation over the simple random sample.
        Considering a more general setup, when we need to obtain tail bounds for $\sup_{k\in\cK}|R(\tf_M;\cD_n,\{I_{k,\ell}\}_{\ell=1}^M) - \sR_{M,k}|$, we again need to control its expectation conditional on $\cD_n$.
        The result is summarized as in the following lemma.
        Note that when $|\cK_n|=1$, it simply reduces to controlling the data conditional risk.
        
    \begin{lemma}[Concentration of expected risk]\label{lem:conv:data-risk-M12}
        Consider a dataset $\cD_n$ with $n$ observations, a subsample grid $\cK_n\subset[n]$, and a base predictor $\hf$.
        Suppose \Cref{ass:kappa,cond:conv-risk-M12} hold, then it holds for $M=1,2$ that
        \begin{align*}
            \EE\left(\sup_{k\in\cK_n} |R(\tf_M;\cD_n,\{I_{k,\ell}\}_{\ell=1}^M) - \sR_{M,k}| \right)= \cO(n^{\epsilon}\gamma_{M,n}) ,
        \end{align*}
        where $\{I_{k,\ell}\}_{\ell=1}^M\overset{\SRS}{\sim}\cI_k$.
    \end{lemma}
    \begin{proof}[Proof of \Cref{lem:conv:data-risk-M12}.]
        Define 
        \begin{align}
            B_{1,n}&=
            \sup_{k\in\cK_n}|R(\tf_{1,k}; \cD_n, \{I_{k,1}\}) - \sR_{1,k}|.
            \label{eq:sup-conv-M1}\\
            B_{2,n}&=\sup_{k\in\cK_n}|R(\tf_{2,k}; \cD_n, \{I_{k,1},I_{k,2}\}) - \sR_{2,k}|.\label{eq:sup-conv-M2}
        \end{align}
        
        We start with the $M = 1$ case by bounding the expectation of \eqref{eq:sup-conv-M1}.
        Note that from \eqref{eq:subsample_cond_risk_M12}, we have that for all $\eta\geq \eta_0$,
        \begin{align*}
            \PP\left(n^{-\epsilon}\gamma_{1,n}^{-1}| R(\tf_{1,k};\mathcal{D}_{I_{1}}) - \sR_{1,k} |
            \geq \eta \right) &\leq \frac{1}{n\eta^{1/\epsilon}}
        \end{align*}
        Let $\gamma_{1,n}'= n^{\epsilon}\gamma_{1,n}$.
        Taking a union bound over $k\in\cK_n$, we have        
        \begin{align}
            \PP(\gamma_{1,n}'^{-1}B_{1,n} \geq \eta) \leq \frac{1}{\eta^{1/\epsilon}}. \label{eq:B1-tailbound}
        \end{align}        
        Then, it follows that
        \begin{align*}
            \EE(B_{1,n}) &= \gamma_{1,n}'\int_0^{\infty}\PP(\gamma_{1,n}'^{-1}B_{1,n}\geq \eta) \rd \eta \\
            &\leq \gamma_{1,n}'\left(  \int_0^{\eta_0}1 \rd \epsilon + \int_{\eta_0}^{\infty}\frac{C_0}{\eta^{1/\epsilon}} \rd \epsilon\right)\\
            &\leq\gamma_{1,n}'\left( \eta_0 - C_0\frac{\epsilon-1}{\epsilon}\eta^{1-1/\epsilon}\big|_{\eta_0}^{\infty}\right)\\
            &= \left(\eta_0+ C_0\frac{\epsilon-1}{\epsilon}\eta_0^{1-1/\epsilon}\right)\gamma_{1,n}'.
        \end{align*}
        The proof for $M=2$ follows analogously.
    \end{proof}

    \subsection[Proof of \Cref{lem:conv:risk-M}]{Proof of \Cref{lem:conv:risk-M} (Concentration of conditional risk)}
    \label{app:tail-bound-deviation-risk}
    
    \begin{lemma}[Concentration of conditional risk]\label{lem:conv:risk-M}
        Consider a dataset $\cD_n$ with $n$ observations and a base predictor $\hf$.
        Under \Cref{cond:conv-risk-M12}, it holds that,
        \begin{align}
            \sup_{M\in\NN,k\in\cK_n} |R_{M,k} - \sR_{M,k} |&= \cO_p
            \left(
            C_{1}(|\cK_n|)\gamma_{1,n}+
            C_{2}(|\cK_n|)\gamma_{1,n}
            \right),\label{eq:risk-extrapolation}
        \end{align}
        where $\sR_{M,k}=2\sR_{2,k}-\sR_{1,k}+2(\sR_{1,k}-\sR_{2,k})/M$.
    \end{lemma}
    \begin{proof}[Proof of \Cref{lem:conv:risk-M}.]
        By \Cref{prop:squared_risk_decom}, we have for $\{I_{k,\ell}\}_{\ell=1}^M\overset{\SRS}{\sim}\cI_k$
        \begin{align*}
                &\left|R(\tf_{M,k};\mathcal{D}_n,\{I_{k,\ell}\}_{\ell=1}^M) - \sR_{M,k}\right|\\
                &=\left|R(\tf_{M,k};\mathcal{D}_n,\{I_{k,\ell}\}_{\ell=1}^M) - \left[(2\sR_{2,k} - \sR_{1,k}) + \frac{2(\sR_{1,k} - \sR_{2,k})}{M}\right] \right|\\
                &= \left|- \left(\frac{1}{M}-\frac{2}{M^2}\right)\sum_{\ell=1}^M \left(R(\tf_{1,k}; \, \cD_n,\{I_{k,\ell}\}) - \sR_{1,k}  \right)\right|
                \\
                &\qquad + \left|\frac{2}{M^2}\sum\limits_{\substack{i,j\in[M]\\i\neq j}}  \left(R(\tf_{2,k} ; \, \cD_n, \{I_{k,i},I_{k,j}\}) - \sR_{2,k}\right)\right|.
        \end{align*}
        Then, it follows that
        \begin{align}
            &\sup_{M\in\NN,k\in\cK_n}\left|R(\tf_{M,k};\mathcal{D}_n,\{I_{k,\ell}\}_{\ell=1}^M) - \sR_{M,k}\right| \notag\\
            &\leq \left|1-\frac{2}{M}\right|\sup_{M\geq 1}\left|\frac{1}{M}\sum_{\ell\in[M]} \sup_{k\in\cK_n}|R(\tf_{1,k}; \, \cD_n,\{I_{k,\ell}\}) - \sR_{1,k}|\right| \notag\\
            &\qquad + 2 \sup_{M\geq 2}\left| \frac{1}{M(M-1)}\sum_{i,j\in[M],i\neq j}\sup_{k\in\cK_n}|R(\tf_{2,k} ; \, \cD_n, \{I_{i},I_j\}) - \sR_{2,k}|\right|\notag\\
            \leq &\sup_{M\geq 1}\left|\frac{1}{M}\sum_{\ell\in[M]} \sup_{k\in\cK_n}|R(\tf_{1,k}; \, \cD_n,\{I_{k,\ell}\}) - \sR_{1,k}|\right| \notag\\
            &\qquad + 2 \sup_{M\geq 2}\left| \frac{1}{M(M-1)}\sum_{i,j\in[M],i\neq j}\sup_{k\in\cK_n}|R(\tf_{2,k} ; \, \cD_n, \{I_{k,i},I_{k,j}\}) - \sR_{2,k}|\right|.\label{eq:sup-M12}
        \end{align}
        We start by observing that the two terms
        \begin{align*}
            U_M &= \frac{1}{M}\sum_{\ell\in[M]} \sup_{k\in\cK_n}|R(\tf_{1,k}; \cD_n, \{I_{k,\ell}\}) - \sR_{1,k}|,\\ 
            U_M' &= \frac{1}{M(M-1)}\sum_{i,j\in[M],i\neq j} \sup_{k\in\cK_n}|R(\tf_{2,k}; \cD_n, \{I_{k,i}, I_{k,j}\}) - \sR_{2,k}|
        \end{align*}
        are $U$-statistics based on sample $\cD_{I_1}, \ldots, \cD_{I_M}$.
        Theorem 2 in Section 3.4.2 of~\cite{lee2019u} implies that $\{U_M\}_{M\ge1}$ and $\{U_M'\}_{M\ge2}$ are a reverse martingale conditional on $\cD_n$ with respect to some filtration.
        This combined with Theorem 3 (maximal inequality for reverse martingales) in Section 3.4.1 of~\cite{lee2019u} (for $r = 1$) yields
        \begin{align*}
        \mathbb{P}\left(\sup_{M\ge1}|U_M| \ge \delta \right) &\le \frac{1}{\delta}\mathbb{E}\left[|U_1|\right] = \frac{1}{\delta}\mathbb{E}\left[\sup_{k\in\cK_n}| R(\tf_{1,k}; \cD_n, \{I_{k,1}\}) - \sR_{1,k}| \right].
        \end{align*}
        On the other hand, from \Cref{lem:conv:data-risk-M12}, the expectations are bounded as 
        \begin{align*}
            \mathbb{E}\left[\sup_{k\in\cK_n}\left| R(\tf_{1,k}; \cD_n, \{I_{k,1}\}) - \sR_{1,k} \right| \right] = \cO(n^{\epsilon}\gamma_{1,n}).
        \end{align*}
        It follows that
        \begin{align*}
            \sup_{M\ge1}|U_M| &= \cO_p(n^{\epsilon}\gamma_{1,n}).
        \end{align*}
        Analogously, for the second $U$-statistic we also have
        \begin{align*}
            \sup_{M\ge2}|U_M'| &= \cO_p(n^{\epsilon}\gamma_{2,n}).
        \end{align*}
        From \eqref{eq:sup-M12}, it follows that
        \begin{align}
            &\sup_{M\in\NN,k\in\cK_n} |R(\tf_M;\cD_n,\{{I_{k,\ell}}\}_{\ell=1}^M) - \sR_{M,k} | \notag\\
            &= \sup_{M\ge1}|U_M| + 2\sup_{M\ge2}|U_M'|\notag\\
            &= \cO_p(n^{\epsilon}(\gamma_{1,n}+\gamma_{2,n})). \label{eq:sup-R-Xinstar}
        \end{align}
        This completes the proof.
    \end{proof}

\section{Proofs of results in \Cref{subsec:risk-extrapolation}}\label{app:extrapolation}

\subsection[Proof of \Cref{thm:limiting-oobcv-for-arbitrary-M-cond}]{Proof of \Cref{thm:limiting-oobcv-for-arbitrary-M-cond} ($\delta$-optimality of \ECV)}\label{app:thm:limiting-oobcv-for-arbitrary-M-cond}

\begin{proof}[Proof of \Cref{thm:limiting-oobcv-for-arbitrary-M-cond}.]
    For simplicity, we denote $R(\tf_{M,k}; \mathcal{D}_n, \{ I_{k,\ell} \}_{\ell = 1}^{M})$ by $R_{M,k}$ as a function of $M$ and $k$.
    We split the proof for the two parts below.

    \noindent\underline{Part (1) Error bound on the estimated risk.}
    From \Cref{lem:consistency-oobcv}, we have
    \begin{align*}
         \sup_{M\in\NN,k\in\cK_n} | \hRoob_{M,k} - R_{M,k} |=\cO_p(\zeta_n) ,
    \end{align*}
    where $\zeta_n=\hsigma_n \log n /n^{1/2}+ n^{\epsilon}(\gamma_{1,n}+ \gamma_{2,n})$ and $\hsigma_n= \max_{m,\ell\in[M_0],k\in\cK_n}\hsigma_{I_{k,\ell}\cup I_{k,m}}$.
    Then the conditional risk of the \ECV-tuned predictor $R_{\hat{M},\hat{k}}$ admits 
    \begin{align*}
        | \hRoob_{\hat{M},\hat{k}} - R_{\hat{M},\hat{k}} | \leq  \sup_{M\in\NN,k\in\cK_n}  | \hRoob_{M,k} - R_{M,k} |=\cO_p(\zeta_n)  .
    \end{align*}

    \noindent\underline{Part (2) Additive suboptimality.}
    The proof proceeds in two steps.

    \noindent\textbf{Step 1: Bounding the difference between
    $\hRoob_{\infty}(\tf_{M_0, \widehat{k}})$ and the oracle-tuned risk.}
    Let 
    \begin{align*}
        (M^*,k^*)\in\arginf_{M\in\NN,k\in\cK_n} R(\tf_{M,k}; \mathcal{D}_n, \{ I_{\ell} \}_{\ell = 1}^{M}),    
    \end{align*}
    which is a tuple of random variables and also functions of $n$.
    For any $k\in \cK_n$, by the risk decomposition \eqref{eq:risk-decomp-M} we have that
    \begin{align*}
        \inf_{M\in\NN} R_{M,k} = R_{\infty,k}.
    \end{align*}
    That is, $M^*=\infty$ is one minimizer for any $k\in\cK_n$.
    Then it follows that
    \begin{align}
         \hRoob_{\infty,\hat{k}}=& \inf_{k\in\cK_n}R_{\infty,k} + \cO_p(\zeta_n) = \inf_{M\in\NN,k\in\cK_n}R_{M,k} + \cO_p(\zeta_n)  \label{eq:lem:risk-minimization-delta-eq-0}
    \end{align}
     where the first equality is due to \Cref{lem:consistency-oobcv}.

    \noindent\textbf{Step 2: $\delta$ optimality.}
    Next, we bound the suboptimality by the triangle inequality:
    \begin{align}
        &| R_{\hat{M},\hat{k}} - \inf_{M\in\NN,k\in\cK_n} R_{M,k} | \notag\\
        &= | R_{\hat{M},\hat{k}} - \hRoob_{\hat{M},\hat{k}} + \hRoob_{\hat{M},\hat{k}} - \hRoob_{\infty,\hat{k}} + \hRoob_{\infty,\hat{k}} - R_{\infty,k^*}| \notag\\
        &\leq|R_{\hat{M},\hat{k}} - \hRoob_{\hat{M},\hat{k}} |+ | \hRoob_{\hat{M},\hat{k}} - \hRoob_{\infty,\hat{k}} |  + | \hRoob_{\infty,\hat{k}} - R_{\infty,k^*} |. \label{eq:lem:risk-minimization-delta-eq-1}
    \end{align}
    From Part (1) we know that the first term in \eqref{eq:lem:risk-minimization-delta-eq-1} can be bounded as $| R_{\hat{M},\hat{k}} - \hRoob_{\hat{M},\hat{k}} |=\cO_p(\zeta_n)$.
    From \eqref{eq:lem:risk-minimization-delta-eq-0} we know that the last term in \eqref{eq:lem:risk-minimization-delta-eq-1} can also be bounded as $|  \hRoob_{\infty,\hat{k}} - R_{\infty,k^*} |=\cO_p(\zeta_n)$.
    It remains to bound the second term in \eqref{eq:lem:risk-minimization-delta-eq-1} .
    By the definition of $\hRoob_{M,k}$ in \eqref{eq:Roob-M}, we have that
    \begin{align*}
        \hRoob_{\infty, k} = 2\hRoob_{2,k} - \hRoob_{1,k}.
    \end{align*}
    Since $\hat{M}= \lceil 2/\delta\cdot\hRoob_{1,\hat{k}} - \hRoob_{2,\hat{k}} \rceil \geq 2/\delta\cdot\hRoob_{1,\hat{k}} - \hRoob_{2,\hat{k}}$, the second term in \eqref{eq:lem:risk-minimization-delta-eq-1} is bounded by
    \begin{align}
        |\hRoob_{\hat{M},\hat{k}} - \hRoob_{\infty,\hat{k}} | &= \frac{2}{\hat{M}}\left|\hRoob_{1,\hat{k}} - \hRoob_{2,\hat{k}}\right|\leq\delta. \label{eq:lem:risk-minimization-delta-gt-0}
    \end{align}
    Therefore, the $\delta$-optimality conclusion on $\hRoob_{M,k}$ follows.
\end{proof}

\subsection{Proof of multiplicative optimality}\label{app:subsec:mul-opt}
    \begin{proposition}[Multiplicative optimality]\label{prop:mul-opt}
        Under the same conditions in \Cref{thm:limiting-oobcv-for-arbitrary-M-cond}, if $\int\{y-\EE(y\mid \bx)\}^2\rd P(\bx,y)$ is lower bounded away from zero and $(\hat{M},\hat{k})$ is defined for relative optimality such that 
        \begin{align}
        \hRoob_{\hat{M},\hat{k}} \leq (1+\delta) \inf_{M\in\NN,k\in\cK_n}\hRoob_{M,k},\label{eq:Rh-mul}    
        \end{align} 
        then it holds that
        \begin{align*}
          R_{\hat{M},\hat{k}} & \leq (1+\delta)  \inf_{M\in\NN,k\in\cK_n} R_{M,k} (1 + \cO_p(\zeta_n)).   
        \end{align*}
    \end{proposition}
    \begin{proof}[Proof of \Cref{prop:mul-opt}.]        
    By definition of \eqref{eq:Rh-mul}, we have
    \begin{align*}
        \hRoob_{\hat{M},\hat{k}} &=(1+\delta) \inf_{M\in\NN,k\in\cK_n}\hRoob_{M,k}\\
        &\leq (1+\delta)  \left(\inf_{M\in\NN,k\in\cK_n}R_{M,k} + \cO_p(\zeta_n)\right)\\
        &= (1+\delta)  \inf_{M\in\NN,k\in\cK_n}  R_{M,k}(1 + \cO_p(\zeta_n)).
    \end{align*}
    where the inequality is from \Cref{thm:limiting-oobcv-for-arbitrary-M-cond} and the last equality is from the assumption that the risks are lower bounded.
    Further, since from \Cref{lem:consistency-oobcv}, $\sup_{M\in\NN,k\in\cK_n}|R_{M,k} - \hRoob_{M,k}|=\cO_p(\zeta_n)$, we have
    \begin{align*}
        R_{\hat{M},\hat{k}} & \leq \hRoob_{\hat{M},\hat{k}} + \cO_p(\zeta_n) \leq (1+\delta)  \inf_{M\in\NN,k\in\cK_n} R_{M,k} (1 + \cO_p(\zeta_n)).
    \end{align*}
    This finishes the proof.
    \end{proof}

\subsection{Concrete example: bagged ridge predictors}\label{app:ridge}

    Application of \Cref{thm:limiting-oobcv-for-arbitrary-M-cond} to a specific data model and a base predictor requires verification of \Cref{ass:kappa,cond:conv-risk-M12}.
    As an illustration, we verify all the conditions for ridge predictors.
    Consider a dataset $\mathcal{D}_n = \{(\bx_1, y_1), \ldots, (\bx_n, y_n)\}$ consisting of random vectors in $\RR^{p} \times \RR$.
    Let $\bX\in\RR^{n\times p}$ denote the corresponding feature matrix whose $j$-th row contains $\bx_j^\top$, and let $\by\in\RR^n$ denote the corresponding response vector whose $j$-th entry contains $y_j$.
    Recall that the \emph{ridge} estimator with regularization parameter $\lambda>0$ fitted on $\cD_I$ for $I\subseteq[n]$ is defined as $\betaridge(\cD_I) = \argmin\limits_{\bbeta\in\RR^p}
    \sum_{j \in I} (y_j  - \bx_j^\top \bbeta)^2 / | I | + \lambda \| \bbeta \|_2^2.$
    The associated ridge base predictor and the subagged predictor are given by $\hat{f}_{\lambda}(\bx;\cD_I) = \bx^{\top}\betaridge(\cD_I)$ and $\tfWR{M}{k}(\bx;\cD_n) = \bx^{\top}\tbetaridge{M}$, where $I\in\cI_k$ and $\{I_{\ell}\}_{\ell=1}^M\overset{\SRS}{\sim}\cI_k$.

We consider Assumptions~\ref{asm:rmt-feat}-\ref{asm:lin-mod} on the dataset $\cD_n$ to characterize the risk, which are standard in the study of the ridge and ridgeless regression under proportional asymptotics; see, e.g., \citet{hastie2022surprises,patil2022bagging}.
   
    \begin{assumption}[Feature model]\label{asm:rmt-feat}
        The feature vectors $\bx_i \in \RR^{p}$, $i = 1, \dots, n$, multiplicatively decompose as $\bx_i = \bSigma^{1/2} \bz_i$, where $\bSigma \in \RR^{p \times p}$ is a positive semi-definite matrix and $\bz_i \in \RR^{p}$ is a random vector containing i.i.d.\ entries with mean $0$, variance $1$, and bounded $k$th moment for $k\geq2$.
        Let $\bSigma = \bW \bR \bW^\top$ denote the eigenvalue decomposition of the covariance matrix $\bSigma$, where $\bR \in \RR^{p \times p}$ is a diagonal matrix containing eigenvalues (in non-increasing order) $r_1 \ge r_2 \ge \dots \ge r_{p} \ge 0$, and $\bW~\in~\RR^{p \times p}$ is an orthonormal matrix containing the associated eigenvectors $\bw_1, \bw_2, \dots, \bw_{p}~\in~\RR^{p}$. 
        Let $H_{p}$ denote the empirical spectral distribution of $\bSigma$ (supposed on $\RR_{> 0}$) whose value at any $r \in \RR$ is given by
        \[
            H_{p}(r)
            = \frac{1}{p} \sum_{i=1}^{p} \ind_{\{r_i \le r\}}.
        \]
        Assume there exists $0 < r_{\min} \le r_{\max} < \infty$ such that $r_{\min}\leq r_1\leq r_p\leq r_{\max}$, and there exists a fixed distribution $H$ such that $H_{p} \rightarrow H$ in distribution as $p  \to \infty$.
    \end{assumption}
    
   \begin{assumption}[Response model]\label{asm:lin-mod}
        The response variables $y_i \in \RR$, $i = 1, \dots, n$, additively decompose as $y_i = \bx_i^\top \bbeta_0 + \varepsilon_i$, where $\bbeta_0 \in \RR^{p}$ is an unknown signal vector and $\varepsilon_i$ is an unobserved error that is assumed to be independent of $\bx_i$ with mean $0$, variance $\sigma^2$, and bounded moment of order $4 + \delta$ for some $\delta > 0$.
        The $\ell_2$-norm of the signal vector $\| \bbeta_0 \|_2$ is uniformly bounded in $p$, and $\lim_{p \to \infty} \| \bbeta_0 \|_2^2 = \rho^2 < \infty$.
        Let $G_{p}$ denote a certain distribution (supported on $\RR_{> 0}$) that encodes the components of the signal vector $\bbeta_0$ in the eigenbasis of $\bSigma$ via the distribution of (squared) projection of $\bbeta_0$ along the eigenvectors $\bw_j, 1 \le j \le p$, whose value at any $r \in \RR$ is given by
        \[
            G_{p}(r)
            = \frac{1}{\| \bbeta_0 \|_2^2} \sum_{i = 1}^{p} (\bbeta_0^\top \bw_i)^2 \, \ind_{\{ r_i \le r \}}.
        \]
        Assume there exists a fixed distribution $G$ such that $G_{p} \rightarrow G$ in distribution as $p  \to \infty$.
   \end{assumption}

    Assumptions~\ref{asm:rmt-feat}-\ref{asm:lin-mod} provide risk characterization for subagged ridge predictors \citep{patil2022bagging}, which establishes the consistency of \splitcv for any fixed ensemble size $M$, without obtaining the convergence rate.
        In \Cref{cor:ridge}, Assumptions~\ref{asm:rmt-feat}-\ref{asm:lin-mod} assume a linear model $y_i=\bx_i^{\top}\bbeta_0+\epsilon_i$ for observation $i$, where the feature vector is generated by $\bx_i=\bSigma^{1/2}\bz_i$, $\bSigma$ is the covariance matrix and $\bz_i$ contains i.i.d. entries with zero mean and unit variance.
        We make mild assumptions about such a data-generating process:
        (1) the noise $\epsilon_i$ and the entries of $\bz_i$ have bounded moments, and
        (2) the covariance and signal-weighted spectrums converge weakly to some distributions as $n,p\rightarrow\infty$.

    Under Assumptions~\ref{asm:rmt-feat}-\ref{asm:lin-mod}, we will show that $\hsigma_n=o(1)$ when using with \MOM.
    To prove the result for \AVG, we need the modified assumptions by replacing the bounded moment conditions on $z_{ij}$ in \Cref{asm:rmt-feat} and $\epsilon_i$ in \Cref{asm:rmt-feat} by bounded $\psi_2$-norm conditions.

    We will analyze the bagged predictors (with $M$ bags) in the proportional asymptotics regime, where the original data aspect ratio ($p / n$) converges to $\phi \in (0, \infty)$ as $n, p \to \infty$, and the subsample data aspect ratio ($p/k$) converges to $\phi_s$ as $k, p \to \infty$. 
    Because $k \le n$, $\phi_s$ is always no less than $\phi$.

       Under these assumptions, the results for ridge predictors are summarized in \Cref{cor:ridge}.
       
        \begin{proposition}[\ECV for ridge predictors]\label{cor:ridge}
            Suppose Assumptions~\ref{asm:rmt-feat}-\ref{asm:lin-mod} in \Cref{app:ridge} hold.
            Then, the ridge predictors with $\lambda>0$ using subagging satisfy \Cref{ass:kappa} with $\hsigma_n=\cO_p(1)$ for \MOM and \Cref{cond:conv-risk-M12} for any $\psi\in(0,1)$ such that $p/n\to\phi\in[\psi,\psi^{-1}]$ and $p/k\to\phi_s\in[\phi, \psi^{-1}]$ as $k,n,p\rightarrow\infty$.
            Consequently, the conclusions in \Cref{thm:limiting-oobcv-for-arbitrary-M-cond} hold.
        \end{proposition}
        
        We remark that \Cref{cor:ridge} verifies \Cref{thm:limiting-oobcv-for-arbitrary-M-cond} for \CEN=\MOM, but one can also verify for \CEN=\AVG under slightly different assumptions; see \Cref{app:ridge} for more details.
        It is worth mentioning that \Cref{cond:conv-risk-M12} only requires the conditional risks for the ridge ensemble with $M = 1,2$ converge to their respective conditional (on $\cD_n$) limits, while the limiting forms of them are not required.
        In this regard, Assumptions \ref{asm:rmt-feat}-\ref{asm:lin-mod} can be further relaxed with some efforts, but we do not pursue fine-tuning of assumptions as our intent is only to illustrate the end-to-end applicability of \Cref{thm:limiting-oobcv-for-arbitrary-M-cond}.
        The generality of \Cref{thm:limiting-oobcv-for-arbitrary-M-cond} to general predictors is illustrated empirically in the next section.

    \begin{proof}[Proof of \Cref{cor:ridge}.]
        We split the proof into different parts.
        
        \noindent\underline{Part (1) Bounded moments of the risk for $M=1$.}
        For $I\in\cI_k$, define $\bL_I$ be the diagnoal matrix whose $i$th diagonal entry is one if $i\in I$ and zero otherwise, let $\hSigma_I=\bX^{\top}\bL_I\bX/|I|$.
        We begin with analyzing the risk for $M=1$:
        \begin{align}
            &R(\tf_{1,k};\mathcal{D}_I) \notag\\
            &= \EE\{(y_0-\bx_0^{\top}\bbeta_0)^2\} + \|\hbeta_1(\cD_I) - \bbeta_0\|_{\bSigma^{1/2}}^2 \notag\\
            &= \EE\{(y_0-\bx_0^{\top}\bbeta_0)^2\} + 
            \bbeta_0^{\top}\hSigma_I(\hSigma_I+\lambda\bI_p)^{-2}\hSigma_I\bbeta_0 +
            \frac{1}{k^2}\bvarepsilon^{\top}\bL_I\bX(\hSigma_I+\lambda\bI_p)^{-2}\bX^{\top}\bL_I\bvarepsilon    .  \label{eq:ridge-risk-eq-1}
        \end{align}
        Note that the first term is just a constant, and the last two terms can be bounded as:
        \begin{align*}
            \bbeta_0^{\top}\hSigma_I(\hSigma_I+\lambda\bI_p)^{-2}\hSigma_I\bbeta_0 & \leq \|\hSigma_I(\hSigma_I+\lambda\bI_p)^{-2}\hSigma_I\|_{\oper}\|\bbeta_0\|_2^2 \\
            &\leq \rho^2\max_{j\in[p]}\frac{s_j(\hSigma_I)^2}{(\lambda+s_j(\hSigma_I))^2}\\
            &\leq \rho^2\\
            \frac{1}{k^2}\bvarepsilon^{\top}\bL_I\bX(\hSigma_I+\lambda\bI_p)^{-2}\bX^{\top}\bL_I\bvarepsilon  &\leq \frac{1}{k^2}\|\bL_I\bX(\hSigma_I+\lambda\bI_p)^{-2}\bX^{\top}\bL_I\|_{\oper}\|\bL_I\bvarepsilon\|_2^2 \\
            &\leq \frac{\|\bL_I\bvarepsilon\|_2^2}{k} \max_{j\in[p]}\frac{s_j(\hSigma_I)}{(\lambda+s_j(\hSigma_I))^2}\\
            &\leq \frac{\|\bL_{I}\bvarepsilon\|_2^2}{k\lambda},
        \end{align*}
        where $s_j(\hSigma_I)\geq 0$ is the $j$th eigenvalue of $\hSigma_I$.
        For all $q\leq 2+\delta/2$, it follows that
        \begin{align*}
            \|R(\tf_{1,k};\mathcal{D}_I)\|_{L_q} &\leq \EE\{(y_0-\bx_0^{\top}\bbeta_0)^2\} + \rho^2 + \frac{1}{k\lambda}\|\|\bL_{I}\bvarepsilon\|_2^2\|_{L_q}\\
            &\leq \EE\{(y_0-\bx_0^{\top}\bbeta_0)^2\} + \rho^2 + \frac{1}{\lambda}\max_{j\in[n]}\|\epsilon_j^2\|_{L_q},
        \end{align*}
        which we denote as $C_{0,q}$.
        From the bounded moment assumption \ref{asm:lin-mod}, we know that $\|\epsilon_j^2\|_{L_q}<\infty$ for all $j$.
        Thus, $\|R(\tf_{1,k};\mathcal{D}_I)\|_{L_q}$ is upper bounded by constant $C_{0,q}$ for all $n\in\NN$.
        
        \noindent\underline{Part (2) Uniform integrability of the risk for $M=1$.}
        Under Assumptions \ref{asm:rmt-feat}-\ref{asm:lin-mod}, from \citet[Theorem 4.1]{patil2022bagging} we have that there exist deterministic functions $\sR_1$ and $\sR_{2}$ such that, for all $I\in\cI_k$ and $\{I_1,I_2\}\overset{\SRS}{\sim}\cI_k$,
        \begin{align}
            R(\tf_{1,k};\mathcal{D}_I) - \sR_1(\phi,\phi_s) &\rightarrow 0,\qquad 
            R(\tf_{2,k};\cD_n,\{{I_1},{I_{2}}\}) - \sR_2(\phi,\phi_s) \rightarrow 0, \label{eq:ridge-eq-conv}
        \end{align}
        with probability tending to one, as $k,n,p\rightarrow\infty$, $p/n\rightarrow\phi\in[\psi,\psi^{-1}]$ and $p/k\rightarrow\phi_s\in[\phi,\psi^{-1}]$.
        Furthermore, $\sR_1$ and $\sR_2$ are continuous functions on $\phi_s$ for any $\phi$, which are bounded in the domain.
        To verify the tail bound condition for $M=1$, \citet[Theorem 5]{hastie2022surprises} shows that 
        \begin{align*}
            \PP(n^{(1-\epsilon')/2} |R(\tf_{1,k};\mathcal{D}_I) - \sR_1(\phi,\phi_s)|\geq C_1)\leq Cn^{-D},
        \end{align*}
        where $
        C_1=C(\rho^2+\lambda^{-1})\lambda^{-1}$ for any $D,\epsilon'>0$ and the constant $C$ depends on $\epsilon'$, $D$ and other model parameters.
        For all $q\leq 1+\delta/4$, let $\gamma_{1,n}=n^{-(1-\epsilon')/(2q)}$ and $Z_n=\gamma_{1,n}^{-1}|R(\tf_{1,k};\mathcal{D}_I) - \sR_1(\phi,\phi_s)|$. Then, we have $\PP(Z_n^q\geq C_1^q)\leq 1-Cn^{-D}$, and
        \begin{align*}
            \EE(Z_n^q) &= \EE(Z_n^q\ind\{Z_n<C_1\}) + \EE(Z_n^q\ind\{Z_n\geq C_1\}) \\
            &\leq C_1^q  + \{\EE(Z_n^{2q})\}^{\frac{1}{2}} \PP(Z_n\geq C_1)^{\frac{1}{2}}\\
            &\leq C_1^q +  (C_{0,2q}^{q}+\sR_1(\phi,\phi_s)^{\frac{q}{2}})\gamma_{1,n}^{-q}2^{q-1}C^{\frac{1}{2}} n^{-\frac{D}{2}},
        \end{align*}
        where the first inequality is from Holder's inequality.
        For $\epsilon'\in(0,1)$ fixed, setting $D=(1-\epsilon')/2$ yields that
        \begin{align*}
            \EE(Z_n^q)  &\leq C_1^q +  (C_{0,2q}^{q}+\sR_1(\phi,\phi_s)^{\frac{q}{2}}) 2^{q-1}C^{\frac{1}{2}},
        \end{align*}
        Therefore, $Z_n^q$ is uniformly integrable.
        By Chebyshev's inequality, we further have
        \begin{align}
            \limsup_{n\rightarrow\infty}\sup_{\eta >0 } \eta^q\PP(Z_n \geq \eta) &\leq C_1^q +  (C_{0,2q}^{q}+\sR_1(\phi,\phi_s)^{\frac{q}{2}}) 2^{q-1}C^{\frac{1}{2}},\label{eq:ridge-M1-tail}
        \end{align}
        which implies \Cref{cond:conv-risk-M12} for $\gamma_{1,n} = n^{-(1-\epsilon')/2}$ and for any $\epsilon=1/q$ where $q\leq 1+\delta/4$.

        \noindent\underline{Part (3) Concentration of the risk for $M=2$.}
        To show the existence of $\gamma_{2,n}$, from \citet[Lemma 3.8.]{patil2022bagging}, we have that with probability tending to one, 
        \begin{align*}
            R(\tf_{2,k};\cD_n,\{{I_1},{I_{2}}\}) - \sR_2(\phi,\phi_s) \rightarrow 0. 
        \end{align*}
        By convexity, we have
        \begin{align*}
            0\leq R(\tf_{2,k};\cD_n,\{{I_1},{I_{2}}\}) \leq 2^{-1}(R(\tf_{1,k};\cD_n,\{{I_1}\}) + R(\tf_{1,k};\cD_n,\{{I_2}\})),
        \end{align*}
        which implies that
        \begin{align}
            0\leq R(\tf_{2,k};\cD_n,\{{I_1},{I_{2}}\})^{1/\epsilon} \leq 2^{-1/\epsilon}(R(\tf_{1,k};\cD_n,\{{I_1}\}) + R(\tf_{1,k};\cD_n,\{{I_2}\}))^{1/\epsilon}.\label{eq:ridge-eq-2}
        \end{align}
        We next apply Pratt's lemma to show $L^{1/\epsilon}$-convergence for $M=2$.
        Note that the tail bound condition \eqref{eq:ridge-M1-tail} implies $L^{1/\epsilon}$-convergence:
        \begin{align*}
            \gamma_{1,n}^{-1/\epsilon}\EE\{(R(\tf_{1,k};\cD_n,\{{I_j}\}) - \sR_1(\phi,\phi_s) )^{1/\epsilon}\}\rightarrow 0,\qquad j=1,2
        \end{align*}
        and the $1/\epsilon$-th moment converges:
        \begin{align*}
            \gamma_{1,n}^{-1/\epsilon}\EE(R(\tf_{1,k};\cD_n,\{{I_j}\})^{1/\epsilon} - \sR_1(\phi,\phi_s) ^{1/\epsilon})\rightarrow 0,\qquad j=1,2.
        \end{align*}
        By Pratt's lemma \citep[see, e.g.,][Theorem  5.5]{gut_2005}, we have the $1/\epsilon$-th moment for $M=2$ also converges:
        \begin{align*}
            \gamma_{1,n}^{-1/\epsilon}\EE(R(\tf_{2,k};\cD_n,\{{I_1},{I_{2}}\})^{1/\epsilon} - \sR_2(\phi,\phi_s)^{1/\epsilon}) \to 0.
        \end{align*}
        From \citet[Theorem 5.2]{gut_2005}, we further have $L^{1/\epsilon}$-convergence for $M=2$:
        \begin{align*}
            \gamma_{1,n}^{-1/\epsilon}\EE\{(R(\tf_{2,k};\cD_n,\{{I_1},{I_{2}}\}) - \sR_2(\phi,\phi_s))^{1/\epsilon}\} \to 0.
        \end{align*}
        This implies that there exists a constant sequence of $\gamma_{2,n}'$ such that 
        $\EE\{(R(\tf_{2,k};\cD_n,\{{I_1},{I_{2}}\}) - \sR_2(\phi,\phi_s))^{1/\epsilon}\}\leq \gamma_{2,n}'$ and $\gamma_{2,n}'\geq \gamma_{1,n}^{1/\epsilon}$.
        Hence, we can simply pick $\gamma_{2,n}'=\gamma_{1,n}^{1/\epsilon}$.
        Then, by Markov's inequality, we have that for all $\eta'>0$
        \begin{align*}
            \PP(|R(\tf_{2,k};\cD_n,\{{I_1},{I_{2}}\}) - \sR_2(\phi,\phi_s)|^{1/\epsilon}> \eta' ) \leq \gamma_{2,n}'/\eta',
        \end{align*}
        or equivalently for all $\eta>1$,
        \begin{align*}
            \PP(|R(\tf_{2,k};\cD_n,\{{I_1},{I_{2}}\}) - \sR_2(\phi,\phi_s)|> \eta\gamma_{2,n} ) \leq 1/\eta^{1/\epsilon},
        \end{align*}
        where $\gamma_{2,n}= \gamma_{2,n}'^{\epsilon}$ and $\gamma_{2,n}=o(n^{-\epsilon})$.
        Thus, \Cref{cond:conv-risk-M12} is satisfied under proportional asymptotics.

        \noindent\underline{Part (4) Bounded variance proxy for CV estimates.}
        From results by \citet[Remark 2.19]{patil2022mitigating}, Assumption \ref{asm:rmt-feat} in random matrix theory implies $L_4-L_2$ norm equivalence, since the components of $\bZ$ are independent and have bounded kurtosis.
        Invoking Proposition 2.16 of \citet{patil2022mitigating}, there exists $\tau>0$ such that
        \begin{align*}
            \hsigma_{I} &\leq \tau \inf_{\bbeta\in\RR^p} (\|y_0-\bx_0^{\top}\bbeta\|_{L_2} + \|\bbeta-\hbeta_{\lambda}(\cD_I) \|_{\bSigma})^2.
        \end{align*}
        When $\bbeta = \bbeta_0$, $\|y_0-\bx_0^{\top}\bbeta_0\|_{L_2}=\sigma$ and $\|\bbeta_0-\hbeta_{\lambda}(\cD_I) \|_{\bSigma} \rightarrow \sR_1(\phi,\phi_s)$ with probability tending to one, which is continuous and bounded in $\phi\in [\psi,\psi^{-1}]$ and $\phi_s\in[\phi,\psi^{-1}]$.
        This implies that $\hsigma_I=\cO_p(1)$.
        Analogously, $\hsigma_{I_1\cup I_2}=\cO_p(1)$ and hence it holds for $\CEN=\MOM$.

        Combining the parts above finishes the proof.
    \end{proof}

\section{Helper concentration results}
\label{sec:appendix-concerntration}

\subsection{Size of the intersection of randomly sampled datasets}

In this section, we collect a helper result concerned with convergences that are used in the proofs of \Cref{thm:limiting-oobcv-for-arbitrary-M-cond}.
Before stating the lemma, we recall the definition of a hypergeometric random variable along with its mean and variance;
see \citet{greene2017exponential} for more related details.

    \begin{definition}[Hypergeometric random variable]
        A random variable $X$ follows the hypergeometric distribution $X\sim \operatorname {Hypergeometric} (n,K,N)$ if its probability mass function is given by
        \begin{align*}
            \PP(X=k)=\frac{
         \binom{K}{k}\binom{N-K}{n-k}
         }{\binom{N}{n}
         },\qquad \max\{0,n+K-N\}\leq k\leq \min\{n,K\}.
        \end{align*}
        The expectation and variance of $X$ are given by
        \begin{align*}
            \EE(X) &= \frac{nK}{N},
            \qquad
            \Var(X) = \frac{nK(N-K)(N-n)}{N^2(N-1)}.
        \end{align*}
    \end{definition}

    The following lemma characterizes the limiting proportions of shared observations in two simple random samples when both the subsample size $k$ and the full data size $n$ tend to infinity.
    \begin{lemma}[Asymptotic proportions of shared observations, adapted from \citet{patil2022bagging}]\label{lem:i0_mean}
        For $n\in\NN$, define $\cI_k = \{\{i_1, i_2, \ldots, i_k\}:\, 1\le i_1 < i_2 < \ldots < i_k \le n\}$.
        Let $I_1,I_2\overset{\textup{\texttt{SRSWR}}}{\sim}\cI_k$, define the random variable $i_{0}^{\textup{\texttt{SRSWR}}} =|I_1\cap I_2|$ to be the number of shared samples, and define $i_{0}^{\textup{\texttt{SRSWOR}}}$ accordingly.
        Then $i_0^{\textup{\texttt{SRSWR}}}\sim \textup{\text{Binomial}}(k,k/n)$ and $i_0^{\textup{\texttt{SRSWOR}}} \sim \operatorname {Hypergeometric} (k,k,n)$.
        Let $\{k_m\}_{m=1}^{\infty}$ and $\{n_m\}_{m=1}^{\infty}$ be two sequences of positive integers such that $n_m$ is strictly increasing in $m$, $n_m^{\nu}\leq k_m\leq n_m$ for some constant $\nu\in(0,1)$.
        Then, $i_0^{\textup{\texttt{SRSWR}}}/k_m-k_m/n_m\rightarrow 0$, and $i_0^{\textup{\texttt{SRSWOR}}}/k_m-k_m/n_m\rightarrow 0$ with probability tending to one.
    \end{lemma}

\section{Additional experimental details and results}\label{app:ex-result}

\subsection{Risk estimation and extrapolation in \Cref{subsec:experiment-risk-extrapolation} }\label{app:ex-risk-extrapolation}

In these experiments, we use $n=500,p=5,000$ ($\phi=0.1$) and $n=5,000,p=500$ ($\phi=10$) for underparameterized and overparameterized regimes, where the subsample aspect ratio $\phi_s$ varies from 0.1 to 10, and from 10 to 100, respectively.
The out-of-sample prediction errors are computed on $n_{\test}=2,000$ samples, and the results are averaged over 50 dataset repetitions.
The ECV cross-validation estimates \eqref{eq:Roob-M} are computed on $M_0=10$ base predictors.
For the kNN predictor, we use 5 nearest neighbors.
For the logistic predictor, we further binarize the response at the median with the null risk of a predictor always outputs 0.5 being 0.25.

\begin{table}[H]
    \caption{Summary of experimental results in \Cref{subsec:risk-extrapolation}.}
    \centering
    \begin{tabular}{lcccccc}
    \toprule
         \textbf{Predictors} &
        ridgeless & ridge &  lassoless & lasso & logistic & kNN\\
        \midrule
        Figure & \Cref{fig:est-ridgeless} & \Cref{fig:est-ridge} & \Cref{fig:est-lassoless} & \Cref{fig:est-lasso} & \Cref{fig:est-logistic} & \Cref{fig:est-kNN} \\
        \bottomrule
    \end{tabular}
\end{table}

\subsection{Tuning ensemble and subsample sizes in \Cref{subsec:ex-k-M}}
\label{app:ex-k-M}

In these experiments, we use $n=1,000$ and $p=\lfloor n\phi\rfloor$ with data aspect ratio $\phi$ varying from 0.1 to 10.
The \ECV is performed on a grid of subsample aspect ratios $\phi_s$ given in \Cref{alg:cross-validation}, with $M_0=10$ and $M_{\max}=50$.
The out-of-sample prediction errors are computed on $n_{\test}=2,000$ samples, and the results are averaged over 50 dataset repetitions.

\begin{table}[H]
    \caption{Summary of experimental results in \Cref{subsec:ex-k-M}.}
    \centering
    \begin{tabular}{lcccccc}
    \toprule
        \textbf{Procedure} & \multicolumn{3}{c}{bagging} & \multicolumn{3}{c}{subagging}\\\cmidrule(lr){2-4} \cmidrule(lr){5-7}
        \textbf{Model}& \ref{model:linear} & \ref{model:quad} & \ref{model:tanh}& \ref{model:linear} & \ref{model:quad} & \ref{model:tanh}\\
        \midrule
        Figure  & \Cref{fig:oobcv-linear-bagging} & \Cref{fig:oobcv-quad-bagging} & \Cref{fig:oobcv-tanh-bagging}& \Cref{fig:oobcv-linear-subagging} & \Cref{fig:oobcv-quad-subagging} & \Cref{fig:oobcv-tanh-subagging}\\
        \bottomrule
    \end{tabular}
\end{table}

\clearpage
\begin{figure}
    \centering
    \includegraphics[width=\textwidth]{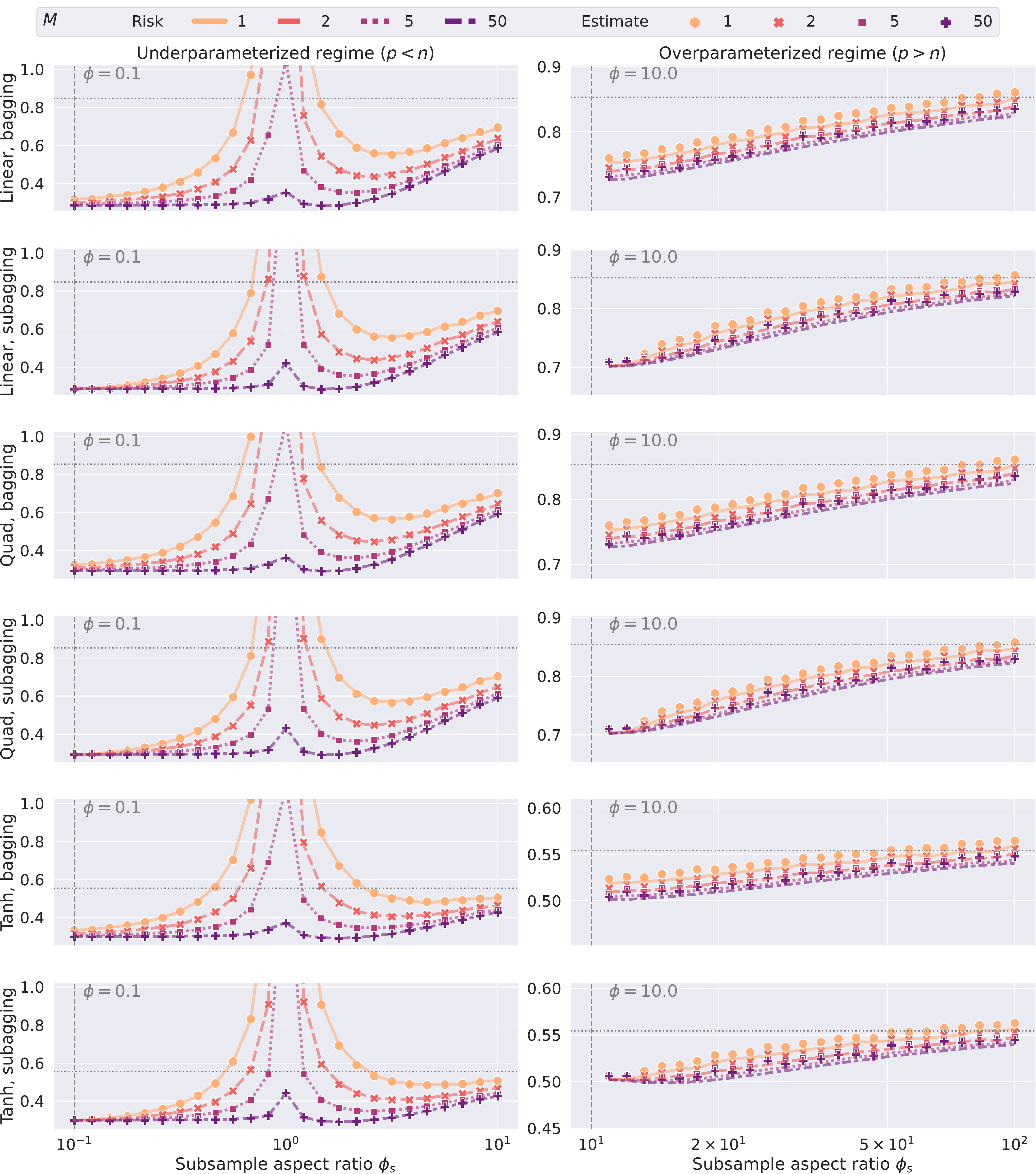}    
    \caption{
    Finite-sample ECV estimate (points) and prediction risk (lines) for ridge predictors ($\lambda=0.1$) using bagging and subagging, under nonlinear quadratic and tanh models (\Cref{subsec:experiment-risk-extrapolation}) for varying subsample size $k=\lfloor p/\phi_s\rfloor$, and the ensemble size $M$.
    For each value of $M$, the points denote the finite-sample ECV cross-validation estimates \eqref{eq:Roob-M}, and the lines denote the out-of-sample prediction error.
    See \Cref{app:ex-risk-extrapolation} for more details.}
    \label{fig:est-ridge}
\end{figure}

\begin{figure}
    \centering
    \includegraphics[width=\textwidth]{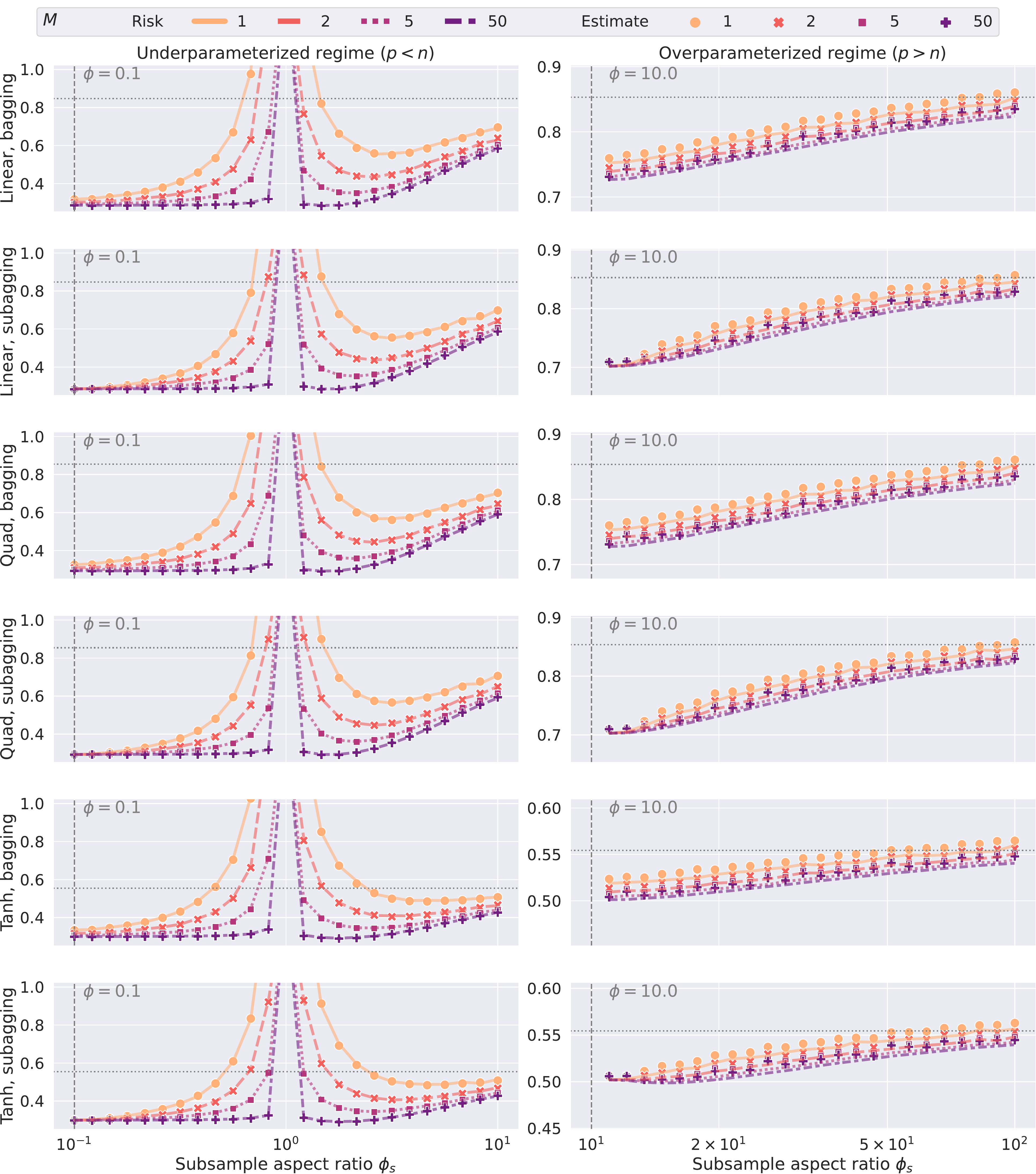}    
    \caption{
    Finite-sample ECV estimate (points) and prediction risk (lines) for ridgeless predictors using bagging and subagging, under nonlinear quadratic and tanh models (\Cref{subsec:experiment-risk-extrapolation}) for varying subsample size $k=\lfloor p/\phi_s\rfloor$, and the ensemble size $M$.
    For each value of $M$, the points denote the finite-sample ECV cross-validation estimates \eqref{eq:Roob-M}, and the lines denote the out-of-sample prediction error.
    See \Cref{app:ex-risk-extrapolation} for more details.}
    \label{fig:est-ridgeless}
\end{figure}

\begin{figure}
    \centering
    \includegraphics[width=\textwidth]{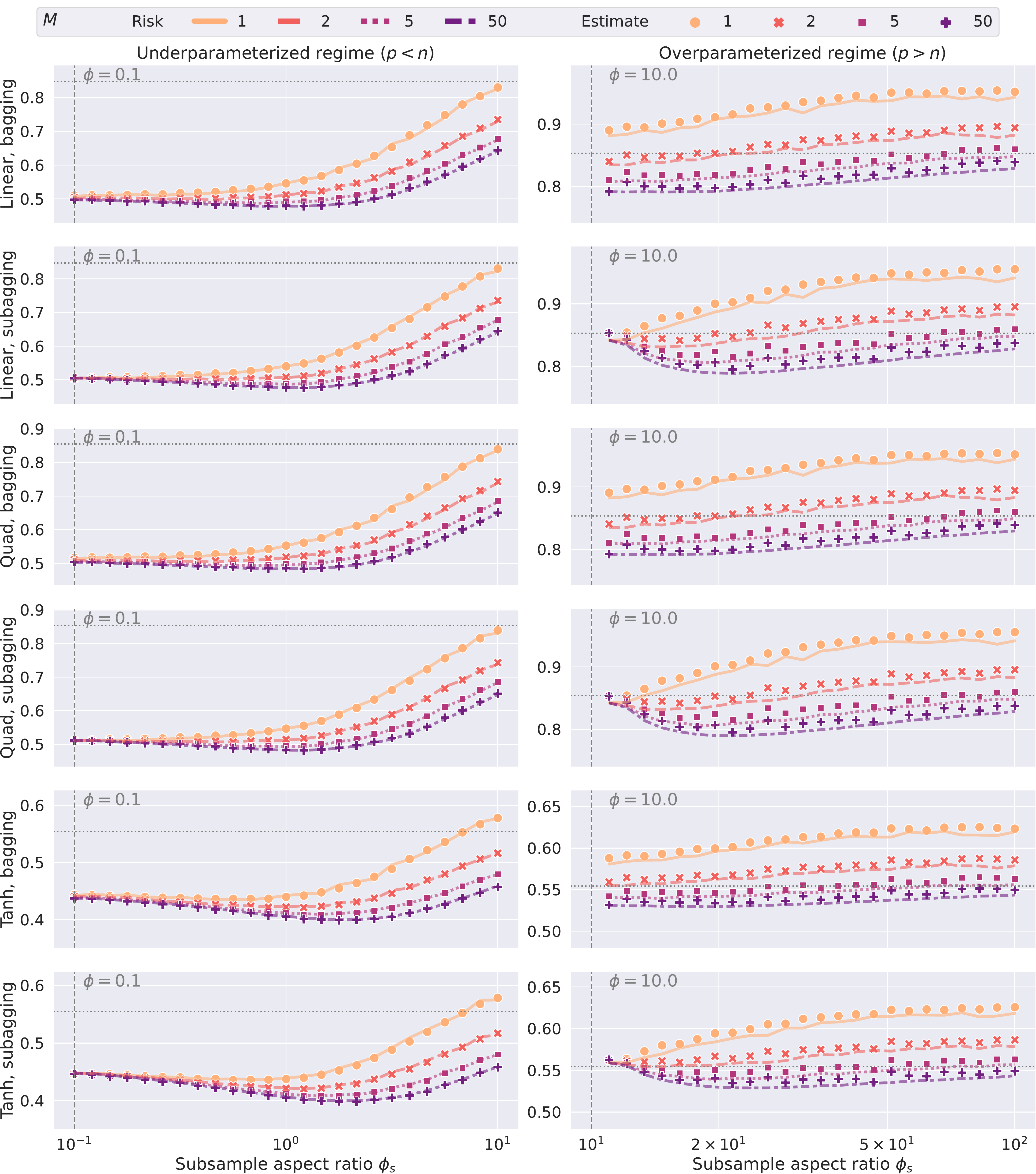}    
    \caption{
    Finite-sample ECV estimate (points) and prediction risk (lines) for lasso predictors ($\lambda=0.1$) using bagging and subagging, under nonlinear quadratic and tanh models (\Cref{subsec:experiment-risk-extrapolation}) for varying subsample size $k=\lfloor p/\phi_s\rfloor$, and the ensemble size $M$.
    For each value of $M$, the points denote the finite-sample ECV cross-validation estimates \eqref{eq:Roob-M}, and the lines denote the out-of-sample prediction error.
    See \Cref{app:ex-risk-extrapolation} for more details.}
    \label{fig:est-lasso}
\end{figure}

\begin{figure}
    \centering
    \includegraphics[width=\textwidth]{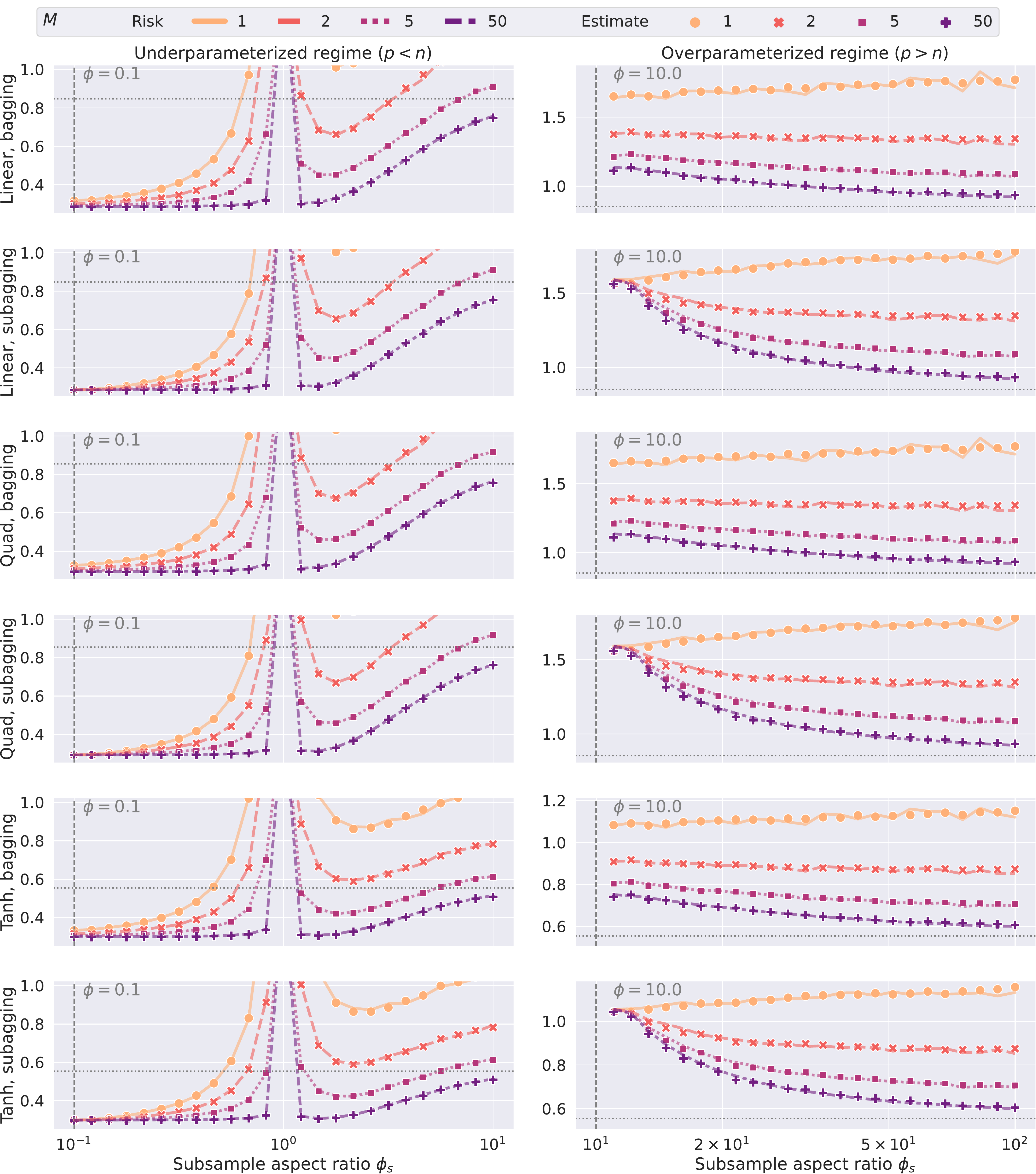}    
    \caption{
    Finite-sample ECV estimate (points) and prediction risk (lines) for lassoless predictors using bagging and subagging, under nonlinear quadratic and tanh models (\Cref{subsec:experiment-risk-extrapolation}) for varying subsample size $k=\lfloor p/\phi_s\rfloor$, and the ensemble size $M$.
    For each value of $M$, the points denote the finite-sample ECV cross-validation estimates \eqref{eq:Roob-M} computed on $M_0=10$ base predictors, and the lines denote the out-of-sample prediction error computed on $n_{\test}=2,000$ samples, averaged over 50 dataset repetitions, with $n=1,000$ and $p=\lfloor n\phi\rfloor$, and $\phi=0.1$ and $10$ for underparameterized ($p<n$) and overparameterized ($p>n$) regimes.}
    \label{fig:est-lassoless}
\end{figure}

\begin{figure}
    \centering
    \includegraphics[width=\textwidth]{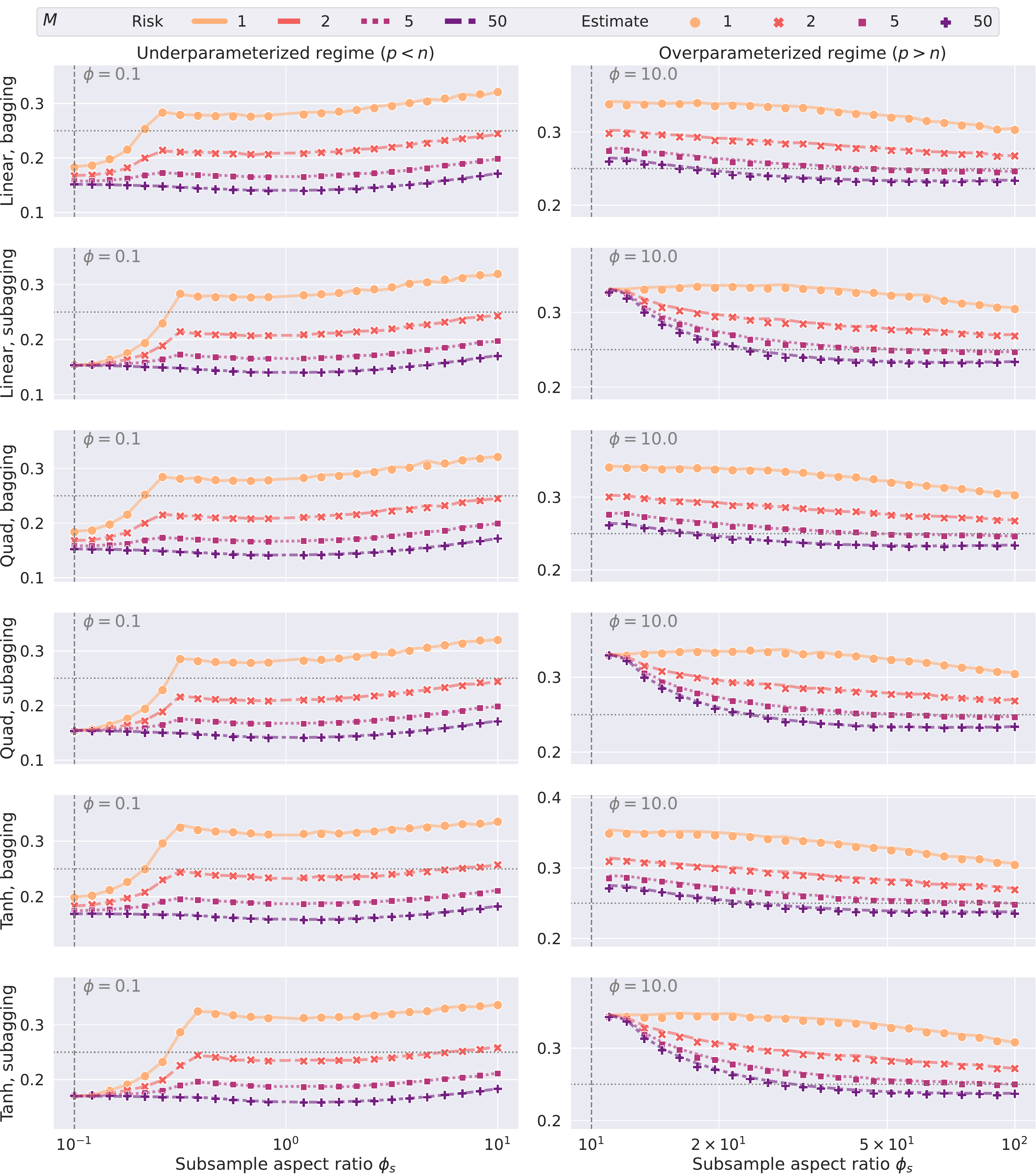}    
    \caption{
    Finite-sample ECV estimate (points) and prediction risk (lines) for logistic predictors using bagging and subagging, under nonlinear quadratic and tanh models (\Cref{subsec:experiment-risk-extrapolation}) for varying subsample size $k=\lfloor p/\phi_s\rfloor$, and the ensemble size $M$.
    For each value of $M$, the points denote the finite-sample ECV cross-validation estimates \eqref{eq:Roob-M}, and the lines denote the out-of-sample prediction error.
    See \Cref{app:ex-risk-extrapolation} for more details.}
    \label{fig:est-logistic}
\end{figure}

\begin{figure}
    \centering
    \includegraphics[width=\textwidth]{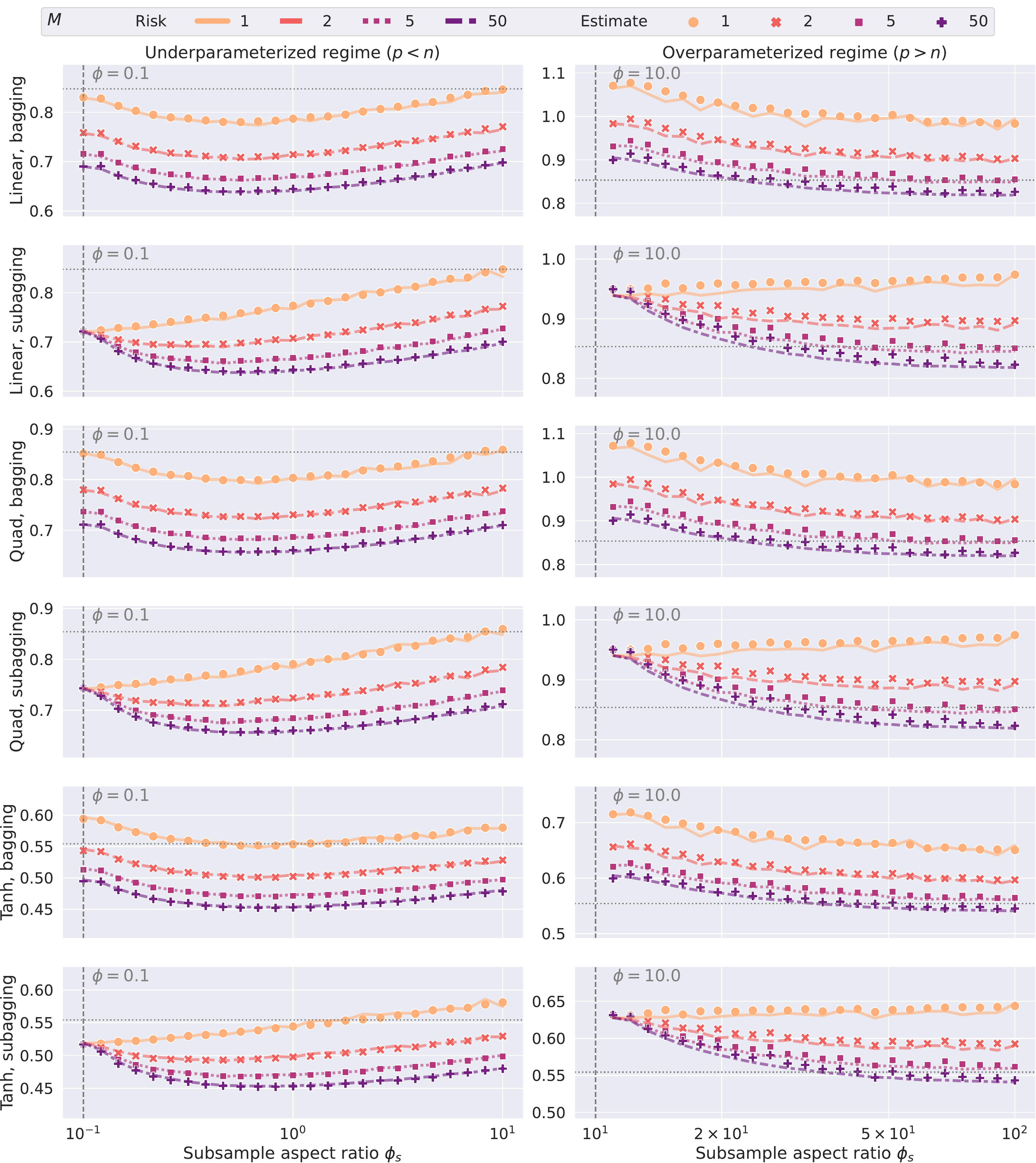}    
    \caption{
    Finite-sample ECV estimate (points) and prediction risk (lines) for kNN predictors using bagging and subagging, under nonlinear quadratic and tanh models (\Cref{subsec:experiment-risk-extrapolation}) for varying subsample size $k=\lfloor p/\phi_s\rfloor$, and the ensemble size $M$.
    For each value of $M$, the points denote the finite-sample ECV cross-validation estimates \eqref{eq:Roob-M}, and the lines denote the out-of-sample prediction error.
    See \Cref{app:ex-risk-extrapolation} for more details.}
    \label{fig:est-kNN}
\end{figure}

\clearpage
\begin{figure}[!t]
    \centering
    \includegraphics[width=0.8\textwidth]{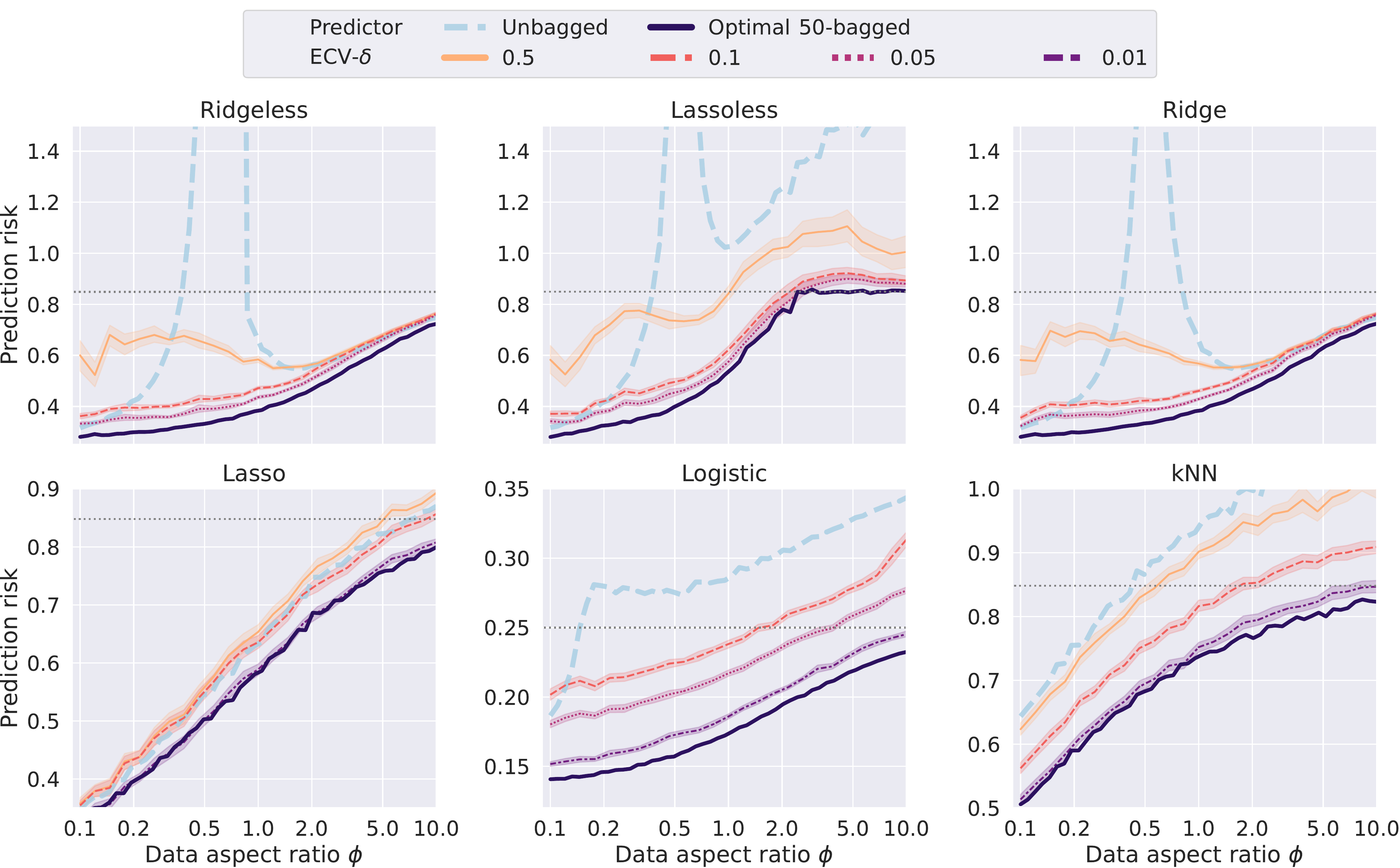}
    \caption{Prediction risk for different bagged predictors by \ECV, under model \ref{model:linear}, for varying $\phi$ and tolerance threshold $\delta$.
    An ensemble is fitted when \eqref{eq:tobag} is satisfied with $\zeta=5$.
    The null risks and the risks for the non-ensemble predictors are marked as gray dotted lines and blue dashed lines, respectively.
    The points denote finite-sample risks, and the shaded regions denote the values within one standard deviation.
    See \Cref{app:ex-k-M} for more details.}
    \label{fig:oobcv-linear-bagging}
\end{figure}

\begin{figure}[!t]
    \centering
    \includegraphics[width=0.8\textwidth]{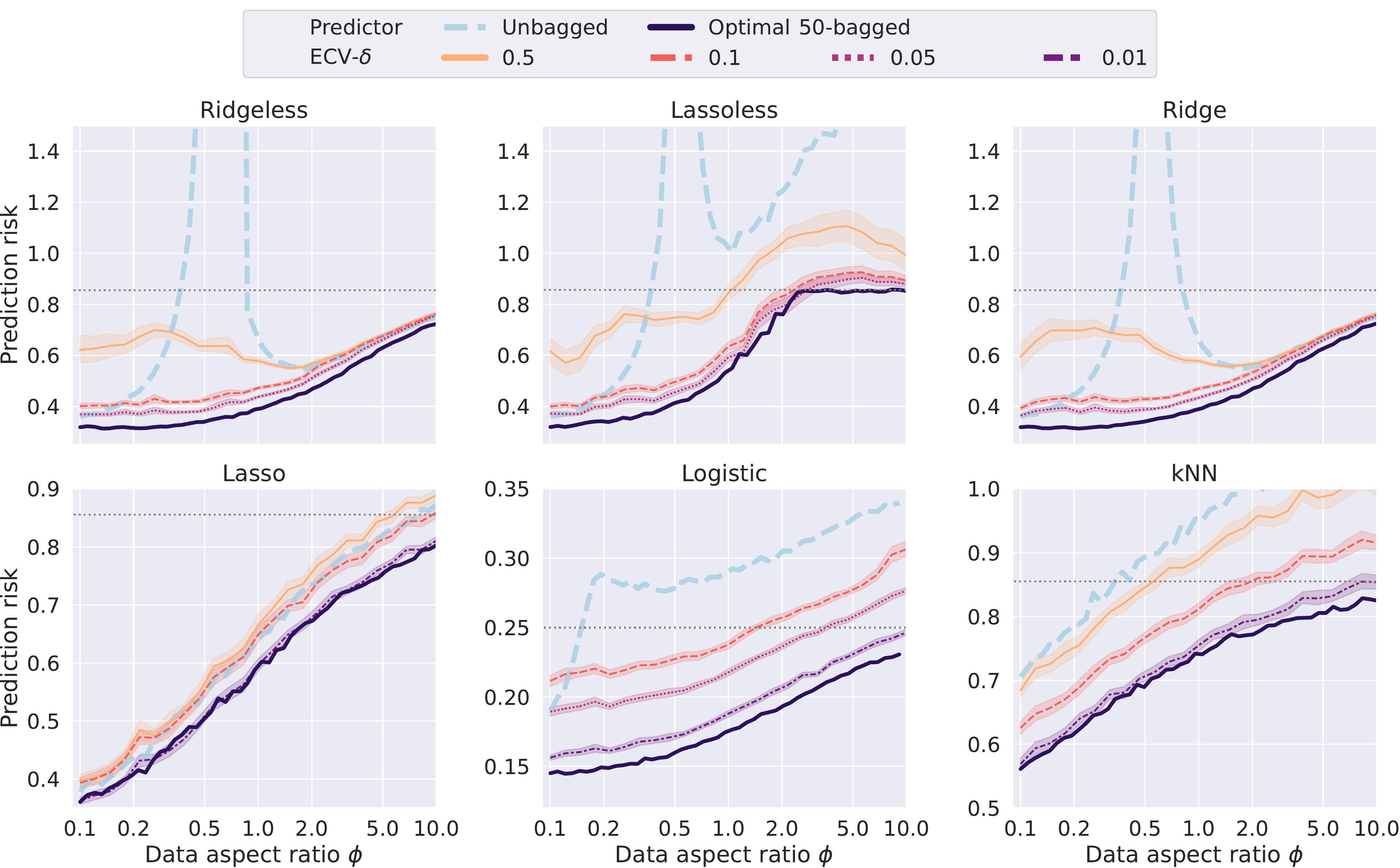}
    \caption{Prediction risk for different bagged predictors by \ECV, under model \ref{model:tanh}, for varying $\phi$ and tolerance threshold $\delta$.
    An ensemble is fitted when \eqref{eq:tobag} is satisfied with $\zeta=5$.
    The null risks and the risks for the non-ensemble predictors are marked as gray dotted lines and blue dashed lines, respectively.
    The points denote finite-sample risks, and the shaded regions denote the values within one standard deviation.
    See \Cref{app:ex-k-M} for more details.}
    \label{fig:oobcv-tanh-bagging}
\end{figure}

\begin{figure}[!t]
    \centering
    \includegraphics[width=0.8\textwidth]{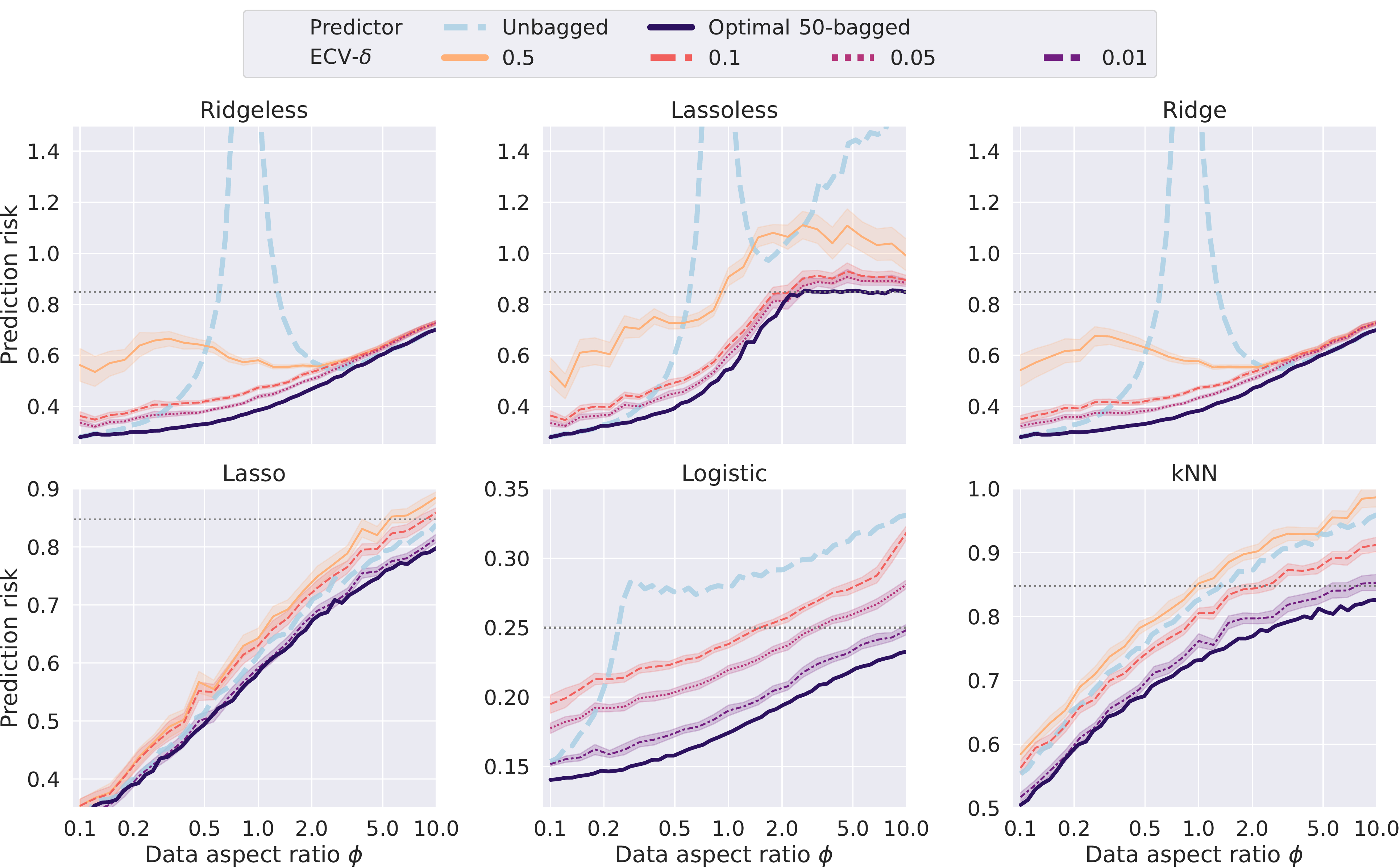}
    \caption{Prediction risk for different subagged predictors by \ECV, under model \ref{model:linear}, for varying $\phi$ and tolerance threshold $\delta$.
    An ensemble is fitted when \eqref{eq:tobag} is satisfied with $\zeta=5$.
    The null risks and the risks for the non-ensemble predictors are marked as gray dotted lines and blue dashed lines, respectively.
    The points denote finite-sample risks, and the shaded regions denote the values within one standard deviation.
    See \Cref{app:ex-k-M} for more details.}
    \label{fig:oobcv-linear-subagging}
\end{figure}

\begin{figure}[!t]
    \centering
    \includegraphics[width=0.8\textwidth]{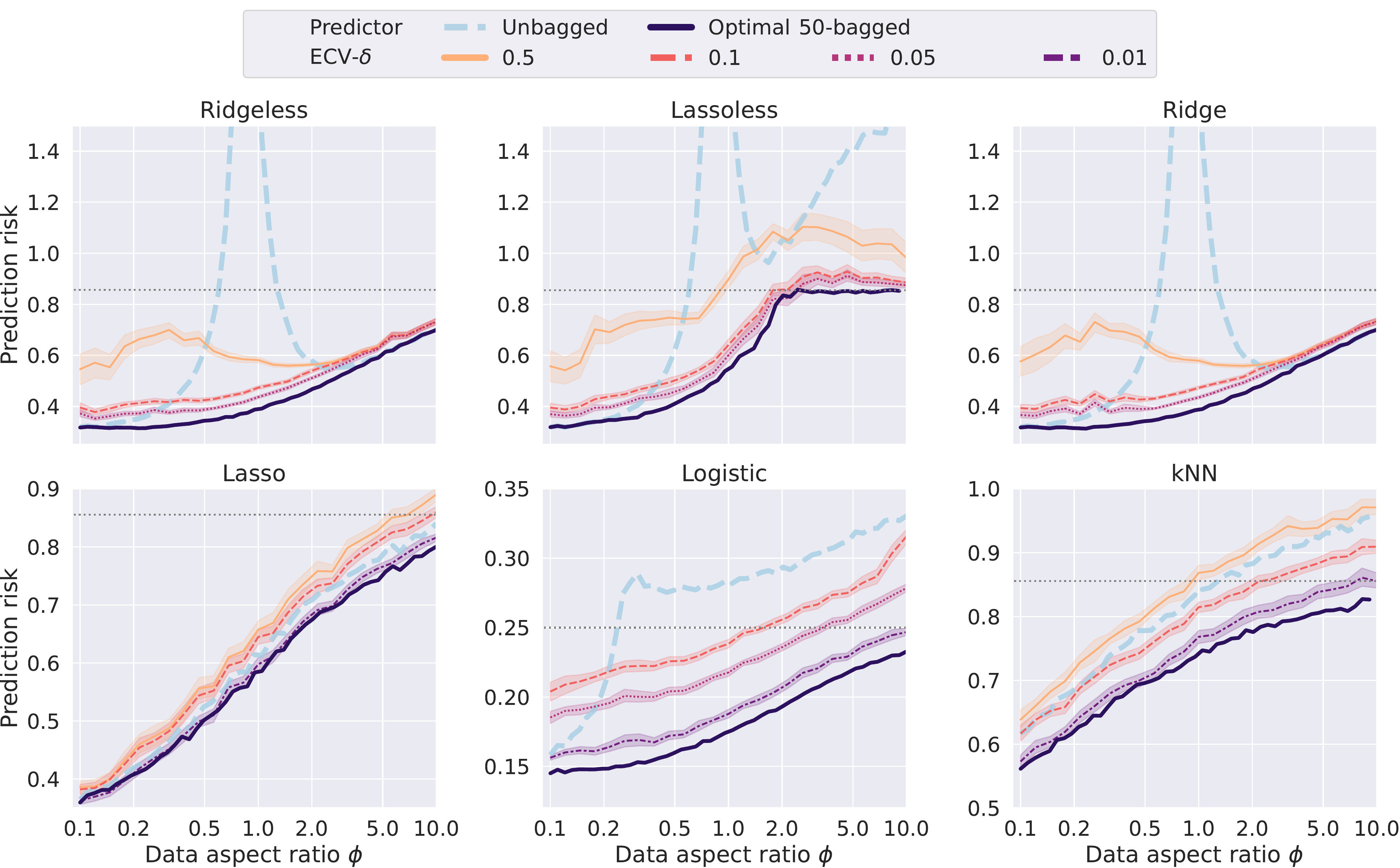}
    \caption{Prediction risk for different subagged predictors by \ECV, under model \ref{model:quad}, for varying $\phi$ and tolerance threshold $\delta$.
    An ensemble is fitted when \eqref{eq:tobag} is satisfied with $\zeta=5$.
    The null risks and the risks for the non-ensemble predictors are marked as gray dotted lines and blue dashed lines, respectively.
    The points denote finite-sample risks, and the shaded regions denote the values within one standard deviation.
    See \Cref{app:ex-k-M} for more details.}
    \label{fig:oobcv-quad-subagging}
\end{figure}

\begin{figure}[!t]
    \centering
    \includegraphics[width=0.8\textwidth]{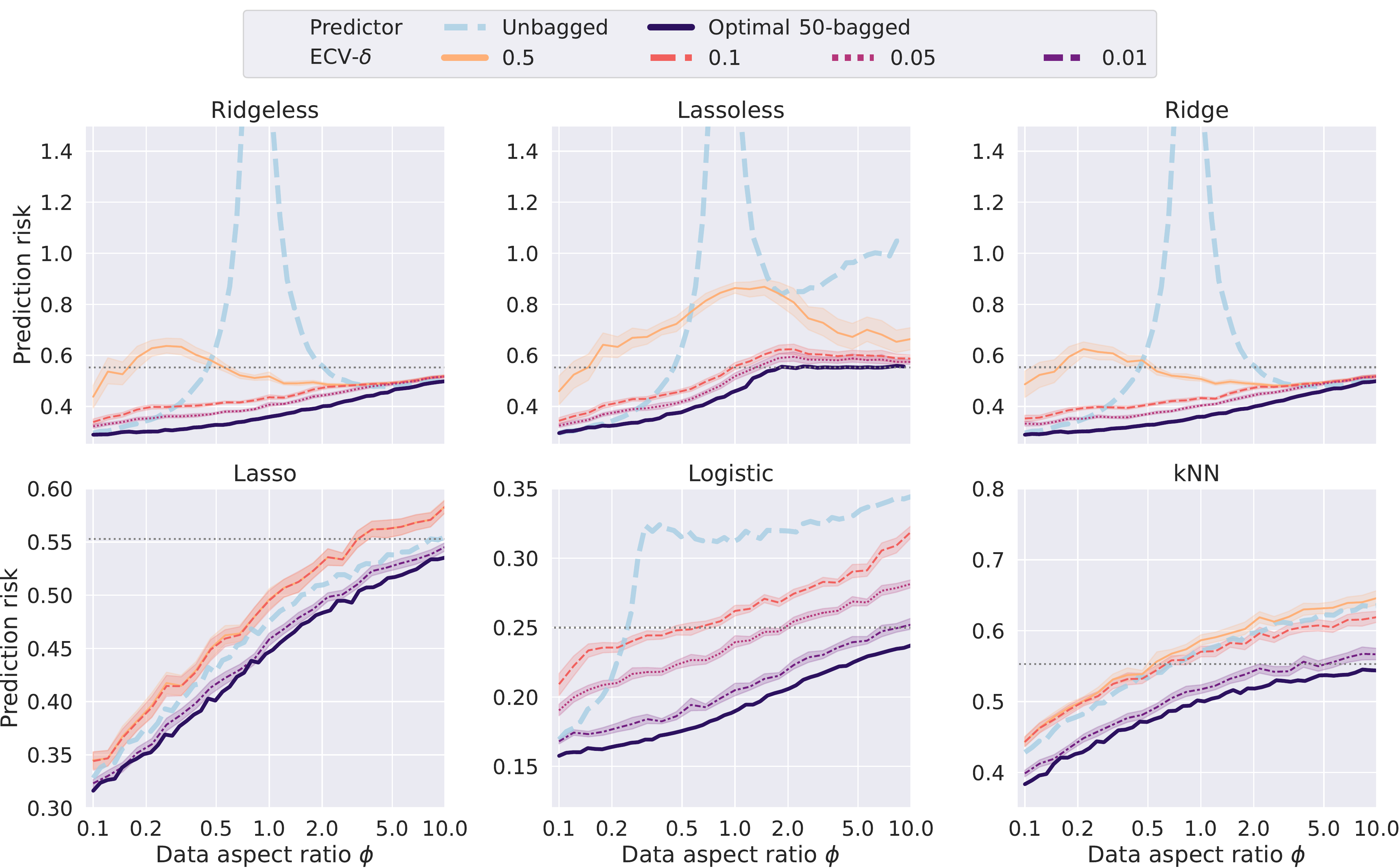}
    \caption{Prediction risk for different subagged predictors by \ECV, under model \ref{model:tanh}, for varying $\phi$ and tolerance threshold $\delta$.
    An ensemble is fitted when \eqref{eq:tobag} is satisfied with $\zeta=5$.
    The null risks and the risks for the non-ensemble predictors are marked as gray dotted lines and blue dashed lines, respectively.
    The points denote finite-sample risks, and the shaded regions denote the values within one standard deviation.
    See \Cref{app:ex-k-M} for more details.}
    \label{fig:oobcv-tanh-subagging}
\end{figure}

\clearpage
\subsection{Imbalanced classification}\label{app:subsec:imbalance}

    For classification tasks, $K$-fold cross-validation is recognized to have suboptimal performance in the case of imbalanced data. 
    In this subsection, we evaluate the performance of ECV and $K$-fold cross-validation under varying degrees of class imbalance.
    When $K$ is large, $K$-fold cross-validation also behaves similarly to the leave-one-out cross-validation.
    However, due to the increased computational complexity, we only examine $ K=3$, $5$, and $10$.
    We evaluate these methods using the following two metrics:
    \begin{itemize}[labelsep=1mm,leftmargin=7mm]
        \item Pointwise prediction error: the absolute error between the risk estimate and the true prediction risk of each ensemble predictor.
        
        \item Tuned prediction error: the absolute error between the prediction risk of the tuned ensemble predictor and the one of the optimal ensemble predictors fitted on all training data.
    \end{itemize}
    From \Cref{fig:clf}, we can see that when $M$ is small, the pointwise prediction error of \kfoldcv increases in both the number of fold $K$ and the class proportion.
    This indicates that \kfoldcv is unstable in the case of unbalanced data.    
    As $M$ increases, the prediction risk gets stabilized and all methods have similar performance.
    On the other hand, \ECV has stable prediction errors across different class proportions and much smaller computational complexity than \kfoldcv, as pointed out in \Cref{subsec:practice}.

    In terms of tuned prediction errors, when $M$ is small, the optimal risks are obtained at either the smallest subsample size in the overparameterized regime or the largest one in the overparameterized regime, as shown in \Cref{fig:est-logistic}.
    Thus, it is easier to achieve good predictive performance once the subsample size is close to the endpoint of the grid.
    As a result, we see that the tuned prediction errors of all methods are smaller than their pointwise prediction errors.
    However, \kfoldcv has increasing prediction error and variability as the class proportion increases, even when $M$ is large.

    Overall, the \ECV method is more accurate, robust, and efficient than the \kfoldcv method for binary classification tasks in the presence of class imbalance.
    
    \begin{figure}[!ht]
        \centering
        \includegraphics[width=0.8\textwidth]{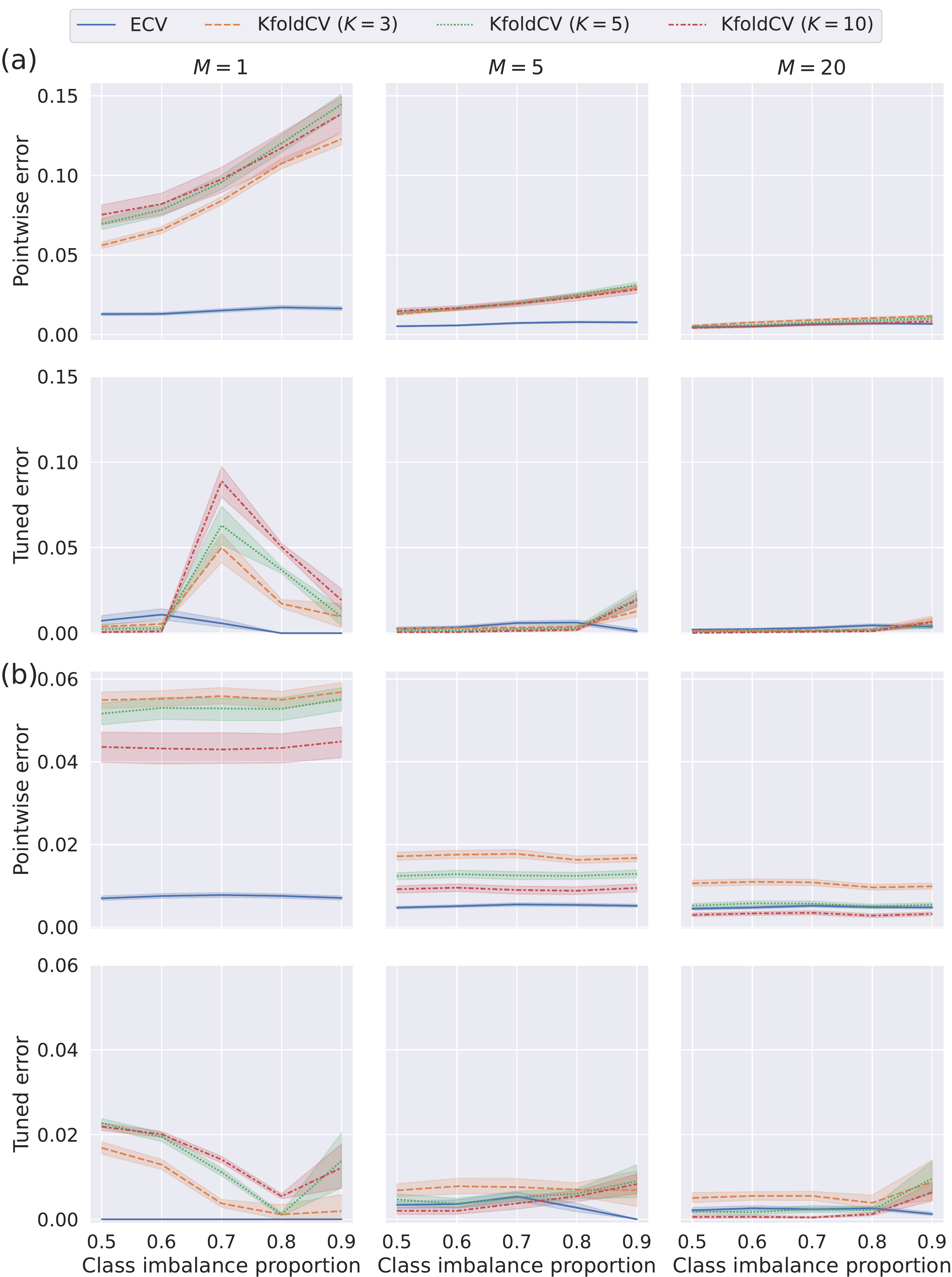}
        \caption{The pointwise and tuned errors for Logistic predictors with varying negative class proportions, in (a) the underparameterized regime ($\phi=0.1$) and (b) the overparameterized ($\phi=10$) regimes, with $n=1,000$ and $p=\lfloor n\phi\rfloor$ across 50 repetitions.
        The setup is under model \ref{model:quad} with $\rhoar=0.5$, where the response is further binarized by thresholding at the specified quantile from 0.5 to 0.9 on the training data.
        \ECV uses $M_0=10$ to extrapolate the estimates.
        }\label{fig:clf}
    \end{figure}

\clearpage
\subsection{Risk extrapolation in \Cref{subsec:tuning-random-forests}}
\label{app:tuning-random-forests}

\subsubsection[Sensitivity analysis of M0]{Sensitivity analysis of $M_0$}\label{app:subsubsec:M0}
    To inspect the performance of \ECV on different values $M_0$, the number of base predictors used for estimation, we conduct sensitivity analysis of $M_0$ in terms of relative pointwise error (defined as the ratio of pointwise error to the null risk) and time complexity.
    Because the scale of the (absolute) prediction errors may be quite different in the underparameterized and overparameterized regimes, we normalize it by the null risk so that it is more informative when comparing the two scenarios.
    The results are shown in \Cref{tab:M0_err}.
    We see that the relative error decreases as $M_0$ increases initially; however, it gets stable when $M_0 \geq 20$ in both the underparameterized and overparameterized scenarios.
    As anticipated, the errors are slightly greater in the overparameterized regime, reflecting the inherent difficulty of estimation in this scenario.
    On the other hand, the time complexity increases in $M_0$ sublinearly due to the parallel computation. 
    We also observe an increase in the time complexity in the overparameterized regime, which is unavoidable because of fitting deeper random forests.
    This indicates that one could choose a relatively large $M_0$, e.g. $M_0=20$ to obtain good accuracy with a small computation cost.

    One can also determine $M_0$ adaptively when building the ensembles in an online learning fashion.
    Since the risk estimate as a function $M_0$ converges to a certain limit, we can stop increasing $M_0$ when the risk estimate changes slowly enough.
    There are various criteria one may consider for early stopping.
    For instance, comparing the variance of risk estimate to a pre-specified threshold, as in \citet{lopes2019estimating,lopes2020measuring}.
    
\begin{table}[!ht]\centering
    \begin{tabular}{cccrrrrr}
    \toprule
    \multicolumn{3}{c}{\multirow{2}{*}{$\phi$}} & \multicolumn{5}{c}{$M_0$} \\\cmidrule(lr){4-8}
    & & & 5 & 10 & 15 & 20 & 25 \\
    \midrule
    \multirow{4}{*}{Pointwise error} & \multirow{2}{*}{0.1} & mean & 0.0874 & 0.0731 & 0.0668 & 0.0638 & 0.0661 \\
    & & sd & 0.0691 & 0.0570 & 0.0628 & 0.0881 & 0.0696 \\\cmidrule(lr){3-8}
     & \multirow{2}{*}{10} & mean & 0.1133 & 0.1311 & 0.1156 & 0.1058 & 0.0945 \\
    & & sd & 0.1090 & 0.1272 & 0.0882 & 0.0939 & 0.0812 \\\cmidrule(lr){2-8}
    \multirow{4}{*}{Tuned error} & \multirow{2}{*}{0.1} & mean & 0.0002 & 0.0006 & 0.0004 & 0.0004 & 0.0001 \\
     &  & sd & 0.0006 & 0.0012 & 0.0011 & 0.0011 & 0.0004 \\
    \cmidrule(lr){3-8}
     & \multirow{2}{*}{10} & mean & 0.0002 & 0.0003 & 0.0002 & 0.0002 & 0.0002 \\
     &  & sd & 0.0006 & 0.0011 & 0.0007 & 0.0007 & 0.0005 \\\cmidrule(lr){2-8}
    \multirow{4}{*}{Time} & \multirow{2}{*}{0.1} & mean & 0.0718 & 0.0870 & 0.1067 & 0.1388 & 0.1728 \\
    & & sd & 0.0295 & 0.0346 & 0.0292 & 0.0381 & 0.0301 \\\cmidrule(lr){3-8}
     & \multirow{2}{*}{10} & mean & 0.9249 & 1.2543 & 1.4492 & 1.9112 & 1.9652 \\
    & & sd & 0.2702 & 0.3360 & 0.3746 & 0.4383 & 0.4035 \\
    \bottomrule
    \end{tabular}
    \caption{The effect of $M_0$ on the relative prediction errors and computational time for random forests under model \ref{model:quad} with $n = 500$ across 50 repetitions.}
    \label{tab:M0_err}
\end{table}

\subsubsection[Sensitivity analysis of rhoar1]{Sensitivity analysis of $\rhoar$}\label{app:subsubsec:rho}
To inspect the impact of $\rhoar$ related to feature correlation, we conduct sensitivity analysis of $\rhoar$ and the results of the relative prediction error (defined as the ratio of prediction error to the null risk) are shown in \Cref{tab:rho}.
When the magnitude of $\rhoar$ is small, the variability in the pointwise prediction errors of \ECV tends to increase.
This phenomenon occurs because when $\rhoar$ approaches zero, the features exhibit almost complete independence. In such cases, random forests struggle to leverage collinearity for mitigating prediction risk.
Overall, the relative error lies between 0.05 and 0.15 in various settings.

\begin{table}[!ht]\centering
    \begin{tabular}{cccrrrrrrr}
    \toprule
    & \multirow{2}{*}{$\phi$} & & \multicolumn{7}{c}{$\rhoar$} \\\cmidrule(lr){4-10}
    & & & -0.75 & -0.5 & -0.25 & 0 & 0.25 & 0.50 & 0.75 \\
    \midrule
    \multirow{4}{*}{Pointwise error} & \multirow{2}{*}{0.1} & mean & 0.0663 & 0.0822  & 0.1119 & 0.1158 & 0.1053 & 0.0916 & 0.0552 \\
    && sd & 0.0471 & 0.0803 & 0.0803 & 0.0948 & 0.0819 & 0.0816 & 0.0714 \\\cmidrule(lr){3-10}
    &\multirow{2}{*}{10} & mean & 0.1322 & 0.1348 & 0.1524 & 0.1277 & 0.1348 & 0.1143 & 0.1154 \\
    && sd & 0.1080 & 0.0927 & 0.1145 & 0.1102 & 0.1145 & 0.1107 & 0.1004 \\\cmidrule(lr){2-10}
    \multirow{4}{*}{Tuned error} & \multirow{2}{*}{0.1} & mean & 0.0004 & 0.0002 & 0.0007 & 0.0006 & 0.0003 & 0.0003 & 0.0004 \\
    && sd & 0.0008 & 0.0007 & 0.0018 & 0.0014 & 0.0009 & 0.0008 & 0.0007 \\\cmidrule(lr){3-10}
    &\multirow{2}{*}{10} & mean & 0.0004 & 0.0004 & 0.0010 & 0.0004 & 0.0005 & 0.0007 & 0.0004 \\
    && sd & 0.0009 & 0.0009 & 0.0022 & 0.0011 & 0.0013 & 0.0016 & 0.0008 \\
    \bottomrule
    \end{tabular}
    \caption{The effect of $\rhoar$ on the relative prediction errors for random forests under model \ref{model:quad} with $n = 500$ and $M_0=20$ with varying values of $\rhoar$ across 50 repetitions.}
    \label{tab:rho}
\end{table}

\subsubsection{Sensitivity analysis of covariance structures}\label{app:subsubsec:cov}
    To inspect the performance of \ECV on more complex covariance structures, we test it on the following three block-diagonal or graph-based correlation structures under data model \ref{model:quad}:
    \begin{itemize}[labelsep=1mm,leftmargin=7mm]
        \item A usual AR1 covariance matrix $\bSigma_{0}$.
        \item A block-diagonal covariance matrix with 2 blocks $(\bSigma_{0}^{(1)},\bSigma_{0.5}^{(2)})$ and $\bSigma_{0}^{(1)},\bSigma_{0.5}^{(2)}\in\RR^{p/2}$.
        \item A block-diagonal covariance matrix with 3 blocks $(\bSigma_{0}^{(1)},\bSigma_{0.5}^{(2)},\bSigma_{-0.5}^{(3)})$ and $\bSigma_{0}^{(1)},\bSigma_{0.5}^{(2)},\bSigma_{-0.5}^{(3)}\in\RR^{p/3}$.
    \end{itemize}
    The results of the relative prediction error (defined as the ratio of prediction error to the null risk) are shown in \Cref{tab:rho}.
    In these more complex situations, the relative pointwise prediction errors are also between 0.05 and 0.15, as we observed in \Cref{app:subsubsec:rho}.
    Thus, we conclude that \ECV is robust across different feature correlation structures as tested in the current experiment.

    \begin{table}[!ht]\centering
        \begin{tabular}{cccrrr}
        \toprule
        & \multirow{2}{*}{$\phi$} & & \multicolumn{3}{c}{Type of block structure} \\\cmidrule(lr){4-6}
        & & & 1 & 2 & 3 \\
        \midrule
        \multirow{4}{*}{Pointwise error} & \multirow{2}{*}{0.1} & mean & 0.1158 & 0.0827 & 0.0725 \\
        & & sd & 0.0729 & 0.0513 & 0.0736 \\\cmidrule(lr){3-6}
        & \multirow{2}{*}{10} & mean & 0.1277 & 0.1287 & 0.1497 \\
        & & sd & 0.1190 & 0.0988 & 0.1003 \\\cmidrule(lr){2-6}
        \multirow{4}{*}{Tuned error} & \multirow{2}{*}{0.1} & mean & 0.0006 & 0.0004 & 0.0004 \\
         &  & sd & 0.0014 & 0.0009 & 0.0009 \\\cmidrule(lr){3-6}
         & \multirow{2}{*}{10} & mean & 0.0004 & 0.0005 & 0.0004 \\
         &  & sd & 0.0011 & 0.0010 & 0.0011 \\
        \bottomrule
        \end{tabular}
        \caption{The effect of covariance structure on the relative prediction errors under model \ref{model:quad} with $n = 500$ and $M_0=20$ across 50 repetitions.}
        \label{tab:cov}
    \end{table}

    \subsubsection{Tuning feature subsampling ratios}\label{app:subsubsec:mtry}

    In the paper, we describe how to tune the ensemble and subsample sizes.
    In practice, there are also other hyperparameters that people would like to tune.
    For instance, the faction of features to consider when looking for the best split of random forests, related to parameter \texttt{max\_features} in Python package \texttt{sckit-learn} \citep{scikit-learn} or parameter \texttt{mtry} in R package \texttt{randomForest} \citep{liaw2002classification}.
    In practice, there are a few common heuristics for choosing a value for \texttt{mtry}. 
    For example, \texttt{mtry} can be set as $\sqrt{p}$, $p/3$, or $\log_2(p)$, where $p$ is the total number of features.
    In the regression setting, $\texttt{mtry}=p/3$ is suggested by \citet{breiman2001random}, while $\texttt{mtry}=p$ was more recently justified empirically in \citet{geurts2006extremely}.
    These heuristics are a good place to start when determining what value to use for \texttt{mtry}, before doing a grid search on \texttt{mtry}.
    Below, we illustrate the utility of ECV in tuning factions of features for random forests.

    We consider data model \ref{model:quad} with $\sigma=5$, $n=100$ and $p=1,000$ (such that data aspect ratio is $\phi=10$).
    We set the maximum ensemble size $M_{\max}$ to be 100, which is the default in \texttt{sckit-learn}, and set $M_0=25$.
    We first fix the subsample size as $k=90$, and tune over factions of features $r\in\cR=\{0.1,0.2,\ldots,0.9,1\}$.
    After $\hat{r}$ is obtained based on the ECV estimates, we further tune $k\in\cK=\{10,20,\ldots,90\}$ by fixing the factions of features as $\hat{r}$.

    The results are shown in \Cref{fig:mtry}.
    For tuning factions of features, we see a huge gain of time-saving from extrapolation of the risk estimate of the ensemble size $M$.
    This suggests that ECV is also beneficial in tuning other parameters other than the subsample size. 
    On the other hand, even after fractions of features are tuned, we observed improvement with further tuning of the subsample size.
    When the ensemble size is $M=10$, the improvement is significant. 
    When $M$ is large, we still see consistent slight improvement compared to the one with only fractions of features tuned.

    \begin{figure}
        \centering
        \includegraphics[width=\textwidth]{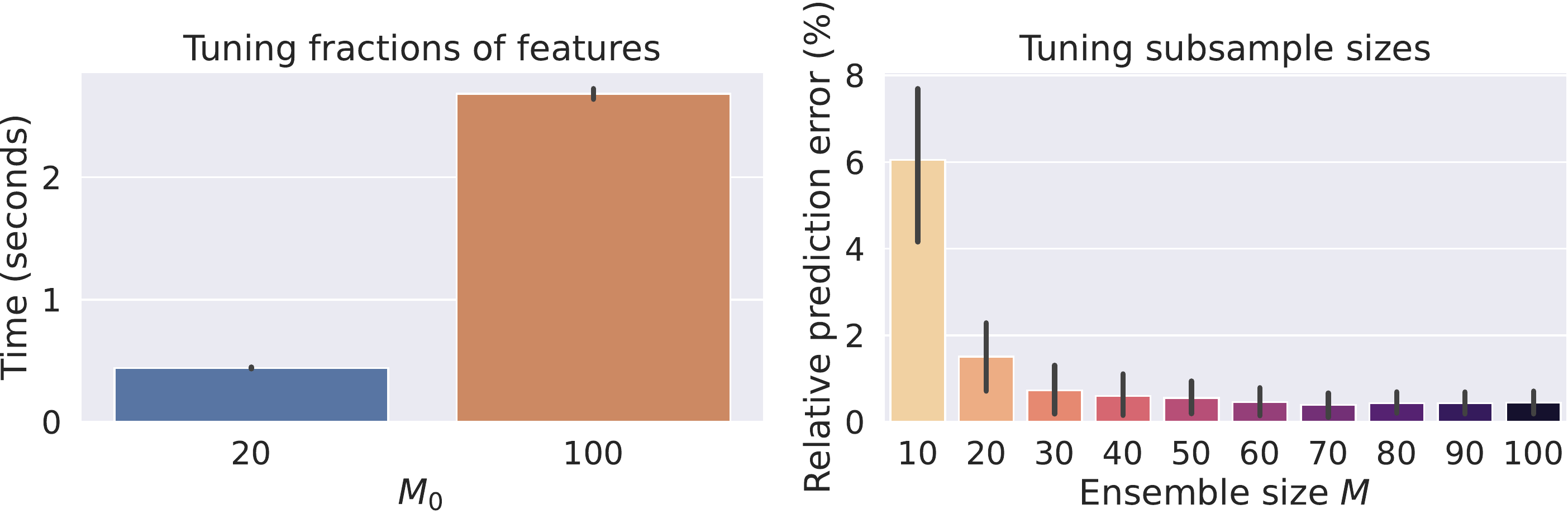}
        \caption{Tuning factions of features, subsample size by ECV.
        The left panel shows that computational time with extrapolation ($M_0=20$) and without extrapolation ($M_0=100$).
        The right panel shows the relative prediction error (the difference between the prediction errors without and with subsample size tunings, normalized by the null risk) in ensemble size $M$.}
        \label{fig:mtry}
    \end{figure}

\subsection{Risk profile in \Cref{sec:app-sc}}
    \label{app:sc}

    \citet{hao2021integrated} collected samples of 50,781 human peripheral blood mononuclear cells (PBMCs) originating from eight volunteers post-vaccination (day 3) of an HIV vaccine.
    The readers can follow the Seurat tutorial: \url{https://satijalab.org/seurat/articles/multimodal_reference_mapping.html} to download the preprocessed dataset.
    This single-cell CITE-seq dataset simultaneously measures 20,729 genes and 228 proteins in individual cells.
    As most genes are not variable across the dataset, we further subset the top 5,000 highly variable genes that exhibit high cell-to-cell variation in the dataset and the top 50 highly abundant surface proteins.
    Then, the genes and proteins are size-normalized to have total counts $10^4$ and log-normalized.
    In Figure \ref{fig:sc_overview}, we visualize the low-dimensional cell embeddings as well as the histograms of gene expressions and protein abundances in the Mono cell type.
    Overall, the gene expressions are extraordinarily sparse and zero-inflated distributed.
    On the contrary, the distribution of protein abundances is close to normal distributions.    
    \begin{figure}[!ht]
        \centering            \includegraphics[width=\textwidth]{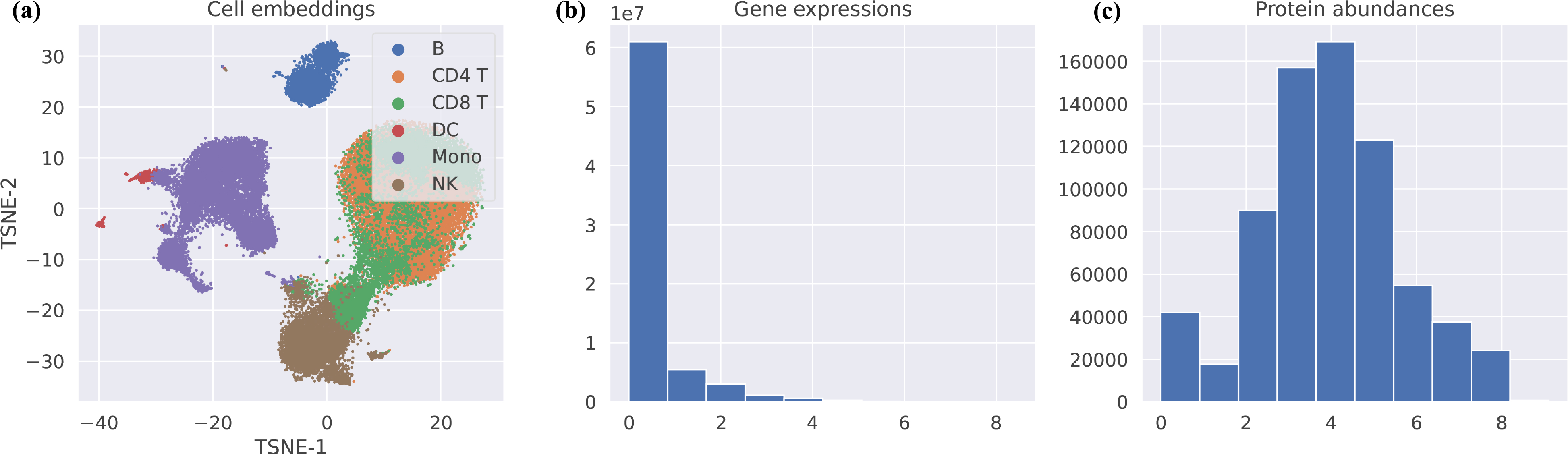}        
        \caption{Overview of the single-cell CITE-seq dataset from \citet{hao2021integrated}.
        {(a)} The 2-dimensional tSNE cell embeddings.
        {(b-c)} The histogram of overall log-normalized gene expressions and protein abundances in the Mono cell type.}
        \label{fig:sc_overview}
    \end{figure}
    
    \begin{table}[!ht]
        \caption{Description of different cell types in the CITE-sep dataset \citep{hao2021integrated}.
      The samples of PBMCs originate from eight volunteers post-vaccination (day 3) of an HIV vaccine.}\label{tab:dataset}
        \centering
        \begin{tabular}{cccc}
            \toprule
            \textbf{Cell type} & \textbf{Training size $n$} & \textbf{Test size $n_{\test}$} & \textbf{Data aspect ratio $p/n$}\\
            \midrule
            DC & 516 & 515 & 9.71\\[0.2em]
            B & 2279 & 2279 & 2.19\\[0.2em]
            NK & 3152 & 3152 & 1.59\\[0.2em]
            Mono & 7156 & 7156 & 0.70\\
            \bottomrule
        \end{tabular}
    \end{table}
    
    As there is the most outstanding level of heterogeneity within T cell subsets, we only consider the non-T cell types for the experiments and randomly split cells in each cell type into the training and test sets with equal probability.
    Different cell types consist of different numbers of cells and have various data aspect ratios.
    As summarized in \Cref{tab:dataset}, the four cell types cover both low-dimensional $(n>p)$ and high-dimensional $(n<p)$ datasets.
    Because the gene expressions exhibit high dimensionality, sparsity, and heterogeneity, predicting the protein abundances based on the transcriptome is thus a challenging problem.

    \begin{figure}[!t]
        \centering
            \includegraphics[width=0.8\textwidth]{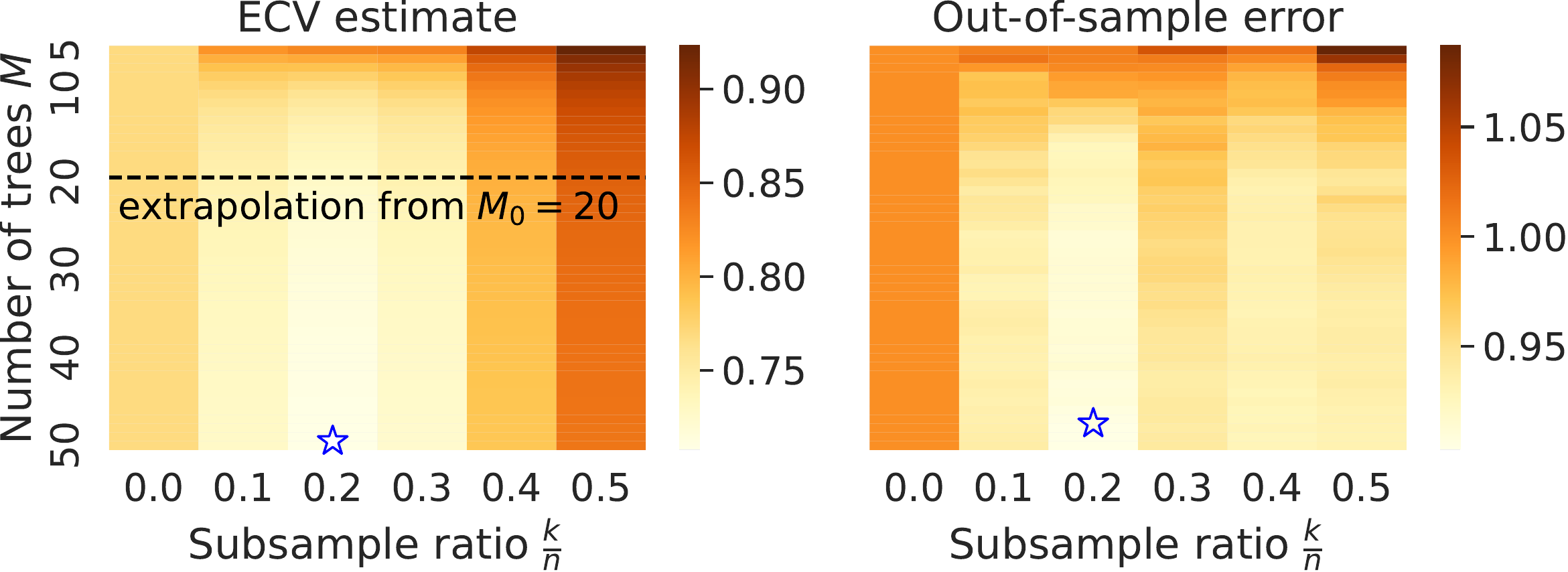}
        \caption{Heatmap of \ECV performance on predicting the abundances of surface protein CD103 in the DC cell type by random forests (subagging).
        The left and right panels show the NMSEs of OOB risk estimates and out-of-sample prediction risk, respectively.
        The values are normalized by the empirical variance of the response estimated from the test set; the darker, the larger value of NMSE.
        The extrapolated risk estimates are based on $M_0=20$ trees.}
        \label{fig:heatmap_5}
    \end{figure}

    \begin{figure}[H]
        \centering
            \includegraphics[width=0.9\textwidth]{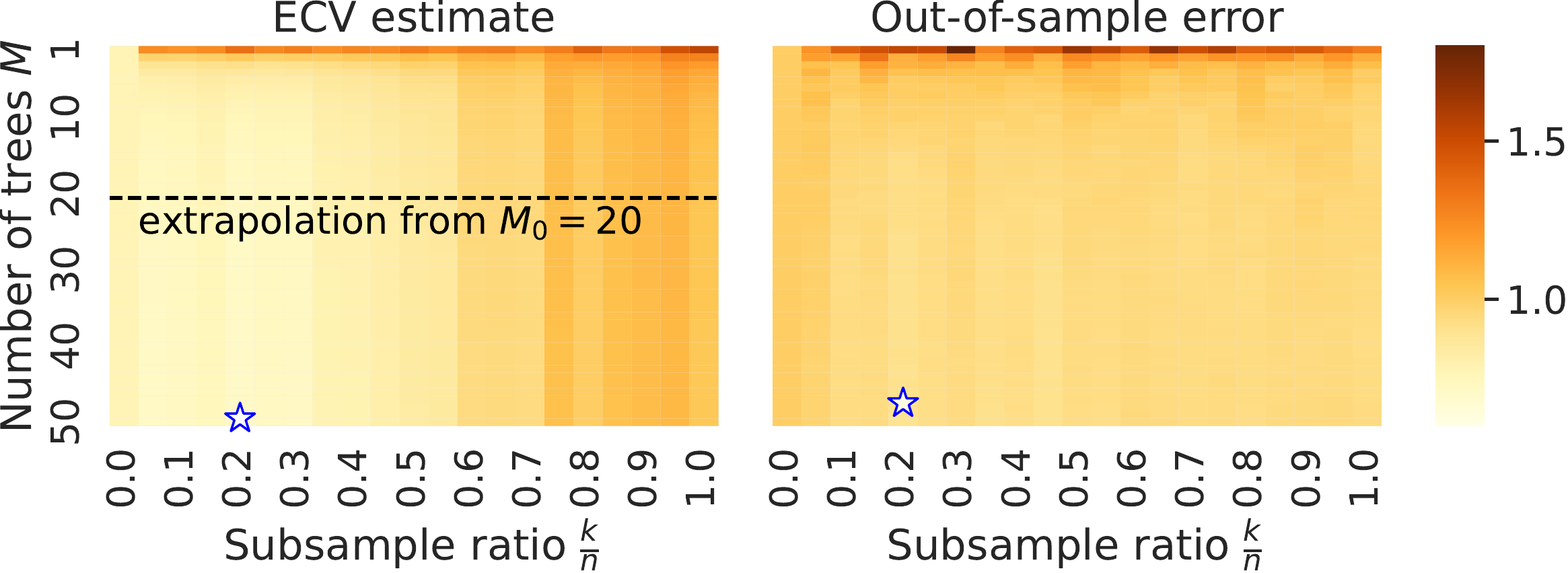}
        \caption{Heatmap of \ECV performance on predicting the abundances of surface protein CD103 in the DC cell type by random forests (subagging) as in \cref{fig:heatmap_5} with $M\in\{1,\ldots, 50\}$ and $k/n$ ranging from 0 to 1 displayed.}
        \label{fig:heatmap}
    \end{figure}

    \begin{figure}[H]
        \centering
            \includegraphics[width=\textwidth]{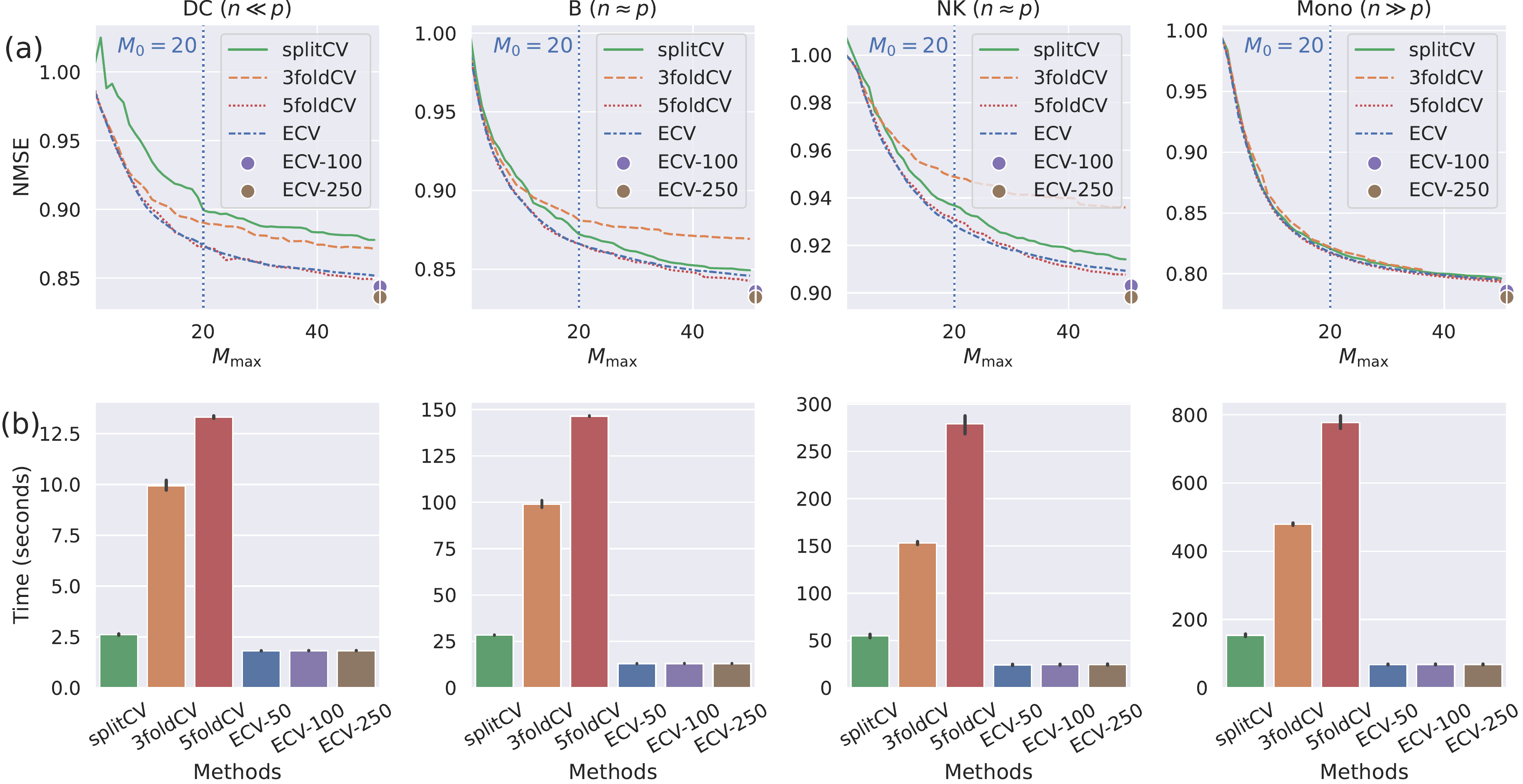}
        \caption{Performance of cross-validation methods on predicting the protein abundances in different cell types. 
        (a) The average normalized mean squared error (NMSE) of the cross-validated predictors for different methods. \ECV (bagging) uses $M_0=20$ trees to extrapolate risk estimates, and the dashed lines represent the performance using $M_{\max}\in\{100,250\}$.
        (b) The average time consumption for cross-validation in seconds.
        }
    \end{figure}

\end{bibunit}

\end{document}